\documentclass[a4paper,11pt]{article}
\usepackage[utf8]{inputenc}
\usepackage{fullpage}

\usepackage{amssymb,amsmath,amsthm}
\usepackage{stmaryrd}
\usepackage{wasysym}

\usepackage{enumerate}
\usepackage{xspace}
\usepackage{tabularx}
\usepackage{footnote}

\usepackage{algorithm2e}

\newtheorem{theorem}{Theorem}
\newtheorem{lemma}{Lemma}[section]
\newtheorem{claim}[lemma]{Claim}

\newtheorem{proposition}[lemma]{Proposition}

\theoremstyle{definition}
\newtheorem{definition}[lemma]{Definition}

\theoremstyle{remark}

\newcommand{\depth}{\alpha}

\newcommand{\findTD}{\mathtt{FindTD}}
\newcommand{\findPTD}{\mathtt{FindPartialTD}}
\newcommand{\whatsep}{\mathtt{whatsep}}

\newcommand{\compress}[1]{\mathtt{Compress}_{#1}}
\newcommand{\alg}[1]{\mathtt{Alg}_{#1}}

\newcommand{\build}{\mathtt{build}}
\newcommand{\ds}{\mathcal{DS}}

\newcommand{\oldS}{\mathrm{old_S}}
\newcommand{\oldPin}{\mathrm{old_\pi}}
\newcommand{\sep}{\mathtt{sep}}
\newcommand{\children}{\mathrm{children}}
\newcommand{\pins}{\mathtt{pins}}
\newcommand{\bags}{\mathtt{bags}}

\newcommand{\qnei}{\textnormal{findNeighborhood}}
\newcommand{\qUsep}{\textnormal{findUSeparator}}
\newcommand{\qSsep}{\textnormal{findSSeparator}}
\newcommand{\qpin}{\textnormal{findNextPin}}

\newcommand{\ext}{\textrm{ext}}
\newcommand{\true}{\top}
\newcommand{\false}{\bot}
\newcommand{\pin}{\pi}

\newcommand{\getS}{\mathrm{get}_S}
\newcommand{\getPin}{\mathrm{get}_{\pi}}
\newcommand{\insertF}{\mathrm{insert}_{F}}
\newcommand{\insertS}{\mathrm{insert}_{S}}
\newcommand{\insertX}{\mathrm{insert}_{X}}

\newcommand{\setPin}{\mathrm{set}_{\pi}}

\newcommand{\clearX}{\mathrm{clear}_{X}}
\newcommand{\clearF}{\mathrm{clear}_{F}}
\newcommand{\clearS}{\mathrm{clear}_{S}}

\newcommand{\td}{\mathcal{T}} 
\newcommand{\tw}{\mathrm{tw}} 
\newcommand{\w}{\mathrm{w}}   

\newcommand{\R}{\mathbb{R}}   

\newcommand{\defquery}[3]{
  \vspace{1mm}
\noindent\fbox{
  \begin{minipage}{0.97\textwidth}
  #1 \\
  {\bf{Output:}} #2  \\
  {\bf{Time:}} #3
  \end{minipage}
}
  \vspace{1mm}
}

\newcommand{\apx}{\textrm{apx}}

\begin{document}

\title{A $c^k n$ 5-Approximation Algorithm for Treewidth}
\author{
  Hans L.\ Bodlaender\thanks{Department of Information and
    Computing Sciences, Utrecht University, P.O. Box 80.089, 3508 TB
    Utrecht, the Netherlands.  {h.l.bodlaender@uu.nl}}
  \and 
  Pål Grønås Drange\thanks{Department of Informatics, Univerity of
    Bergen, Norway.  \{Pal.Drange, Markus.Dregi, fomin,
      Daniel.Lokshtanov, Michal.Pilipczuk\}@ii.uib.no}
  \and
  Markus S.\ Dregi\footnotemark[2]
  \and
  Fedor V. Fomin\footnotemark[2]
  \and
  Daniel Lokshtanov\footnotemark[2]
  \and
  Micha\l{} Pilipczuk\footnotemark[2]}

\maketitle

\begin{abstract}
  We give an algorithm that for an input $n$-vertex graph $G$ and
  integer $k>0$, in time $2^{O(k)} n$ either outputs that the
  treewidth of $G$ is larger than $k$, or gives a tree decomposition
  of $G$ of width at most $5k+4$.  This is the first algorithm
  providing a constant factor approximation for treewidth which runs
  in time single-exponential in $k$ and linear in $n$.
  
  Treewidth based computations are subroutines of numerous algorithms.
  Our algorithm can be used to speed up many such algorithms to work
  in time which is single-exponential in the treewidth and linear in
  the input size.
\end{abstract}

\section{Introduction}
\label{section:introduction}
Since its invention in the 1980s, the notion of treewidth has come to
play a central role in an enormous number of fields, ranging from very
deep structural theories to highly applied areas.  An important (but
not the only) reason for the impact of the notion is that many graph
problems that are intractable on general graphs become efficiently
solvable when the input is a graph of bounded treewidth.  In most
cases, the first step of an algorithm is to find a tree decomposition
of small width and the second step is to perform a dynamic programming
procedure on the tree decomposition.

In particular, if a graph on $n$ vertices is given together with a
tree decomposition of width $k$, many problems can be solved by
dynamic programming in time $2^{O(k)}n$, i.e., single-exponential in
the treewidth and linear in $n$.  Many of the problems admitting such
algorithms have been known for over thirty years~\cite{Bodlaender87}
but new algorithmic techniques on graphs of bounded
treewidth~\cite{BodlaenderCKN12, CyganKN12} as well as new problems
motivated by various applications (just a few of many examples
are~\cite{AbrahamBDR12,Gildea11,KosterHK02,RinaudoPBD12}) continue to
be discovered.  While a reasonably good tree decomposition can be
derived from the properties of the problem sometimes, in most of the
applications, the computation of a good tree decomposition is a
challenge.
Hence the natural question here is what can be done when no tree
decomposition is given.  In other words, is there an algorithm that
for a given graph $G$ and integer $k$, in time $2^{O(k)}n$ either
correctly reports that the treewidth of $G$ is at least $k$, or finds
an optimal solution to our favorite problem (finds a maximum
independent set, computes the chromatic number, decides if $G$ is
Hamiltonian, etc.)?  To answer this question it would be sufficient to
have an algorithm that in time $2^{O(k)}n$ either reports correctly
that the treewidth of $G$ is more that $k$, or construct a tree
decomposition of width at most $ck$ for some constant $c$.

However, the lack of such algorithms has been a bottleneck, both in
theory and in practical applications of the treewidth concept.  The
existing approximation algorithms give us the choice of running times
of the form $2^{O(k)} n^2$, $2^{O(k \log{k})} n\log{n}$, or
$k^{O(k^3)} n$, see Table~\ref{table:tw_history}.  Remarkably, the
newest of these current record holders is now almost 20 years old.
This ``newest record holder'' is the linear time algorithm of
Bodlaender~\cite{Bodlaender93s,Bodlaender96} that given a graph $G$,
decides if the treewidth of $G$ is at most $k$, and if so, gives a
tree decomposition of width at most $k$ in $O(k^{O(k^3)}n)$ time.  The
improvement by Perkovi{\'{c}} and Reed~\cite{PerkovicR00} is only a
factor polynomial in $k$ faster (but also, if the treewidth is larger
than $k$, it gives a subgraph of treewidth more than $k$ with a tree
decomposition of width at most $k$, leading to an $O(n^2)$ algorithm
for the fundamental disjoint paths problem).  Recently, a version
running in logarithmic space was found by Elberfeld et
al.~\cite{ElberfeldJT10}, but its running time is not linear.

\begin{savenotes}
\begin{table}[htdp]
  \begin{center}
    \begin{tabular}{|c|c|c|c|}
      \hline
      Reference                                 & Approximation   &   $f(k)$&  $g(n)$ \\ \hline
      Arnborg  et al.~\cite{ArnborgCP87}        & exact         & $O(1)$ & $O(n^{k+2})$ \\
      Robertson \& Seymour~\cite{RobertsonS13}  & $4k + 3$      &  $O(3^{3k})$  & $n^2$\\
      Lagergren~\cite{Lagergren96}              & $8k + 7$      & $2^{O(k\log{k})}$ & $n\log^2{n}$\\
      Reed~\cite{Reed92}                        & $8k+O(1)$\footnotemark & $2^{O(k\log{k})}$& $n \log{n}$ \\
      Bodlaender~\cite{Bodlaender96}            & exact      &$O(k^{O(k^3)})$ & $n$\\
      Amir~\cite{Amir10}                        & $4.5k$        & $O(2^{3k}k^{3/2} )$ & $n^2$ \\
      Amir~\cite{Amir10}                        & $(3+2/3)k$    & $O(2^{3.6982k}k^3 )$ & $n^2$ \\
      Amir~\cite{Amir10}                        & $O(k\log{k})$ & $O(k\log{k})$ & $n^4$ \\
      Feige et al.~\cite{FeigeHL08}             & $O(k \cdot \sqrt{\log k})$ & $O(1)$  & $n^{O(1)}$ \\
      This paper                                & $3k + 4$      &  $2^{O(k)}$ & $n\log{n}$ \\
      This paper                                & $5k + 4$      &  $2^{O(k)}$ & $n$ \\
      \hline
    \end{tabular}
  \end{center}
  \caption{Overview of treewidth  algorithms.  Here $k$ is the treewidth
    and $n$ is the number of vertices of an input graph $G$.
    Each of the algorithms outputs in time $f(k)\cdot g(n)$ a
    decomposition of width given in the Approximation column.}
  \label{table:tw_history}
\end{table}
\end{savenotes}

\footnotetext{Reed~\cite{Reed92} does not state the approximation ratio of his algorithm explicitly. However, a careful analysis of his manuscript show that the algorithm can be implemented to give a tree decomposition of width at most $8k+O(1)$.}

In this paper, we give the first constant factor approximation
algorithm for the treewidth graph such that its running time is single
exponential in treewidth and linear in the size of the input
graph. Our main result is the following theorem.

\begin{theorem}\label{thm:mainThm}
  There exists an algorithm, that given an $n$-vertex graph $G$ and
  an integer $k$, in time $2^{O(k)} n$ either outputs that the treewidth
  of $G$ is larger than $k$, or constructs a tree decomposition of $G$
  of width at most $5k + 4$.
\end{theorem}

Of independent interest are a number of techniques that we use to
obtain the result and the intermediate result of an algorithm that
either tells that the treewidth is larger than $k$ or outputs a tree
decomposition of width at most $3k+4$ in time $2^{O(k)} n \log n$.


\paragraph{Related results and techniques.}
The basic shape of our algorithm is along the same lines as about all
of the treewidth approximation algorithms~\cite{Amir10,
  BodlaenderGHK95, FeigeHL08, Lagergren96, Reed92, RobertsonS13},
i.e., a specific scheme of repeatedly finding separators.  If we ask
for polynomial time approximation algorithms for treewidth, the
currently best result is that of~\cite{FeigeHL08} that gives in
polynomial (but not linear) time a tree decomposition of width $O(k
\cdot \sqrt{\log k})$ where $k$ is the treewidth of the graph.  Their
work also gives a polynomial time approximation algorithm with ratio
$O(|V_H|^2)$ for $H$-minor free graphs.
By Austrin et al.~\cite{AustrinPW12}, assuming the Small Set Expansion
Conjecture, there is no polynomial time approximation algorithm for
treewidth with a constant performance ratio.

An important element in our algorithms is the use of a data structure
that allows to perform various queries in time $O(c^k \log n)$ each,
for some constant $c$.  This data structure is obtained by adding
various new techniques to old ideas from the area of dynamic
algorithms for graphs of bounded treewidth~\cite{Bodlaender92b,
  CohenSTV93, ChaudhuriZ98, ChaudhuriZ00, Hagerup00}.

A central element in the data structure is a tree decomposition of the
input graph of bounded (but too large) width such that the tree used
in the tree decomposition is binary and of logarithmic depth.  To
obtain this tree decomposition, we combine the following techniques:
following the scheme of the exact linear time
algorithms~\cite{Bodlaender96, PerkovicR00}, but replacing the call to
the dynamic programming algorithm of Bodlaender and
Kloks~\cite{BodlaenderK96} by a recursive call to our algorithm, we
obtain a tree decomposition of $G$ of width at most $10k+9$ (or
$6k+9$, in the case of the $O(c^k n \log n)$ algorithm of
Section~\ref{section:nlogn}.)  Then, we use a result by Bodlaender and
Hagerup~\cite{BodlaenderH98} that this tree decomposition can be
turned in a tree decomposition with a logarithmic depth binary tree in
linear time, or more precisely, in $O(\log n)$ time and $O(n)$ operations on an EREW PRAM.
The cost of this transformation is increasing the width of the decomposition roughly three times.
The latter result is an application of classic results from parallel computing for solving problems on trees,
in particular Miller-Reif tree contraction~\cite{MillerR89,MillerR91}.

Using the data structure to ``implement'' the algorithm of Robertson
and Seymour~\cite{RobertsonS13} already gives an $O(c^k n\log n)$
3-approximation for treewidth (Section~\ref{section:nlogn}).
Additional techniques are needed to speed this algorithm up.  We build
a series of algorithms, with running times of the forms $O(c^k n \log
\log n)$, $ O(c^k n \log \log \log n), \ldots$, etc.  Each algorithm
``implements'' Reeds algorithm~\cite{Reed92}, but with a different
procedure to find balanced separators of the subgraph at hand, and
stops when the subgraph at hand has size $O(\log n)$.  In the latter
case, we call the previous algorithm of the series on this subgraph.

Finally, to obtain a linear time algorithm, we consider two cases, one
case for when $n$ is ``small'' (with respect to $k$), and one case
when $n$ is ``large'', where we consider $n$ to be small if
\[
n \leq 2^{2^{c_0 k^3}} \textrm{, for some constant $c_0$.}
\]
For small values of $n$, we apply $O(c^k n \log^{(2)} n)$ algorithm
from Section~\ref{section:logi}.  This will yield a linear running
time in $n$ since $\log^{(2)}n \leq k$.  For larger values of $n$, we
show that the linear time algorithms of~\cite{Bodlaender96}
or~\cite{PerkovicR00} can be implemented in truly linear time, without any overhead depending on $k$. This seemingly surprising result can be roughly obtained as follows. 
We explicitly construct the finite state tree automaton of the dynamic
programming algorithm in sublinear time before processing the graph, and then view the dynamic programming routine as a run of the automaton, where productions are implemented as constant time table lookups. Viewing a dynamic programming algorithm on a tree decomposition as a finite state automaton traces
back to early work by Fellows and Langston~\cite{FellowsL89}, see e.g., also~\cite{AbrahamsonF93}. Our algorithm assumes the RAM model of computation~\cite{SavageBook}, and the only aspect of the RAM model which is exploited by our algorithm is the ability to look up an entry in a table in constant time, independently of the size of the table. This capability is crucially used in almost every linear time graph algorithm including breadth first search and depth first search.

%

\paragraph{Overview of the paper.}
In Section~\ref{sec:outline} we give the outline of the main algorithms, focusing on explaining main intuitions rather than formal details of the proofs.
Some concluding remarks and open questions are made in Section~\ref{section:conclusions}. 
Sections~\ref{section:nlogn}, \ref{section:logi}
and~\ref{section:linear} give the formal descriptions of the main algorithms: first the $O(c^k n
\log n)$ algorithm, then the series of $O(c^k n \log^{(\depth)} n)$
algorithms, before describing the $O(c^k n)$ algorithm.  Each of the
algorithms described in these sections use queries to a data structure
which is described in Section~\ref{section:datastructure}.

\paragraph{Notation.}
We give some basic definitions and notation, used throughout the paper.  For
$\depth \in \mathbb{N}$, the function $\log^{(\depth)} n$ is defined
as follows: $\log^{(1)} n = \log n$, and for $\depth > 1$,
$\log^{(\depth)} n = \log (\log^{(\depth-1)} n)$.

For the presentation of our results, it is more convenient when we
regard tree decompositions as rooted.  This yields the following
definition of tree decompositions.
\begin{definition}
  \label{def:prelim:treewidth}
  A \emph{tree decomposition} of a graph $G=(V,E)$ is a pair $\td =
  (\{B_i \mid i \in I\}, T=(I,F))$ where $T=(I,F)$ is a rooted tree,
  and $\{B_i \mid i \in I\}$ is a family of subsets of $V$, such that
  \begin{itemize}
  \item $\bigcup_{i\in I} B_i = V$,
  \item for all $\{v,w\}\in E$, there exists an $i\in I$ with $v,w\in B_i$,
    and
  \item for all $v\in V$, $I_v = \{i\in I \mid v \in B_i \}$ induces a
    subtree of $T$.
  \end{itemize}
  The \emph{width} of $\td = (\{B_i \mid i\in I\}, T=(I,F))$, denoted
  $\w(\td)$ is $\max_{i\in I} |B_i|-1$.  The \emph{treewidth} of a
  graph $G$, denoted by $\tw(G)$, is the minimum width of a tree
  decomposition of $G$.
\end{definition}
For each $v \in V$, the tree induced by $I_v$ is denoted by $T_v$; the
root of this tree, i.e., the node in the tree with smallest distance
to the root of $T$ is denoted by $r_v$.

The sets $B_i$ are called the \emph{bags} of the decomposition.  For
each node $i \in I$, we define $V_i$ to be the union of bags contained
in the subtree of $T$ rooted in $i$, including $B_i$.  Moreover, we
denote $W_i = V_i\setminus B_i$ and $G_i = G[V_i]$, $H_i = G[W_i]$.
Note that by the definition of tree decomposition, $B_i$ separates
$W_i$ from $V\setminus V_i$.

\section{Proof outline}
\label{sec:outline}
Our algorithm builds on the constant factor approximation algorithm
for treewidth described in Graph Minors XIII~\cite{RobertsonS13} with
running time $O(c^kn^2)$.  We start with a brief explanation of a
variant of this algorithm.

\subsection{The $O(3^{3k}n^2)$ time $4$-approximation algorithm from
  Graph Minors XIII.}\label{sec:nSquareAppx}
The engine behind the algorithm is a lemma that states that graphs of
treewidth $k$ have balanced separators of size $k+1$.
In particular, for any way to assign non-negative weights to the
vertices there exists a set $X$ of size at most $k+1$ such that the
total weight of any connected component of $G \setminus X$ is at most
half of the total weight of $G$.  We will use the variant of the lemma
where some vertices have weight $1$ and some have weight $0$.

\begin{lemma}[Graph Minors II~\cite{RobertsonS2}]
  \label{lemma:halfhalf}
  If $\tw(G)\leq k$ and $S \subseteq V(G)$, then there exists $X \subseteq V(G)$ with
  $|X| \leq k+1$ such that every component of $G \setminus X$ has at
  most $\frac{1}{2}|S|$ vertices which are in $S$.
\end{lemma}

We note that the original version of~\cite{RobertsonS2} is seemingly
stronger: it gives bound $\frac{1}{2}|S\setminus X|$ instead of
$\frac{1}{2}|S|$.  However, we do not need this stronger version and
we find it more convenient to work with the weaker.  The set $X$ with
properties ensured by Lemma~\ref{lemma:halfhalf} will be called a
{\emph{balanced $S$-separator}}, or a {\emph{$\frac{1}{2}$-balanced
    $S$-separator}}.  More generally, for an {\emph{$\beta$-balanced
    $S$-separator}} $X$ every connected component of $G\setminus X$
contains at most $\beta|S|$ vertices of $S$.  If we omit the set $S$,
i.e., talk about separators instead of $S$-separators, we mean
$S=V(G)$ and balanced separators of the whole vertex set.

The proof of Lemma~\ref{lemma:halfhalf} is not too hard; start with a
tree decomposition of $G$ with width at most $k$ and orient every edge
of the decomposition tree towards the side which contains the larger
part of the set $S$.  Two edges of the decomposition can not point
``in different directions'', since then there would be disjoint parts
of the tree, both containing more than half of $S$.  Thus there has to
be a node in the decomposition tree such that all edges of the
decomposition are oriented towards it.  The bag of the decomposition
corresponding to this node is exactly the set $X$ of at most $k+1$
vertices whose deletion leaves connected components with at most
$\frac{1}{2}|S|$ vertices of $S$ each.

The proof of Lemma~\ref{lemma:halfhalf} is constructive if one has
access to a tree decomposition of $G$ of width less than $k$.  The
algorithm does not have such a decomposition at hand, after all we are
trying to compute a decomposition of $G$ of small width.  Thus we have
to settle for the following algorithmic variant of
lemma~\cite{RobertsonS2}.

\begin{lemma}[\cite{RobertsonS13}]
  \label{lem:halfhalfalg}
  There is an algorithm that given a graph $G$, a set $S$ and a $k \in
  \mathbb{N}$ either concludes that $\tw(G) > k$ or outputs a set $X$
  of size at most $k + 1$ such that every component of $G \setminus X$
  has at most $\frac{2}{3}|S|$ vertices which are in $S$ and runs in
  time $O(3^{|S|}k^{O(1)}(n + m))$.
\end{lemma}

\begin{proof}[Proof sketch.]
  By Lemma~\ref{lemma:halfhalf} there exists a set $X'$ of size at
  most $k+1$ such that every component of $G \setminus X'$ has at most
  $\frac{1}{2}|S|$ vertices which are in $S$.  A simple packing
  argument shows that the components can be assigned to left or right
  such that at most $\frac{2}{3}|S|$ vertices of $S$ go left and at
  most $\frac{2}{3}|S|$ go right.  Let $S_X$ be $S \cap X'$ and let
  $S_L$ and $S_R$ be the vertices of $S$ that were put left and right
  respectively.  By trying all partitions of $S$ in three parts the
  algorithm correctly guesses $S_X$, $S_L$ and $S_R$.  Now $X$
  separates $S_L$ from $S_R$ and so the minimum vertex cut between
  $S_L$ and $S_R$ in $G \setminus S_X$ is at most $|X \setminus S_X|
  \leq (k+1)-|S_X|$.  The algorithm finds using max-flow a set $Z$ of
  size at most $(k+1)-|S_X|$ that separates $S_L$ from $S_R$ in $G
  \setminus S_X$.  Since we are only interested in a set $Z$ of size
  at most $k-|S_X|$ one can run max-flow in time $O((n+m)k^{O(1)})$.
  Having found $S_L$, $S_R$, $S_X$ and $Z$ the algorithm sets $X = S_X
  \cup Z$, $L$ to contain all components of $G \setminus X$ that
  contain vertices of $S_L$ and $R$ to contain all other vertices.
  Since every component $C$ of $G \setminus X$ is fully contained in
  $L$ or $R$, the bound on $|C \cap S|$ follows.
  
  If no partition of $S$ into $S_L$, $S_R$, $S_X$ yielded a cutset $Z$
  of size $\leq (k+1)-|S_X|$, this means that $\tw(G)>k$, which the
  algorithm reports.
\end{proof}

The algorithm takes as input $G$, $k$ and a set $S$ on at most $3k+3$
vertices, and either concludes that the treewidth of $G$ is larger
than $k$ or finds a tree decomposition of width at most $4k+3$ such
that the top bag of the decomposition contains $S$.

On input $G$, $S$, $k$ the algorithm starts by ensuring that
$|S|=3k+3$.  If $|S| < 3k+3$ the algorithm just adds arbitrary
vertices to $S$ until equality is obtained.  Then the algorithm
applies Lemma~\ref{lem:halfhalfalg} and finds a set $X$ of size at
most $k+1$ such that each component $C_i$ of $G \setminus X$ satisfies
$|C_i \cap S| \leq \frac{2|S|}{3} \leq 2k+2$.  Thus for each $C_i$ we
have $|(S \cap C_i) \cup X| \leq 3k+3$.  For each component $C_i$ of
$G \setminus X$ the algorithm runs itself recursively on $(G[C_i \cup
X], (S \cap C_i) \cup X, k)$.

If either of the recursive calls returns that the treewidth is more
than $k$ then the treewidth of $G$ is more than $k$ as well.
Otherwise we have for every component $C_i$ a tree decomposition of
$G[C_i \cup X]$ of width at most $4k+3$ such that the top bag contains
$(S \cap C_i) \cup X$.  To make a tree decomposition of $G$ we make a
new root node with bag $X \cup S$, and connect this bag to the roots
of the tree decompositions of $G[C_i \cup X]$ for each component
$C_i$.  It is easy to verify that this is indeed a tree decomposition
of $G$.  The top bag contains $S$, and the size of the top bag is at
most $|S|+|X| \leq 4k+4$, and so the width if the decomposition is at
most $4k+3$ as claimed.

The running time of the algorithm is governed by the recurrence
\begin{align}\label{eqn:gm13} T(n,k) \leq O(3^{|S|}k^{O(1)}(n+m)) +
  \sum_{C_i} T(|C_i \cup X|, k) \end{align} which solves to $T(n,k) \leq
(3^{3k}k^{O(1)}n(n+m))$ since $|S|=3k+3$ and there always are at least
two non-empty components of $G \setminus X$.  Finally, we use the
following observation about the number of edges in a graph of
treewidth $k$.

\begin{lemma}[\cite{BodlaenderF05a}]
  \label{lemma:prelim:nk-edges}
  Let $G = (V,E)$ be a graph with treewidth at most $k$.  Then $|E|
  \leq |V|k$.
\end{lemma}

Thus if $|E| > nk$ the algorithm can safely output that $\tw(G) > k$.
After this, running the algorithm above takes time
$O(3^{3k}k^{O(1)}n(n+m)) = O(3^{3k}k^{O(1)}n^2)$.


\subsection{The $O(k^{O(k)}n\log n)$ time approximation algorithm of
  Reed.}\label{sec:reedAlg}
Reed~\cite{Reed92} observed that the running time algorithm of
Robertson and Seymour~\cite{RobertsonS13} can be sped up from $O(n^2)$
for fixed $k$ to $O(n \log n)$ for fixed $k$, at the cost of a worse
(but still constant) approximation ratio, and a $k^{O(k)}$ dependence
on $k$ in the running time, rather than the ${3}^{3k}$ factor in the
algorithm of Robertson and Seymour.  We remark here that
Reed~\cite{Reed92} never states explicitly the dependence on $k$ of
his algorithm, but a careful analysis shows that this dependence is in
fact of order $k^{O(k)}$.  The main idea of this algorithm is that the
recurrence in Equation~\ref{eqn:gm13} only solves to $O(n^2)$ for
fixed $k$ if one of the components of $G \setminus X$ contains almost
all of the vertices of $G$.  If one could ensure that each component
$C_i$ of $G \setminus X$ had at most $c \cdot n$ vertices for some
fixed $c < 1$, the recurrence in Equation~\ref{eqn:gm13} solves to
$O(n \log n)$ for fixed $k$.  To see that this is true we simply
consider the recursion tree.  The total amount of work done at any
level of the recursion tree is $O(n)$ for a fixed $k$.  Since the size
of the components considered at one level is always a constant factor
smaller than the size of the components considered in the previous
level, the number of levels is only $O(\log n)$ and we have $O(n \log
n)$ work in total.

By using Lemma~\ref{lemma:halfhalf} with $S = V(G)$ we see that if $G$
has treewidth $\leq k$, then there is a set $X$ of size at most $k+1$
such that each component of $G \setminus X$ has size at most
$\frac{n}{2}$.  Unfortunately if we try to apply
Lemma~\ref{lem:halfhalfalg} to {\em find} an $X$ which splits $G$ in a
balanced way using $S = V(G)$, the algorithm of
Lemma~\ref{lem:halfhalfalg} takes time $O(3^{|S|}k^{O(1)}(n+m)) =
O(3^nn^{O(1)})$, which is exponential in $n$.  Reed~\cite{Reed92} gave
an algorithmic variant of Lemma~\ref{lemma:halfhalf} especially
tailored for the case where $S = V(G)$.

\begin{lemma}[\cite{Reed92}]
  \label{lem:reedBalSep} There is an algorithm that given $G$ and $k$,
  runs in time $O(k^{O(k)}n)$ and either concludes that $\tw(G)>k$ or
  outputs a set $X$ of size at most $k+1$ such that that every
  component of $G \setminus X$ has at most $\frac{3}{4}|S|$ vertices
  which are in $S$.
\end{lemma}
Let us remark that Lemma~\ref{lem:reedBalSep} as stated here is never
explicitly proved in~\cite{Reed92}, but it follows easily from the
arguments given there.

Having Lemmata~\ref{lem:halfhalfalg} and~\ref{lem:reedBalSep} at hand,
we show how to obtain an $8$-approximation of treewidth in time
$O(k^{O(k)}n\log n)$.  The algorithm takes as input $G$, $k$ and a set
$S$ on at most $6k+6$ vertices, and either concludes that the
treewidth of $G$ is at least $k$, or finds a tree decomposition of
width at most $8k+7$ such that the top bag of the decomposition
contains $S$.

On input $G$, $S$, $k$ the algorithm starts by ensuring that
$|S|=6k+6$.  If $|S| < 6k+6$ the algorithm just adds vertices to $S$
until equality is obtained.  Then the algorithm applies
Lemma~\ref{lem:halfhalfalg} and finds a set $X_1$ of size at most
$k+1$ such that each component $C_i$ of $G \setminus X_1$ satisfies
$|C_i \cap S| \leq \frac{2}{3}|S| \leq 4k+4$.  Now the algorithm
applies Lemma~\ref{lem:reedBalSep} and finds a set $X_2$ of size at
most $k+1$ such that each component $C_i$ of $G \setminus X_2$
satisfies $|C_i| \leq \frac{3}{4}|V(G)| \leq \frac{3}{4}n$.  Set $X =
X_1 \cup X_2$.  For each component $C_i$ of $G \setminus S$ we have
that $|(S \cap C_i) \cup X| \leq 6k+6$.  For each component $C_i$ of
$G \setminus X$ the algorithm runs itself recursively on $(G[C_i \cup
X], (S \cap C_i) \cup X, k)$.

If either of the recursive calls returns that the treewidth is more
than $k$ then the treewidth of $G$ is more than $k$ as well.
Otherwise we have for every component $C_i$ a tree decomposition of
$G[C_i \cup X]$ of width at most $8k+7$ such that the top bag contains
$(S \cap C_i) \cup X$.  Similarly as before, to make a tree
decomposition of $G$ we make a new root node with bag $X \cup S$, and
connect this bag to the roots of the tree decompositions of $G[C_i
\cup X]$ for each component $C_i$.  It is easy to verify that this is
indeed a tree decomposition of $G$.  The top bag contains $S$, and the
size of the top bag is at most $|S|+|X| \leq |S| + |X_1| + |X_2| \leq
6k + 6 + 2k + 2 = 8k + 8$, and the width of the decomposition is at
most $8k+7$ as claimed.

The running time of the algorithm is governed by the recurrence
\begin{align}\label{eqn:reedTime}
  T(n,k) \leq O\left(k^{O(k)}(n+m)\right) + \sum_{C_i} T(|C_i \cup X|, k)
\end{align}
which solves to $T(n,k) \leq O(k^{O(k)}(n+m)\log n)$ since each $C_i$ has
size at most $\frac{3}{4}n$.  By Lemma~\ref{lemma:prelim:nk-edges} we
have $m \leq kn$ and so the running time of the algorithm is upper
bounded by $O(k^{O(k)}n\log n)$.


\subsection{A new $O(c^kn \log n)$ time 3-approximation
  algorithm}\label{sec:recursiveScheme}
The goal of this section is to sketch a proof of the following
theorem.  A full proof of Theorem~\ref{theorem:nlogn} can be found in
Section~\ref{section:nlogn}.
\begin{theorem}
  \label{theorem:nlogn}
  There exists an algorithm which given a graph $G$ and an integer
  $k$, either computes a tree decomposition of $G$ of width at most
  $3k + 4$ or correctly concludes that $\tw(G) > k$, in time
  $O(c^k\cdot n \log n)$ for some $c \in \mathbb{N}$.
\end{theorem}

The algorithm employs the same recursive compression scheme which is
used in Bodlaender's linear time
algorithm~\cite{Bodlaender93s,Bodlaender96} and the algorithm
of~Perkovi{\'{c}} and Reed~\cite{PerkovicR00}.  The idea is to solve
the problem recursively on a smaller instance, expand the obtained
tree decomposition of the smaller graph to a ``good, but not quite
good enough'' tree decomposition of the instance in question, and then
use this tree decomposition to either conclude that $\tw(G) > k$ or
find a decomposition of $G$ which is good enough.  A central concept
in this recursive approach of~\cite{Bodlaender96} is the definition of
an improved graph:

\begin{definition}\label{def:prelim:improved-graph}
  Given a graph $G=(V,E)$ and an integer $k$, the \emph{improved}
  graph of $G$, denoted $G_I$, is obtained by adding an edge between
  each pair of vertices with at least $k+1$ common neighbors of degree
  at most $k$ in $G$.
\end{definition}

Intuitively, adding the edges during construction of the improved
graph cannot spoil any tree decomposition of $G$ of width at most $k$,
as the pairs of vertices connected by the new edges will need to be
contained together in some bag anyway.  This is captured in the
following lemma.

\begin{lemma}
  \label{lem:improvedGraphTw}
  Given a graph $G$ and an integer $k \in \mathbb{N}$, $\tw(G) \leq k$ if
  and only if $\tw(G_I) \leq k$.
\end{lemma}

If $|E|=O(kn)$, which is the case in graphs of treewidth at most $k$,
the improved graph can be computed in $O(k^{O(1)}\cdot n)$ time using
radix sort~\cite{Bodlaender96}.

A vertex $v\in G$ is $I$-simplicial, if it is simplicial in the improved
graph $G_I$.  The intuition behind $I$-simplicial vertices is as
follows: all the neighbors of an $I$-simplicial vertex must be
simultaneously contained in some bag of any tree decomposition of
$G_I$ of width at most $k$, so we can safely remove such vertices from
the improved graph, compute the tree decomposition, and reintroduce
the removed $I$-simplicial vertices.  The crucial observation is that
if no large set of $I$-simplicial vertices can be found, then one can
identify a large matching, which can be also used for a robust
recursion step.  The following lemma, which follows from the work of
Bodlaender~\cite{Bodlaender96}, encapsulates all the main ingredients
that we will use.

\begin{lemma}\label{lem:prelim:bodlaender}
  There is an algorithm working in $O(k^{O(1)}n)$ time that, given a
  graph $G = (V,E)$ and an integer $k$, either
  \begin{enumerate}[(i)]
  \item returns a maximal matching in $G$ of cardinality at least
    $\frac{|V|}{O(k^6)}$,
  \item returns a set of at least $\frac{|V|}{O(k^6)}$ $I$-simplicial
    vertices, or,
  \item correctly concludes that the treewidth of $G$ is larger than
    $k$.
  \end{enumerate}
  Moreover, if a set $X$ of at least $\frac{|V|}{O(k^6)}$
  $I$-simplicial vertices is returned, and the algorithm is in
  addition provided with some tree decomposition $\td_I$ of
  $G_I\setminus X$ of width at most $k$, then in $O(k^{O(1)}\cdot n)$
  time one can turn $\td_I$ into a tree decomposition $\td$ of $G$ of
  width at most $k$, or conclude that the treewidth of $G$ is larger
  than $k$.
\end{lemma}

Lemma~\ref{lem:prelim:bodlaender} allows us to reduce the problem to a
\emph{compression} variant where we are given a graph $G$, an integer
$k$ and a tree decomposition of $G$ of width $O(k)$, and the goal is
to either conclude that the treewidth of $G$ is at least $k$ or output
a tree decomposition of $G$ of width at most $3k+4$.  The proof of
Theorem~\ref{theorem:nlogn} has two parts: an algorithm for the
compression step and an algorithm for the general problem that uses
the algorithm for the compression step together with
Lemma~\ref{lem:prelim:bodlaender} as black boxes.  We now state the
properties of our algorithm for the compression step in the following
lemma.

\begin{lemma}\label{lemma:nlogn-compression}
  There exists an algorithm which on input $G,k,S_0,\td_\apx$, where
  (i) $S_0\subseteq V(G)$, $|S_0|\leq 2k+3$, (ii) $G$ and $G \setminus
  S_0$ are connected, and (iii) $\td_\apx$ is a tree decomposition of
  $G$ of width at most $O(k)$, in $O(c^k\cdot n \log n)$ time for some
  $c \in \mathbb{N}$ either computes a tree decomposition $\td$ of $G$
  with $\w(\td) \leq 3k+4$ and $S_0$ as the root bag, or correctly
  concludes that $\tw(G)>k$.
\end{lemma}

We now give a proof of Theorem~\ref{theorem:nlogn}, assuming the
correctness of Lemma~\ref{lemma:nlogn-compression}.  The correctness
of the lemma will be argued for in Sections~\ref{sec:firstCompress}
and~\ref{sec:proofOfCompress}.
 
\begin{proof}[Proof of Theorem~\ref{theorem:nlogn}]
  Our algorithm will in fact solve a slightly more general problem.
  Here we are given a graph $G$, an integer $k$ and a set $S_0$ on at
  most $2k+3$ vertices, with the property that $G \setminus S_0$ is
  connected.  The algorithm will either conclude that $\tw(G) > k$ or
  output a tree decomposition of width at most $3k+4$ such that $S_0$
  is the root bag.  To get a tree decomposition of any (possibly
  disconnected) graph it is sufficient to run this algorithm on each
  connected component with $S_0 = \emptyset$.  The algorithm proceeds
  as follows.  It first applies Lemma~\ref{lem:prelim:bodlaender} on
  $(G, 3k+4)$.  If the algorithm of Lemma~\ref{lem:prelim:bodlaender}
  concludes that $\tw(G) > 3k+4$ the algorithm reports that $\tw(G) >
  3k+4 > k$.
  
  If the algorithm finds a matching $M$ in $G$ with at least
  $\frac{|V|}{O(k^6)}$ edges, it contracts every edge in $M$ and
  obtains a graph $G'$.  Since $G'$ is a minor of $G$ we know that
  $\tw(G') \leq \tw(G)$.  The algorithm runs itself recursively on
  $(G', k, \emptyset)$, and either concludes that $\tw(G') > k$
  (implying $\tw(G) > k$) or outputs a tree decomposition of $G'$ of
  width at most $3k+4$.  Uncontracting the matching in this tree
  decomposition yields a tree decomposition $\td_\apx$ of $G$ of width
  at most $6k+9$~\cite{Bodlaender96}.  Now we can run the algorithm of
  Lemma~\ref{lemma:nlogn-compression} on $(G,k,S_0,\td_\apx)$ and
  either obtain a tree decomposition of $G$ of width at most $3k+4$
  and $S_0$ as the root bag, or correctly conclude that $\tw(G)>k$.
  
  If the algorithm finds a set $X$ of at least $\frac{|V|}{O(k^6)}$
  $I$-simplicial vertices, it constructs the improved graph $G_I$ and
  runs itself recursively on $(G_I \setminus X, k, \emptyset)$.  If
  the algorithm concludes that $\tw(G_I \setminus X) > k$ then
  $\tw(G_I) > k$ implying $\tw(G) > k$ by
  Lemma~\ref{lem:improvedGraphTw}.  Otherwise we obtain a tree
  decomposition of $G_I \setminus X$ of width at most $3k+4$.  We may
  now apply Lemma~\ref{lem:prelim:bodlaender} and obtain a tree
  decomposition $\td_\apx$ of $G$ with the same width.  Note that we
  can not just output $\td_\apx$ directly, since we can not be sure
  that $S_0$ is the top bag of $\td_\apx$.  However we can run the
  algorithm of Lemma~\ref{lemma:nlogn-compression} on
  $(G,k,S_0,\td_\apx)$ and either obtain a tree decomposition of $G$
  of width at most $3k+4$ and $S_0$ as the root bag, or correctly
  conclude that $\tw(G)>k$.
  
  It remains to analyze the running time of the algorithm.  Suppose
  the algorithm takes time at most $T(n,k)$ on input $(G,k,S_0)$ where
  $n = |V(G)|$.  Running the algorithm of
  Lemma~\ref{lem:prelim:bodlaender} takes $O(k^{O(1)}n)$ time.  Then
  the algorithm either halts, or calls itself recursively on a graph
  with at most $n-\frac{n}{O(k^6)} = n(1-\frac{1}{O(k^6)})$ vertices
  taking time $T(n(1-\frac{1}{O(k^6)}), k)$.  Then the algorithm takes
  time $O(k^{O(1)}n)$ to either conclude that $\tw(G)>k$ or to
  construct a tree decomposition $\td_\apx$ of $G$ of width $O(k)$.
  In the latter case we finally run the algorithm of
  Lemma~\ref{lemma:nlogn-compression}, taking time $O(c^k\cdot n \log
  n)$.  This gives the following recurrence:
  \[
  T(n,k) \leq O\left(c^k\cdot n \log n\right) + T\left(n
    \left(1-\frac{1}{O(k^6)} \right), k \right)
  \]
  The recurrence leads to a geometric series and solves to $T(n,k)
  \leq O(c^kk^{O(1)}\cdot n \log n)$, completing the proof.  For a
  thorough analysis of the recurrence, see
  Equations~\ref{eqn:recurrenceOneD} and~\ref{eqn:recurrenceTwoD} in
  Section~\ref{section:nlogn}.  Pseudocode for the algorithm described
  here is given in Algorithm~\ref{alg:nlogn-alg-1} in
  Section~\ref{section:nlogn}.
\end{proof}

\subsubsection{A compression algorithm}\label{sec:firstCompress}
We now proceed to give a sketch of a proof for a slightly weakened
form of Lemma~\ref{lemma:nlogn-compression}.  The goal is to give an
algorithm that given as input a graph $G$, an integer $k$, a set $S_0$
of size at most $6k+6$ such that $G\setminus S_0$ is connected, and a
tree decomposition $\td_\apx$ of $G$, runs in time $O(c^kn \log n)$
and either correctly concludes that $\tw(G) > k$ or outputs a tree
decomposition of $G$ of width at most $8k+7$.  The paper does not
contain a full proof of this variant of
Lemma~\ref{lemma:nlogn-compression} --- we will discuss the proof of
Lemma~\ref{lemma:nlogn-compression} in
Section~\ref{sec:proofOfCompress}.  The aim of this section is to
demonstrate that the recursive scheme of
Section~\ref{sec:recursiveScheme} together with a nice trick for
finding balanced separators is already sufficient to obtain a factor
$8$ approximation for treewidth running in time $O(c^k n\log n)$.  A
variant of the trick used in this section for computing balanced
separators turns out to be useful in our final $O(c^kn)$ time
$5$-approximation algorithm.

The route we follow here
is to apply the algorithm of Reed described in
Section~\ref{sec:reedAlg}, but instead of using
Lemma~\ref{lem:reedBalSep} to find a set $X$ of size $k+1$ such that
every connected component of $G \setminus X$ is small, finding $X$ by
dynamic programming over the tree decomposition $\td_\apx$ in time
$O(c^kn)$.  There are a few technical difficulties with this approach.

The most serious issue is that, to the best of our knowledge, the only
known dynamic programming algorithms for balanced separators in graphs
of bounded treewidth take time $O(c^kn^2)$ rather than $O(c^kn)$: in
the state, apart from a partition of the bag, we also need to store
the cardinalities of the sides which gives us another dimension of
size $n$.  We now explain how it is possible to overcome this issue.
We start by applying the argument in the proof of
Lemma~\ref{lemma:halfhalf} on the tree decomposition $\td_\apx$ and
get in time $O(k^{O(1)}n)$ a partition of $V(G)$ into $L_0$, $X_0$ and
$R_0$ such that there are no edges between $L_0$ and $R_0$,
$\max(|L_0|, |R_0|) \leq \frac{3}{4}n$ and $|X_0| \leq
\w(\td_\apx)+1$.  For every way of writing $k+1 = k_L + k_X + k_R$ and
every partition of $X_0$ into $X_L \cup X_X \cup X_R$ with $|X_X| =
k_X$ we do the following.

First we find in time $O(c^kn)$ using dynamic programming over the
tree decomposition $\td_\apx$ a partition of $L_0 \cup X_0$ into
$\hat{L}_L \cup \hat{R}_L \cup \hat{X}_L$ such that there are no edges
from $\hat{L}_L$ to $\hat{R}_L$, $|\hat{X}_L| \leq k_L + k_X$, $X_X
\subseteq \hat{X}_L$, $X_R \subseteq \hat{R}_L$ and $X_L \subseteq
\hat{L}_L$ and the size $|\hat{L}_L|$ is maximized.  Then we find in
time $O(c^kn)$ using dynamic programming over the tree decomposition
$\td_\apx$ a partition of $R_0 \cup X_0$ into $\hat{L}_R \cup \hat{R}_R
\cup \hat{X}_R$ such that there are no edges from $\hat{L}_R$ to
$\hat{R}_R$, $|\hat{X}_R| \leq k_R + k_X$, $X_X \subseteq \hat{X}_R$,
$X_R \subseteq \hat{R}_R$ and $X_L \subseteq \hat{L}_R$ and the size
$|\hat{R}_R|$ is maximized.  Let $L = L_L \cup L_R$, $R = R_L \cup
R_R$ and $X = X_L \cup X_R$.  The sets $L$, $X$, $R$ form a partition
of $V(G)$ with no edges from $L$ to $R$ and $|X| \leq k_L + k_X + k_R
+ k_X - k_X \leq k+1$.

It is possible to show using a combinatorial argument (see
Lemma~\ref{lem:c4vc} in Section~\ref{section:datastructure}) that if
$\tw(G) \leq k$ then there exists a choice of $k_L$, $k_X$, $k_R$ such
that $k+1 = k_L + k_X + k_R$ and partition of $X_0$ into $X_L \cup X_X
\cup X_R$ with $|X_X| = k_X$ such that the above algorithm will output
a partition of $V(G)$ into $X$, $L$ and $R$ such that $\max(|L|,|R|)
\leq \frac{8n}{9}$.  Thus we have an algorithm that in time $O(c^kn)$
either finds a set $X$ of size at most $k+1$ such that each connected
component of $G\setminus X$ has size at most $\frac{8n}{9}$ or
correctly concludes that $\tw(G) > k$.

\smallskip

The second problem with the approach is that the algorithm of Reed is
an $8$-approximation algorithm rather than a $3$-approximation.  Thus,
even the sped up version does not quite prove
Lemma~\ref{lemma:nlogn-compression}.  It does however yield a version
of Lemma~\ref{lemma:nlogn-compression} where the compression algorithm
is an $8$-approximation.  In the proof of Theorem~\ref{theorem:nlogn}
there is nothing special about the number $3$ and so one can use this
weaker variant of Lemma~\ref{lemma:nlogn-compression} to give a
$8$-approximation algorithm for treewidth in time $O(c^kn \log n)$.
We will not give complete details of this algorithm, as we will
shortly describe a proof of Lemma~\ref{lemma:nlogn-compression} using
a quite different route.

It looks difficult to improve the algorithm above to an algorithm with
running time $O(c^kn)$.  The main hurdle is the following: both the
algorithm of Robertson and Seymour~\cite{RobertsonS13} and the
algorithm of Reed~\cite{Reed92} find a separator $X$ and proceed
recursively on the components of $G \setminus X$.  If we use $O(c^k
\cdot n)$ time to find the separator $X$, then the total running time
must be at least $O(c^k \cdot n \cdot d)$ where $d$ is the depth of
the recursion tree of the algorithm.  It is easy to see that the depth
of the tree decomposition output by the algorithms equals (up to
constant factors) the depth of the recursion tree.  However there
exist graphs of treewidth $k$ such that no tree decomposition of depth
$o(\log n)$ has width $O(k)$ (take for example powers of paths).  Thus
the depth of the constructed tree decompositions, and hence the
recursion depth of the algorithm must be at least $\Omega(\log n)$.

Even if we somehow managed to reuse computations and find the
separator $X$ in time $O(c^k \cdot \frac{n}{\log n})$ on average, we
would still be in trouble since we need to pass on the list of
vertices of the connected components of $G \setminus X$ that we will
call the algorithm on recursively.  At a first glance this has to take
$O(n)$ time and then we are stuck with an algorithm with running time
$O((c^k \cdot \frac{n}{\log n} + n) \cdot d)$, where $d$ is the
recursion depth of the algorithm.  For $d = \log n$ this is still
$O(c^kn + n\log n)$ which is slower than what we are aiming at.  In
Section~\ref{sec:proofOfCompress} we give a proof of
Lemma~\ref{lemma:nlogn-compression} that {\em almost} overcomes these
issues.

\subsubsection{A better compression algorithm}
\label{sec:proofOfCompress}
We give a sketch of the proof of Lemma~\ref{lemma:nlogn-compression}.
The goal is to give an algorithm that given as input a connected graph
$G$, an integer $k$, a set $S_0$ of size at most $2k+3$ such that
$G\setminus S_0$ is connected, and a tree decomposition $\td_\apx$ of
$G$, runs in time $O(c^kn \log n)$ and either correctly concludes that
$\tw(G) > k$ or outputs a tree decomposition of $G$ of width at most
$3k+4$ with top bag $S_0$.

Our strategy is to implement the $O(c^kn^2)$ time $4$-approximation
algorithm described in Section~\ref{sec:nSquareAppx}, but make some
crucial changes in order to (a) make the implementation run in
$O(c^kn\log n)$ time, and (b) make it a $3$-approximation rather than
a $4$-approximation.  We first turn to the easier of the two changes,
namely making the algorithm a $3$-approximation.

To get an algorithm that satisfies all of the requirements of
Lemma~\ref{lemma:nlogn-compression}, but runs in time $O(c^kn^2)$
rather than $O(c^kn\log n)$ we run the algorithm described in
Section~\ref{sec:nSquareAppx} setting $S = S_0$ in the beginning.
Instead of using Lemma~\ref{lem:halfhalfalg} to find a set $X$ such
that every component of $G \setminus X$ has at most $\frac{2}{3}|S|$
vertices which are in $S$, we apply Lemma~\ref{lemma:halfhalf} to show
the {\em existence} of an $X$ such that every component of $G
\setminus X$ has at most $\frac{1}{2}|S|$ vertices which are in $S$,
and do dynamic programming over the tree decomposition $\td_\apx$ in
time $O(c^kn)$ in order to find such an $X$.  Going through the
analysis of Section~\ref{sec:nSquareAppx} but with $X$ satisfying that
every component of $G \setminus X$ has at most $\frac{1}{2}|S|$
vertices which are in $S$ shows that the algorithm does in fact output
a tree decomposition with width $3k+4$ and top bag $S_0$ whenever
$\tw(G) \leq k$.

It is somewhat non-trivial to do dynamic programming over the tree
decomposition $\td_\apx$ in time $O(c^kn)$ in order to find an $X$
such that every component of $G \setminus X$ has at most
$\frac{2}{3}|S|$ vertices which are in $S$.  The problem is that $G
\setminus X$ could potentially have many components and we do not have
time to store information about each of these components individually.
The following lemma, whose proof can be found in
Section~\ref{sec:qssep}, shows how to deal with this problem.

\begin{lemma}\label{lem:balanced-3coloring}
  Let $G$ be a graph and $S\subseteq V(G)$.  Then a set $X$ is a balanced
  $S$-separator if and only if there exists a partition
  $(M_1,M_2,M_3)$ of $V(G)\setminus X$, such that there is no edge
  between $M_i$ and $M_j$ for $i\neq j$, and $|M_i\cap S|\leq |S|/2$
  for $i=1,2,3$.
\end{lemma}
Lemma~\ref{lem:balanced-3coloring} shows that when looking for a
balanced $S$-separator we can just look for a partition of $G$ into
four sets $X,M_1,M_2,M_3$ such that there is no edge between $M_i$ and
$M_j$ for $i\neq j$, and $|M_i\cap S|\leq |S|/2$ for $i=1,2,3$.  This
can easily be done in time $O(c^kn)$ by dynamic programming over the
tree decomposition $\td_\apx$.  This yields the promised algorithm
that satisfies all of the requirements of
Lemma~\ref{lemma:nlogn-compression}, but runs in time $O(c^kn^2)$
rather than $O(c^kn\log n)$.

\medskip We now turn to the most difficult part of the proof of
Lemma~\ref{lemma:nlogn-compression}, namely how to improve the running
time of the algorithm above from $O(c^kn^2)$ to $O(c^kn\log n)$ in a
way that gives hope of a further improvement to running time
$O(c^kn)$.  The $O(c^kn\log n)$ time algorithm we describe now is
based on the following observations; (a) In any recursive call of the
algorithm above, the considered graph is an induced subgraph of $G$.
Specifically the considered graph is always $G[C \cup S]$ where $S$ is
a set with at most $2k+3$ vertices and $C$ is a connected component of
$G \setminus S$.  (b) The only computationally hard step, finding the
balanced $S$-separator $X$, is done by dynamic programming over the
tree decomposition $\td_\apx$ of $G$.  The observations (a) and (b)
give some hope that one can reuse the computations done in the dynamic
programming when finding a balanced $S$-separator for $G$ during the
computation of balanced $S$-separators in induced subgraphs of $G$.
This plan can be carried out in a surprisingly clean manner and we now
give a rough sketch of how it can be done.

We start by preprocessing the tree decomposition using an algorithm of
Bodlaender and Hagerup~\cite{BodlaenderH98}.  This algorithm is a
parallel algorithm and here we state its sequential form.
\begin{proposition}[Bodlaender and Hagerup~\cite{BodlaenderH98}]
  \label{proposition:parallel}
  There is an algorithm that, given a tree decomposition of width $k$
  with $O(n)$ nodes of a graph $G$, finds a rooted binary tree
  decomposition of $G$ of width at most $3k+2$ with depth $O(\log n)$
  in $O(kn)$ time.
\end{proposition}
Proposition~\ref{proposition:parallel} lets us assume without loss of
generality that the tree decomposition $\td_\apx$ has depth $O(\log
n)$.

In Section~\ref{section:datastructure} we will describe a data
structure with the following properties.  The data structure takes as
input a graph $G$, an integer $k$ and a tree decomposition $\td_\apx$
of width $O(k)$ and depth $O(\log n)$.  After an initialization step
which takes $O(c^kn)$ time the data structure allows us to do certain
operations and queries.  At any point of time the data structure is in
a certain {\em state}.  The operations allow us to change the state of
the data structure.  Formally, the state of the data structure is a
$3$-tuple $(S,X,F)$ of subsets of $V(G)$ and a vertex $\pin$ called
the ``pin'', with the restriction that $\pin \notin S$.  The initial
state of the data structure is that $S=S_0$, $X = F = \emptyset$, and
$\pin$ is an arbitrary vertex of $G\setminus S_0$.  The data structure
allows operations that change $S$, $X$ or $F$ by inserting/deleting a
specified vertex, and move the pin to a specified vertex in time
$O(c^k\log n)$.

For a fixed state of the data structure, the {\em active component}
$U$ is the component of $G \setminus S$ which contains $\pin$.  The
data structure allows the query \qSsep{} which outputs in time $O(c^k
\log n)$ either an $S$-balanced separator $\hat{X}$ of $G[U \cup S]$
of size at most $k+1$, or $\bot$, which means that $\tw(G[S \cup U]) > k$.

The algorithm of Lemma~\ref{lemma:nlogn-compression} runs the
$O(c^kn^2)$ time algorithm described above, but uses the \emph{data
  structure} to find the balanced $S$-separator in time $O(c^k \log
n)$ instead of doing dynamic programming over $\td_\apx$.  All we need
to make sure is that the $S$ in the state of the data structure is
always equal to the set $S$ for which we want to find the balanced
separator, and that the active component $U$ is set such that $G[U
\cup S]$ is equal to the induced subgraph we are working on.  Since we
always maintain that $|S| \leq 2k+3$ we can change the set $S$ to
anywhere in the graph (and specifically into the correct position) by
doing $k^{O(1)}$ operations taking $O(c^k \log n)$ time each.

At a glance, it looks like that if we assume the data structure as a
black box, this is sufficient to obtain the desired $O(c^k n\log n)$
time algorithm.  However, we haven't even used the sets $X$ and $F$ in
the state of the data structure, or described what they mean! The
reason for this is of course that there is a complication.  In
particular, after the balanced $S$-separator $\hat{X}$ is found ---
how can we recurse into the connected components of $G[S \cup U]
\setminus (S \cup \hat{X})$? We need to move the pin into each of
these components one at a time, but if we want to use $O(c^k \log n)$
time in each recursion step, we cannot afford to spend $O(|S \cup U|)$
time to compute the connected components of $G[S \cup U] \setminus (S
\cup \hat{X})$.  We resolve this issue by pushing the problem into the
data structure, and showing that the appropriate queries can be
implemented there.  This is where the sets $X$ and $F$ in the state of
the data structure come in.

The rôle of $X$ in the data structure is that when queries to the data
structure depending on $X$ are called, $X$ equals the set $\hat{X}$,
i.e., the balanced $S$-separator found by the query \qSsep{}.  The set
$F$ is a set of ``finished pins'' whose intention is the following:
when the algorithm calls itself recursively on a component $U'$ of
$G[S \cup U] \setminus (S \cup \hat{X})$ after it is done with
computing a tree decomposition of $G[U' \cup N(U')]$ with $N(U')$ as
its top bag, it selects an arbitrary vertex of $U'$ and inserts it
into $F$.  

The query \qpin{} finds a new pin $\pin{}'$ in a component $U'$ of $G[S
\cup U] \setminus (S \cup \hat{X})$ that does not contain any vertices
of $F$.  And finally, the query \qnei{} allows us to find the
neighborhood $N(U')$, which in turn allows us to call the algorithm
recursively in order to find a tree decomposition of $G[U' \cup N(U')]$
with $N(U')$ as its top bag.

At this point it is possible to convince oneself that the $O(c^kn^2)$
time algorithm described in the beginning of this section can be
implemented using $O(k^{O(1)})$ calls to the data structure in each
recursive step, thus spending only $O(c^k \log n)$ time in each
recursive step.  Pseudocode for this algorithm can be found in
Algorithm~\ref{alg:nlogn-findtd}.  The recurrence bounding the running
time of the algorithm then becomes
\begin{align*} T(n,k) \leq O(c^{k}\log n) + \sum_{U_i} T(|U_i \cup \hat{X}|,
  k).  \end{align*} Here $U_1, \ldots ,U_q$ are the connected
components of $G[S \cup U] \setminus (S \cup \hat{X})$.  This
recurrence solves to $O(c^k n\log n)$, proving
Lemma~\ref{lemma:nlogn-compression}.  A full proof of
Lemma~\ref{lemma:nlogn-compression} assuming the data structure as a
black box may be found in Section~\ref{sec:nlogncompress}.

\subsection{The data structure}
\label{sec:dsoutline}
We sketch the main ideas in the implementation of the data structure.
The goal is to set up a data structure that takes as input a graph
$G$, an integer $k$ and a tree decomposition $\td_\apx$ of width
$O(k)$ and depth $O(\log n)$, and initializes in time $O(c^kn)$.  The
{\em state} of the data structure is a $4$-tuple $(S,X,F,\pin)$ where
$S$, $X$ and $F$ are vertex sets in $G$ and $\pin \in V(G) \setminus
S$.  The initial state of the data structure is
$(S_0,\emptyset,\emptyset,v)$ where $v$ is an arbitrary vertex in
$G\setminus S_0$.  The data structure should support operations that
insert (delete) a single vertex to (from) $S$, $X$ and $F$, and an
operation to change the pin $\pin$ to a specified vertex.  These
operations should run in time $O(c^k\log n)$.  For a given state of
the data structure, set $U$ to be the component of $G \setminus S$
that contains $\pin$.  The data structure should also support the
following queries in time $O(c^k\log n)$.
\begin{itemize}\setlength\itemsep{-.7mm}
\item \qSsep{}: Assuming that $|S| \leq 2k+3$, return a set $\hat{X}$ of
  size at most $k+1$ such that every component of $G[S \cup U]
  \setminus \hat{X}$ contains at most $\frac{1}{2}|S|$ vertices of
  $S$, or conclude that $\tw(G) > k$.
\item \qpin{}: Return a vertex $\pin{}'$ in a component $U'$ of $G[S \cup
  U] \setminus (S \cup \hat{X})$ that does not contain any vertices of $F$.
\item \qnei{}: Return $N(U)$ if $|N(U)| < 2k+3$ and $\bot$ otherwise.
\end{itemize}

Suppose for now that we want to set up a much simpler data structure.
Here the state is just the set $S$ and the only query we want to
support is \qSsep{} which returns a set $\hat{X}$ such that every
component of $G \setminus (S \cup \hat{X})$ contains at most
$\frac{1}{2}|S|$ vertices of $S$, or conclude that $\tw(G) > k$.  At
our disposal we have the tree decomposition $\td_\apx$ of width $O(k)$
and depth $O(\log n)$.  To set up the data structure we run a standard
dynamic programming algorithm for finding $\hat{X}$ given $S$.  Here
we use Lemma~\ref{lem:balanced-3coloring} and search for a partition
of $V(G)$ into $(M_1,M_2,M_3,X)$ such that $|X| \leq k+1$, there is no
edge between $M_i$ and $M_j$ for $i\neq j$, and $|M_i\cap S|\leq
|S|/2$ for $i=1,2,3$.  This can be done in time $O(c^kk^{O(1)}n)$ and
the tables stored at each node of the tree decomposition have size
$O(c^kk^{O(1)})$.  This finishes the initialization step of the data
structure.  The initialization step took time $O(c^kk^{O(1)}n)$.

We will assume without loss of generality that the top bag of the
decomposition is empty.  The data structure will maintain the
following invariant: after every change has been performed the tables
stored at each node of the tree decomposition correspond to a valid
execution of the dynamic programming algorithm on input $(G,S)$.  If
we are able to maintain this invariant, then answering \qSsep{}
queries is easy: assuming that each cell of the dynamic programming
table also stores solution sets (whose size is at most $k+1$) we can
just output in time $O(k^{O(1)})$ the content of the top bag of the
decomposition!

But how to maintain the invariant and support changes in time
$O(c^k\log n)$? It turns out that this is not too difficult: the
content of the dynamic programming table of a node $t$ in the tree
decomposition depends only on $S$ and the dynamic programming tables
of $t$'s children.  Thus, when the dynamic programming table of the
node $t$ is changed, this will only affect the dynamic programming
tables of the $O(\log n)$ ancestors of $t$.  If the dynamic program is
done carefully, one can ensure that adding or removing a vertex
to/from $S$ will only affect the dynamic programming tables for a
single node $t$ in the decomposition, together with all of its $O(\log
n)$ ancestors.  Performing the changes amounts to recomputing the
dynamic programming tables for these nodes, and this takes time
$O(c^kk^{O(1)}\log n)$.

It should now be plausible that the idea above can be extended to work
also for the more complicated data structure with the more advanced
queries.  Of course there are several technical difficulties, the main
one is how to ensure that the computation is done in the connected
component $U$ of $G \setminus S$ without having to store ``all
possible ways the vertices in a bag could be connected below the
bag''.  We omit the details of how this can be done in this outline.
The full exposition of the data structure can be found in
Section~\ref{section:datastructure}.

\subsection{Approximating treewidth in $O(c^kn \log^{(\depth)}n)$
  time.}
We now sketch how the algorithm of the previous section can be sped
up, at the cost of increasing the approximation ratio from $3$ to $5$.
In particular we give a proof outline for the following theorem.
\begin{theorem}\label{theorem:nlogin}
  For every $\depth\in \mathbb{N}$, there exists an algorithm which,
  given a graph $G$ and an integer $k$, in $O(c^k\cdot n
  \log^{(\depth)} n)$ time for some $c \in \mathbb{N}$ either computes a
  tree decomposition of $G$ of width at most $5k+3$ or correctly
  concludes that $\tw(G) > k$.
\end{theorem}
The algorithm of Theorem~\ref{theorem:nlogn} satisfies the conditions
of Theorem~\ref{theorem:nlogin} for $\depth{} = 1$.  We will show how
one can use the algorithm for $\depth{} = 1$ in order to obtain an
algorithm for $\depth{} = 2$.  In particular we aim at an algorithm
which given a graph $G$ and an integer $k$, in $O(c^k\cdot n \log \log
n)$ time for some $c \in \mathbb{N}$ either computes a tree
decomposition of $G$ of width at most $5k+3$ or correctly concludes
that $\tw(G) > k$.

We inspect the $O(c^kn \log n)$ algorithm for the compression step
described in Section~\ref{sec:proofOfCompress}.  It uses the data
structure of Section~\ref{sec:dsoutline} in order to find balanced
separators in time $O(c^k \log n)$.  The algorithm uses $O(c^k \log
n)$ time on each recursive call regardless of the size of the induced
subgraph of $G$ it is currently working on.  When the subgraph we work
on is big this is very fast.  However, when we get down to induced
subgraphs of size $O(\log \log n)$ the algorithm of Robertson and
Seymour described in Section~\ref{sec:nSquareAppx} would spend $O(c^k
(\log \log n)^2)$ time in each recursive call, while our presumably
fast algorithm still spends $O(c^k \log n)$ time.  This suggests that
there is room for improvement in the recursive calls where the
considered subgraph is very small compared to $n$.

The overall structure of our $O(c^k \log \log n)$ time algorithm is
identical to the structure of the $O(c^k \log n)$ time algorithm of
Theorem~\ref{theorem:nlogn}.  The only modifications happen in the
compression step.  The compression step is also similar to the $O(c^k
\log n)$ algorithm described in Section~\ref{sec:proofOfCompress}, but
with the following caveat.  The data structure query \qpin{} finds the
{\em largest} component where a new pin can be placed, returns a
vertex from this component, and also returns the size of this
component.  If a call of \qpin{} returns that the size of the largest
yet unprocessed component is less than $\log n$ the algorithm does not
process this component, nor any of the other remaining components in
this recursive call.  This ensures that the algorithm is never run on
instances where it is slow.  Of course, if we do not process the small
components we do not find a tree decomposition of them either.  A bit
of inspection reveals that what the algorithm will do is either
conclude that $\tw(G) > k$ or find a tree decomposition of an induced
subgraph of $G'$ of width at most $3k+4$ such that for each connected
component $C_i$ of $G \setminus V(G')$, (a) $|C_i| \leq \log n$, (b)
$|N(C_i)| \leq 2k+3$, and (c) $N(C_i)$ is fully contained in some bag
of the tree decomposition of $G'$.

How much time does it take the algorithm to produce this output? Each
recursive call takes $O(c^k \log n)$ time and adds a bag to the tree
decomposition of $G'$ that contains some vertex which was not yet in
$V(G')$.  Thus the total time of the algorithm is upper bounded by
$O(|V(G')| \cdot c^k\log n)$.  What happens if we run this algorithm,
then run the $O(c^kn \log n)$ time algorithm of
Theorem~\ref{theorem:nlogn} on each of the connected components of $G
\setminus V(G')$? If either of the recursive calls return that the
treewidth of the component is more than $k$ then $\tw(G) > k$.
Otherwise we have a tree decomposition of each of the connected
components with width $3k+4$.  With a little bit of extra care we find
tree decompositions of the same width of $C_i \cup N(C_i)$ for each
component $C_i$, such that the top bag of the decomposition contains
$N(C_i)$.  Then all of these decompositions can be glued together with
the decomposition of $G'$ to yield a decomposition of width $3k+4$ for
the entire graph $G$.

The running time of the above algorithm can be bounded as follows.  It
takes $O(|V(G')| \cdot c^k\log n)$ time to find the partial tree
decomposition of $G'$, and
$$O(\sum_i c^k|C_i| \log |C_i|) \leq O(c^k \log \log n \cdot \sum_i |C_i|) \leq O(c^k n \log \log n)$$ 
time to find the tree decompositions of all the small components.
Thus, if $|V(G')| \leq O(\frac{n}{\log n})$ the running time of the
first part would be $O(c^kn)$ and the total running time would be
$O(c^k n \log \log n)$.

How big can $|V(G')|$ be? In other words, if we inspect the algorithm
described in Section~\ref{sec:nSquareAppx}, how big part of the graph
does the algorithm see before all remaining parts have size less than
$\log n$? The bad news is that the algorithm could see almost the
entire graph.  Specifically if we run the algorithm on a path it could
well be building a tree decomposition of the path by moving along the
path and only terminating when reaching the vertex which is $\log n$
steps away from from the endpoint.  The good news is that the
algorithm of Reed described in Section~\ref{sec:reedAlg} will get down
to components of size $\log n$ after decomposing only $O(\frac{n}{\log
  n})$ vertices of $G$.  The reason is that the algorithm of Reed also
finds balanced separators of the considered subgraph, ensuring that
the size of the considered components drop by a constant factor for
each step down in the recursion tree.

Thus, if we augment the algorithm that finds the tree decomposition of
the subgraph $G'$ such that that it also finds balanced separators of
the active component and adds them to the top bag of the decomposition
before going into recursive calls, this will ensure that $|V(G')| \leq
O(\frac{n}{\log n})$ and that the total running time of the algorithm
described in the paragraphs above will be $O(c^k n \log \log n)$.  The
algorithm of Reed described in Section~\ref{sec:reedAlg} has a worse
approximation ratio than the algorithm of Robertson and Seymour
described in Section~\ref{sec:nSquareAppx}.  The reason is that we
also need to add the balanced separator to the top bag of the
decomposition.  When we augment the algorithm that finds the tree
decomposition of the subgraph $G'$ in a similar manner, the
approximation ratio also gets worse.  If we are careful about how the
change is implemented we can still achieve an algorithm with running
time $O(c^k n \log \log n)$ that meets the specifications of
Theorem~\ref{theorem:nlogin} for $\depth = 2$.

The approach to improve the running time from $O(c^k n \log n)$ to
$O(c^k n \log \log n)$ also works for improving the running time from
$O(c^k\cdot n \log^{(\depth)} n)$ to $O(c^k\cdot n \log^{(\depth{}+1)}
n)$.  Running the algorithm that finds in $O(c^kn)$ time the tree
decomposition of the subgraph $G'$ such that all components of $G
\setminus V(G')$ have size $\log n$ and running the $O(c^k\cdot n
\log^{(\depth)} n)$ time algorithm on each of these components yields
an algorithm with running time $O(c^k\cdot n \log^{(\depth{}+1)} n)$.

In the above discussion we skipped over the following issue.  How can
we compute a small balanced separator for the active component in time
$O(c^k \log n)$? It turns out that also this can be handled by the
data structure.  The main idea here is to consider the dynamic
programming algorithm used in Section~\ref{sec:firstCompress} to find
balanced separators in graphs of bounded treewidth, and show that this
algorithm can be turned into a $O(c^k \log n)$ time data structure
query.  We would like to remark here that the implementation of the
trick from Section~\ref{sec:firstCompress} is significantly more
involved than the other queries: we need to use the approximate tree
decomposition not only for fast dynamic programming computations, but
also to locate the separation $(L_0,X_0,R_0)$ on which the trick is
employed.  A detailed explanation of how this is done can be found at
the end of Section~\ref{sec:queries}.  This completes the proof sketch
of Theorem~\ref{theorem:nlogin}.  A full proof can be found in
Section~\ref{section:logi}.

\subsection{$5$-approximation in $O(c^kn)$ time.}
The algorithm(s) of Theorem~\ref{section:logi} are in fact already
$O(c^kn)$ algorithms unless $n$ is astronomically large compared to
$k$.  If, for example, $n \leq 2^{2^{2^k}}$ then $\log^{(3)}n \leq k$
and so $O(c^kn \log^{(3)}n) \leq O(c^kkn)$.  Thus, to get an algorithm
which runs in $O(c^kn)$ it is sufficient to consider the cases when
$n$ is really, really big compared to $k$.  The recursive scheme of
Section~\ref{sec:recursiveScheme} allows us to only consider the case
where (a) $n$ is really big compared to $k$ and (b) we have at our
disposal a tree decomposition $\td_\apx$ of $G$ of width $O(k)$.

For this case, consider the dynamic programming algorithm of
Bodlaender and Kloks~\cite{BodlaenderK96} that given $G$ and a tree
decomposition $\td_\apx$ of $G$ of width $O(k)$ either computes a tree
decomposition of $G$ of width $k$ or concludes that $\tw(G) > k$ in
time $O(2^{O(k^3)}n)$.  The dynamic programming algorithm can be
turned into a tree automata based
algorithm~\cite{FellowsL89,AbrahamsonF93} with running time
$O(2^{2^{O(k^3)}} + n)$ if one can inspect an arbitrary entry of a
table of size $O(2^{2^{O(k^3)}})$ in constant time.  If $n \geq
\Omega(2^{2^{O(k^3)}})$ then inspecting an arbitrary entry of a table
of size $O(2^{2^{O(k^3)}})$, means inspecting an arbitrary entry of a
table of size $O(n)$, which one can do in constant time in the RAM
model.  Thus, when $n \geq \Omega(2^{2^{O(k^3)}})$ we can find an
optimal tree decomposition in time $O(n)$.  When $n \leq
O(2^{2^{O(k^3)}})$ the $O(c^kn \log^{(3)}n)$ time algorithm of
Theorem~\ref{section:logi} runs in time $O(c^kkn)$.  This concludes
the outline of the proof of Theorem~\ref{thm:mainThm}.  A full
explanation of how to handle the case where $n$ is much bigger than
$k$ can be found in Section~\ref{section:linear}.

\section{Conclusions}
\label{section:conclusions}
In this paper we have presented an algorithm that gives a constant
factor approximation (with a factor 5) of the treewidth of a graph,
which runs in single exponential time in the treewidth and linear in
the number of vertices of the input graph.  Here we give some
consequences of the result, possible improvements and open problems.

\subsection{Consequences, corollaries and future work}
A large number of computational results use the following overall
scheme: first find a tree decomposition of bounded width, and then run
a dynamic programming algorithm on it.  Many of these results use the
linear-time exact algorithm of Bodlaender~\cite{Bodlaender96} for the
first step.
If we aim for algorithms whose running time dependency on treewidth is
single exponential, however, then our algorithm is preferable over the
exact algorithm of Bodlaender~\cite{Bodlaender96}.  Indeed, many
classical problems like \textsc{Dominating Set} and
\textsc{Independent Set} are easily solvable in time $O(c^k\cdot n)$
when a tree decomposition of width $k$ is provided.  Furthermore,
there is an on-going work on finding new dynamic programming routines
with such a running time for problems seemingly not admitting so
robust solutions; the fundamental examples are \textsc{Steiner Tree},
\textsc{Traveling Salesman} and \textsc{Feedback Vertex
  Set}~\cite{BodlaenderCKN12}.  The results of this paper show that
for all these problems we may also claim $O(c^k\cdot n)$ running time
even if the decomposition is not given to us explicitly, as we may
find its constant factor approximation within the same complexity
bound.

Our algorithm is also compatible with the celebrated Courcelle's
theorem~\cite{Courcelle90}, which states that every graph problem
expressible in monadic second-order logic (MSOL) is solvable in time
$f(k,||\varphi||)\cdot n$ when a tree decomposition of width $k$ is
provided, where $\varphi$ is the formula expressing the problem and
$f$ is some function.  Again, the first step of applying the
Courcelle's theorem is usually computing the optimum
tree-decomposition using the linear-time algorithm of
Bodlaender~\cite{Bodlaender96}.  Using the results of this paper, this
step can be substituted with finding an approximate decomposition in
$O(c^k\cdot n)$ time.  For many problems, in the overall running time
analysis we may thus significantly reduce the factor dependent on the
treewidth of the graph, while keeping the linear dependence on $n$ at
the same time.

It seems that the main novel idea of this paper, namely treating the
tree decomposition as a data structure on which logarithmic-time
queries can be implemented, can be similarly applied to all the
problems expressible in MSOL.  Extending our results in this direction
seems like a thrilling perspective for future work.

Concrete examples where the results of this paper can be applied, can
be found also within the framework of \emph{bidimensionality
  theory}~\cite{DemaineFHT05a, DemaineH08}.  In all parameterized
subexponential algorithms obtained within this framework, the
polynomial dependence of $n$ in the running time becomes linear if our
algorithm is used.  For instance, it follows immediately that every
parameterized minor bidimensional problem with parameter $k$, solvable
in time $2^{O(t)}n$ on graphs of treewidth $t$, is solvable in time
$2^{O(\sqrt{k})}n$ on graphs excluding some fixed graph as a minor.

\subsection{Improvements and open problems}
Our result is mainly of theoretical importance due to the large
constant $c$ at the base of the exponent.  One immediate open problem
is to obtain a constant factor approximation algorithm for treewidth
with running time $O(c^k n)$, where $c$ is a \emph{small} constant.

Another open problem is to find more efficient \emph{exact} FPT
algorithms for treewidth.  Bodlaender's algorithm \cite{Bodlaender96}
and the version of Reed and Perkovi{\'{c}} both use $O(k^{O(k^3)}n)$
time; the dominant term being a call to the dynamic programming
algorithm of~\cite{BodlaenderK96}.  In fact, no exact FPT algorithm
for treewidth is known whose running time as a function of the
parameter $k$ is asymptotically smaller than this; testing the
treewidth by verifying the forbidden minors can be expected to be
significantly slower.  Thus, it would be very interesting to have an
exact algorithm for testing if the treewidth of a given graph is at
most $k$ in $2^{o(k^3)} n^{O(1)}$ time.

Currently, the best approximation ratio for treewidth for algorithms
whose running time is polynomial in $n$ and single exponential in the
treewidth is the 3-approximation algorithm from
Section~\ref{section:nlogn}.  What is the best approximation ratio for
treewidth that can be obtained in this running time? Is it possible to
give lower bounds?

\section{An $O(c^k n \log n)$ 3-approximation algorithm for treewidth}
\label{section:nlogn}

In this section, we provide formal details of the proof of
Theorem~\ref{theorem:nlogn}:
\begin{theorem}[Theorem~\ref{theorem:nlogn}, restated]
  There exists an algorithm which, given a graph $G$ and an integer
  $k$, in $O(c^k\cdot n \log n)$ time for some $c \in \mathbb{N}$
  either computes a tree decomposition of $G$ of width at most $3k +
  4$ or correctly concludes that $\tw(G) > k$.
\end{theorem}

In fact, the algorithm that we present, is slightly more general.  The
main procedure, $\alg{1}$, takes as input a connected graph $G$, an
integer $k$, and a subset of vertices $S_0$ such that $|S_0|\leq
2k+3$.  Moreover, we have a guarantee that not only $G$ is connected,
but $G\setminus S_0$ as well.  $\alg{1}$ runs in $O(c^k\cdot n \log
n)$ time for some $c \in \mathbb{N}$ and either concludes that $\tw(G)
> k$, or returns a tree decomposition of $G$ of width $\leq 3k + 4$,
such that $S_0$ is the root bag.  Clearly, to prove
Theorem~\ref{theorem:nlogn}, we can run $\alg{1}$ on every connected
component of $G$ separately using $S_0=\emptyset$.  Note that
computation of the connected components takes $O(|V|+|E|)=O(kn)$ time,
since if $|E|>kn$, then we can safely output that $\tw(G)>k$.

The presented algorithm $\alg{1}$ uses two subroutines.  As described
in Section~\ref{sec:outline}, $\alg{1}$ uses the reduction approach
developed by the first author~\cite{Bodlaender96}; in short words, we
can either apply a reduction step, or find an approximate tree
decomposition of width $O(k)$ on which a compression subroutine
$\compress{1}$ can be employed.  In this compression step we are
either able to find a refined, compressed tree decomposition of width
at most $3k + 4$, or again conclude that $\tw(G)>k$.

The algorithm $\compress{1}$ starts by initializing the data structure
(see Section~\ref{sec:outline} for an intuitive description of the
role of the data structure), and then calls a subroutine $\findTD$.
This subroutine resembles the algorithm of Robertson and Seymour (see
Section~\ref{sec:outline}): it divides the graph using balanced
separators, recurses on the different connected components, and
combines the subtrees obtained for the components into the final tree
decomposition.

%
%
\subsection{The main procedure $\alg{1}$}

Algorithm $\alg{1}$, whose layout is proposed as
Algorithm~\ref{alg:nlogn-alg-1}, runs very similarly to the algorithm
of the first author~\cite{Bodlaender96}; we provide here all the
necessary details for the sake of completeness, but we refer
to~\cite{Bodlaender96} for a wider discussion.

First, we apply Lemma~\ref{lem:prelim:bodlaender} on graph $G$ for
parameter $3k+4$.  We either immediately conclude that $\tw(G)>3k+4$,
find a set of I-simplicial vertices of size at least
$\frac{n}{O(k^6)}$, or a matching of size at least $\frac{n}{O(k^6)}$.
Note that in the application of Lemma~\ref{lem:prelim:bodlaender} we
ignore the fact that some of the vertices are distinguished as $S_0$.

If a matching $M$ of size at least $\frac{n}{O(k^6)}$ is found, we
employ a similar strategy as in~\cite{Bodlaender96}.  We first
contract the matching $M$ to obtain $G'$; note that if $G$ had
treewidth at most $k$ then so does $G'$.  Then we apply $\alg{1}$
recursively to obtain a tree decomposition $\td'$ of $G'$ of width at
most $3k+4$, and having achieved this we decontract the matching $M$
to obtain a tree decomposition $\td$ of $G$ of width at most $6k+9$:
every vertex in the contracted graph is replaced by at most two
vertices before the contraction.  Finally, we call the sub-procedure
$\compress{1}$, which given $G,S_0,k$ and the decomposition $\td$ (of
width $O(k)$), either concludes that $\tw(G)>k$, or provides a tree
decomposition of $G$ of width at most $3k+4$, with $S_0$ as the root
bag.  $\compress{1}$ is given in details in the next section.

In case of obtaining a large set $X$ of I-simplicial vertices, we
proceed similarly as in~\cite{Bodlaender96}.  We compute the improved
graph, remove $X$ from it, apply $\alg{1}$ on $G_I\setminus X$
recursively to obtain its tree decomposition $\td'$ of width at most
$3k+4$, and finally reintroduce the missing vertices of $X$ to obtain
a tree decomposition $\td$ of $G$ of width at most $3k+4$ (recall that
reintroduction can fail, and in this case we may conclude that
$\tw(G)>k$).  Observe that the decomposition $\td$ satisfies all the
needed properties, with the exception that we have not guaranteed that
$S_0$ is the root bag.  However, to find a decomposition that has
$S_0$ as the root bag, we may again make use of the subroutine
$\compress{1}$, running it on input $G, S_0, k$ and the tree
decomposition $\td$.  Lemma~\ref{lem:prelim:bodlaender} ensures that
all the described steps, apart from the recursive calls to $\alg{1}$
and $\compress{1}$, can be performed in $O(k^{O(1)}\cdot n)$ time.
Note that the $I$-simplicial vertices can safely be reintroduced since
we used Lemma~\ref{lem:prelim:bodlaender} for parameter $3k+4$ instead
of $k$.

\begin{algorithm}[h!]

  \KwIn{A connected graph $G$, an integer $k$, and $S_0\subseteq V(G)$ s.t.
    $|S_0|\leq 2k+3$ and $G\setminus S_0$ is connected.} \KwOut{A tree
    decomposition $\td$ of $G$ with $\w(\td) \leq 3k + 4$ and $S_0$ as
    the root bag, or conclusion that $\tw(G) > k$.}  \Indp \BlankLine
  \BlankLine Run algorithm of Lemma~\ref{lem:prelim:bodlaender} for
  parameter $3k+4$ \BlankLine \If{Conclusion that $\tw(G)>3k+4$}{
    \KwRet{$\false$} } \BlankLine
  
  \If{$G$ has a matching $M$ of cardinality at least
    $\frac{n}{O(k^6)}$} {
    Contract $M$ to obtain $G'$.\\
    $\td' \leftarrow \alg{1}(G',k)$ \hspace{1cm}/*\textsf{\footnotesize{ 
    $\w(\td')\leq 3k + 4$} }*/\\
    \eIf{$\td' =
      \false$}{ \KwRet{$\false$} } {
      Decontract the edges of $M$ in $\td'$ to obtain $\td$.\\
      \KwRet{$\compress{1}(G,k,\td)$ } } }
  
  \BlankLine \If{$G$ has a set $X$ of at least $\frac{n}{O(k^6)}$
    I-simplicial vertices}{
    Compute the improved graph $G_I$ and remove $X$ from it.    \\
    $\td'\leftarrow \alg{1}{(G_I\setminus X,k)}$ 
    \hspace{1cm}/*\textsf{\footnotesize{ $\w(\td')\leq 3k + 4$} }*/\\
    \If{$\td' = \false$}{ \KwRet{$\false$} }
    Reintroduce vertices of $X$ to $\td'$ to obtain $\td$.\\
    \eIf{Reintroduction failed}{ \KwRet{$\false$} }{
      \KwRet{$\compress{1}(G,k,\td)$ } } } \Indm
  \caption{$\alg{1}$}
  \label{alg:nlogn-alg-1}
\end{algorithm}

Let us now analyze the running time of the presented algorithm,
provided that the running time of the subroutine $\compress{1}$ is
$O(c^k\cdot n\log n)$ for some $c \in \mathbb{N}$.  Since all the
steps of the algorithm (except for calls to subroutines) can be
performed in $O(k^{O(1)}\cdot n)$ time, the time complexity satisfies
the following recurrence relation:
\begin{eqnarray}\label{eqn:recurrenceOneD}
  T(n) & \leq & O(k^{O(1)}\cdot n) + O(c^k\cdot n \log n) +
  T\left(\left(1- \frac{1}{C\cdot k^6} \right)n\right); 
\end{eqnarray}
Here $C$ is the constant hidden in the $O$-notation in
Lemma~\ref{lem:prelim:bodlaender}.  By unraveling the recurrence into
a geometric series, we obtain that
\begin{eqnarray}\label{eqn:recurrenceTwoD}
  T(n) & \leq & \sum_{i=0}^\infty \left(1 - \frac{1}{Ck^6}\right)^i
  O(k^{O(1)}\cdot n + c^k\cdot  n \log n) \\ 
  \nonumber  & = & Ck^6\cdot O(k^{O(1)}\cdot n + c^k\cdot n \log n) =
  O(c_1^k\cdot n \log n), 
\end{eqnarray}
for some $c_1 > c$.

%
%
\subsection{Compression}\label{sec:nlogncompress}

In this section we provide the details of the implementation of the
subroutine $\compress{1}$.  The main goal is encapsulated in the
following lemma.

\begin{lemma}[Lemma~\ref{lemma:nlogn-compression}, restated]
  There exists an algorithm which on input $G,k,S_0,\td_\apx$, where
  (i) $S_0\subseteq V(G)$, $|S_0|\leq 2k+3$, (ii) $G$ and $G\setminus
  S_0$ are connected, and (iii) $\td_\apx$ is a tree decomposition of
  $G$ of width at most $O(k)$, in $O(c^k\cdot n \log n)$ time for some
  $c \in \mathbb{N}$ either computes a tree decomposition $\td$ of $G$
  with $\w(\td) \leq 3k+4$ and $S_0$ as the root bag, or correctly
  concludes that $\tw(G)>k$.
\end{lemma}

The subroutine's layout is given as
Algorithm~\ref{alg:nlogn-compress}.  Shortly speaking, we first
initialize the data structure $\ds$, with $G, k, S_0, \td$ as input,
and then run a recursive algorithm $\findTD$ that constructs the
decomposition itself given access to the data structure.  The
decomposition is returned by a pointer to the root bag.  The data
structure interface will be explained in the following paragraphs, and
its implementation is given in Section~\ref{section:datastructure}.
We refer to Section~\ref{sec:outline} for a brief, intuitive outline.

\begin{algorithm}
  
  \KwIn{Connected graph $G$, $k \in \mathbb{N}$, a set $S_0$ s.t.
    $|S_0|\leq 2k+3$ and $G\setminus S_0$ is connected, and a tree
    decomposition $\td_\apx$ with $\w(\td_\apx) \leq O(k)$}
  
  \KwOut{Tree decomposition of $G$ of width at most $3k + 4$ with
    $S_0$ as the root bag, or conclusion that $\tw(G) > k$.}  \Indp
  \BlankLine
  Initialize data structure $\ds$ with $G,k,S_0,\td_\apx$\\
  
  \KwRet{$\findTD()$}
  
  \caption{$\compress{1}(G,k,\td)$.}
  
  \label{alg:nlogn-compress}
\end{algorithm}

The initialization of the data structure takes $O(c^k n)$ time (see
Lemma~\ref{lemma:ds:init-time}).  The time complexity of $\findTD$,
given in Section~\ref{section:findTDlogn}, is $O(c^k\cdot n \log n)$.

%
%
%


%
%

\subsection{The recursive algorithm
  $\findTD$}\label{section:findTDlogn}

Subroutine $\findTD$ works on the graph $G$ with two disjoint vertex
sets $S$ and $U$ distinguished.  Intuitively, $S$ is small (of size at
most $2k+3$) and represents the root bag of the tree decomposition
under construction.  $U$ in turn, stands for the part of the graph to
be decomposed below the bag containing $S$, and is always one of the
connected components of $G\setminus S$.  As explained in
Section~\ref{sec:outline}, we cannot afford storing $U$ explicitly.
Instead, we represent $U$ in the data structure by an arbitrary vertex
$\pin$ (called the {\emph{pin}}) belonging to it, and implicitly
define $U$ to be the connected component of $G\setminus S$ that
contains $\pin$.  Formally, behavior of the subroutine $\findTD$ is
encapsulated in the following lemma:

\begin{lemma}
  \label{lemma:nlogn-s-root-bag}
  There exists an algorithm that, given access to the data structure
  $\ds$ in a state such that $|S| \leq 2k+3$, computes a tree
  decomposition $\td$ of $G[U \cup S]$ of width $\leq 3k + 4$ with $S$
  as a root bag, or correctly reports that that $\tw(G[U \cup S]) >
  k$.  If the algorithm is run on $S=\emptyset$ and $U=V(G)$, then its
  running time is $O(c^k\cdot n\log n)$ for some $c \in \mathbb{N}$.
\end{lemma}

The data structure is initialized with $S=S_0$ and $\pin$ set to an
arbitrary vertex of $G\setminus S_0$; as we have assumed that
$G\setminus S_0$ is connected, this gives $U=V(G)\setminus S_0$ after
initialization.  Therefore, Lemma~\ref{lemma:nlogn-s-root-bag}
immediately yields Lemma~\ref{lemma:nlogn-compression}.

\paragraph{A gentle introduction to the data structure}
Before we proceed to the implementation of the subroutine $\findTD$,
we give a quick description of the interface of the data structure
$\ds$: what kind of queries and updates it supports, and what is the
running time of their execution.  The details of the data structure
implementation will be given in Section~\ref{section:datastructure}.

The state of the data structure is, in addition to $G,k,\td$, three
subsets of vertices, $S$, $X$ and $F$, and the pin $\pin$ with the
restriction that $\pin \notin S$.  $S$ and $\pin$ uniquely imply the
set $U$, defined as the connected component of $G\setminus S$ that
contains $\pin$.  The intuition behind these sets and the pin is the
following:

\begin{itemize}
\item $S$ is the set that will serve as a root bag for some subtree,
\item $\pin$ is a vertex which indicates the current active component,
\item $U$ is the current active component, the connected component of
  $G\setminus S$ containing $\pin$,
\item $X$ is a balanced $S$-separator (of $G[S \cup U]$) and
\item $F$ is a set of vertices marking the connected components of
  $G[S \cup U]\setminus (S\cup X)$ as ``finished''.
\end{itemize}

The construction of the data structure $\ds$ is heavily based on the
fact that we are provided with some tree decomposition of width
$O(k)$.  Given this tree decomposition, the data structure can be
initialized in $O(c^k\cdot n)$ time for some $c \in \mathbb{N}$.  At
the moment of initialization we set $S=X=F=\emptyset$ and $\pin$ to be
an arbitrary vertex of $G$.  During the run of the algorithm, the
following updates can be performed on the data structure:
\begin{itemize}
\item insert/remove a vertex to/from $S$, $X$, or $F$;
\item mark/unmark a vertex as a pin $\pin$.
\end{itemize}
All of these updates will be performed in $O(c^k\cdot \log n)$ time
for some $c \in \mathbb{N}$.

The data structure provides a number of queries that are used in the
subroutine $\findTD$.  The running time of each query is $O(c^k \cdot
\log n)$ for some $c \in \mathbb{N}$, and in many cases it is actually
much smaller.  We find it more convenient to explain the needed
queries while describing the algorithm itself.



\paragraph{Implementation of $\findTD$}
The pseudocode of the algorithm $\findTD$ is given as
Algorithm~\ref{alg:nlogn-findtd}.  Its correctness is proven as
Claim~\ref{claim:findtd-nlogn-correct}, and its time complexity is
proven as Claim~\ref{claim:findtd-nlogn-complexity}.  The subroutine
is provided with the data structure $\ds$, and the following
invariants hold at each time the subroutine is called and exited:
\begin{itemize}
\item $S \subseteq V(G)$, $|S| \leq 2k+3$,
\item $\pin$ exists, is unique and $\pin \notin S$,
\item $X = F = \emptyset$ and
\item the state of the data structure is the same on exit as it was
  when the function was called.
\end{itemize}
The latter means that when we return, be it a tree decomposition or
$\false$, the algorithm that called $\findTD$ will have $S$, $X$, $F$
and $\pin$ as they were before the call.

We now describe the consecutive steps of the algorithm $\findTD$; the
reader is encouraged to follow these steps in the pseudocode, in order
to be convinced that all the crucial, potentially expensive
computations are performed by calls to the data structure.

First we apply query \qSsep, which either finds a
$\frac{1}{2}$-balanced $S$-separator in $G[S\cup U]$ of size at most
$k+1$, or concludes that $\tw(G) > k$.  The running time of this query
is $k^{O(1)}$.  If no such separator can be found, by
Lemma~\ref{lemma:halfhalf} we infer that $\tw(G[S\cup U])>k$ and we
can terminate the procedure.  Otherwise we are provided with such a
separator $\sep$, which we add to $X$ in the data structure.
Moreover, for a technical reason, we also add the pin $\pin$ to $\sep$
(and thus also to $X$), so we end up with having $|\sep| \leq k + 2$.

The next step is a loop through the connected components of $G[S\cup
U]\setminus (S\cup \sep)$.  This part is performed using the query
\qpin.  Query \qpin, which runs in constant time, either finds an
arbitrary vertex $u$ of a connected component of $G[S\cup U]\setminus
(S\cup X)$ that does not contain any vertex from $F$, or concludes
that each of these components contains some vertex of $F$.  After
finding $u$, we mark $u$ by putting it to $F$ and proceed further,
until all the components are marked.  Having achieved this, we have
obtained a list $\pins$, containing exactly one vertex from each
connected component of $G[S\cup U]\setminus (S\cup \sep)$.  We remove
all the vertices on this list from $F$, thus making $F$ again empty.

It is worth mentioning that the query \qpin{} not only returns some
vertex $u$ of a connected component of $G[S\cup U]\setminus (S\cup
\sep)$ that does not contain any vertex from $F$, but also provides
the size of this component as the second coordinate of the return
value.  Moreover, the components are being found in decreasing order
with respect to sizes.  In this algorithm we do not exploit this
property, but it will be crucial for the linear-time algorithm.

The set $X$ will no longer be used, so we remove all the vertices of
$\sep$ from $X$, thus making it again empty.  On the other hand, we
add all the vertices from $\sep$ to $S$.  The new set $S$ obtained in
this manner will constitute the new bag, of size at most
$|S|+|\sep|\leq 3k+5$.  We are left with computing the tree
decompositions for the connected components below this bag, which are
pinpointed by vertices stored in the list $\pins$.

We iterate through the list $\pins$ and process the components one by
one.  For each vertex $u\in \pins$, we set $u$ as the new pin by
unmarking the old one and marking $u$.  Note that the set $U$ gets
redefined and now is the connected component containing considered
$u$.  First, we find the neighborhood of $U$ in $S$.  This is done
using query \qnei, which in $O(k)$ time returns either this
neighborhood, or concludes that its cardinality is larger than $2k+3$.
However, as $X$ was a $\frac{1}{2}$-balanced $S$-separator, it follows
that this neighborhood will always be of size at most $2k+3$ (a formal
argument is contained in the proof of correctness).  We continue with
$S \cap N(U)$ as our new $S$ and recursively call $\findTD$ in order
to decompose the connected component under consideration, with its
neighborhood in $S$ as the root bag of the constructed tree
decomposition.  $\findTD$ either provides a decomposition by returning
a pointer to its root bag, or concludes that no decomposition can be
found.  If the latter is the case, we may terminate the algorithm
providing a negative answer.

After all the connected components are processed, we merge the
obtained tree decompositions.  For this, we use the function
$\build(S,X,C)$ which, given sets of vertices $S$ and $X$ and a set of
pointers $C$, constructs two bags $B = S$ and $B' = S \cup X$, makes
$C$ the children of $B'$, $B'$ the child of $B$ and returns a pointer
to $B$.  This pointer may be returned from the whole subroutine, after
doing a clean-up of the data structure.

\begin{algorithm}[h!]

  \KwData{Data structure $\ds$}
  \KwOut{Tree decomposition of width at most $3k+4$ of $G[S\cup U]$ with $S$ as root bag or conclusion that $\tw(G) > k$.}
\BlankLine
\Indp
  $\oldS \leftarrow \ds.\getS()$\\
  $\oldPin \leftarrow \ds.\getPin()$\\
  $\sep \leftarrow \ds.\qSsep()$\\
  \If{$\sep = \false$}{
    \KwRet $\false$ \hspace{1cm}/*\textsf{\footnotesize{ safe to return: the 
    state not changed }}*/\\
  }
  $\ds.\insertX(\sep)$\\
  $\ds.\insertX(\pin)$\\
  $\pins \leftarrow \emptyset$\\
  \While{$(u,l) \leftarrow \ds.\qpin() \neq \false$}{
    $\pins.append(u)$\\
    $\ds.\insertF(u)$
  }
  $\ds.\clearX()$\\
  $\ds.\clearF()$\\
  $\ds.\insertS(\sep)$\\
  $\bags \leftarrow \emptyset$\\
  \For{$u \in \pins$}{
    $\ds.\setPin(u)$\\
    $\bags.append(\ds.\qnei())$\\
  }
  $\children \leftarrow \emptyset$\\
  \For{$u,b \in \pins,\bags$}{
    $\ds.\setPin(u)$\\
    $\ds.\clearS()$\\
    $\ds.\insertS(b)$\\
    $\children.append(\findTD())$
  }
  $\ds.\clearS()$\\
  $\ds.\insertS(\oldS)$\\
  $\ds.\setPin(\oldPin)$\\
  \If{$\false \in \children$}{
    \KwRet $\false$ \hspace{1cm}/*\textsf{\footnotesize{ postponed because of 
    rollback of $S$ and $\pin$}} */\\
  }  
  \KwRet $\build(\oldS, \sep, \children)$
  \caption{$\findTD$}
  \label{alg:nlogn-findtd}
\end{algorithm}

\paragraph{Invariants}
Now we show that the stated invariants indeed hold.  Initially $S = X
= F = \emptyset$ and $\pin \in V(G)$, so clearly the invariants are
satisfied.  If no $S$-separator is found, the algorithm returns
without changing the data structure and hence the invariants trivially
hold in this case.  Since both $X$ and $F$ are empty or cleared before
return or recursing, $X = F = \emptyset$ holds.  Furthermore, as $S$
is reset to $\oldS$ (consult Algorithm~\ref{alg:nlogn-findtd} for the
variable names used) and the pin to $\oldPin$ before returning, it
follows that the state of the data structure is reverted upon
returning.

The size of $S = \emptyset$ is trivially less than $2k+3$ when initialized.
Assume that for some call to $\findTD$ we have that $|\oldS| \leq
2k+3$.  When recursing, $S$ is the neighborhood of some component $C$
of $G[\oldS\cup U]\setminus (\oldS\cup \sep)$ (note that we refer to
$U$ before resetting the pin).  This component is contained in some
component $C'$ of $G[\oldS \cup U]\setminus \sep$, and all the
vertices of $\oldS$ adjacent to $C$ must be contained in $C'$.  Since
$\sep$ is a $\frac{1}{2}$-balanced $\oldS$-separator, we know that
$C'$ contains at most $\frac{1}{2}|\oldS|$ vertices of $\oldS$.
Hence, when recursing we have that $|S| \leq \frac{1}{2}|\oldS| +
|\sep| = \frac{1}{2}(2k+3)+k+2 = 2k + \frac{7}{2}$ and, since $|S|$ is
an integer, it follows that $|S| \leq 2k+3$.

Finally, we argue that the pin $\pin$ is never contained in $S$.  When
we obtain the elements of $\pins$ (returned by query \qpin) we know
that $X = \sep$ and the data structure guarantees that the pins will
be from $G[\oldS\cup U] \setminus (\oldS \cup \sep)$.  When recursing,
$S = b \subseteq (\oldS \cup \sep)$ and $\pi \in \pins$, so it follows
that $\pi \notin S$.  Assuming $\pi \notin \oldS$, it follows that
$\pi$ is not in $S$ when returning, and our argument is complete.
From here on we will safely assume that the invariants indeed hold.

\paragraph{Correctness}
\begin{claim}
  \label{claim:findtd-nlogn-correct}
  The algorithm $\findTD$ is correct, that is
  \begin{enumerate}[(a)]
  \item\label{item:1:claim:findtd-nlogn-correct} if $\tw(G) \leq k$,
    $\findTD$ returns a valid tree decomposition of $G[S\cup U]$ of
    width at most $3k+4$ and
  \item\label{item:2:claim:findtd-nlogn-correct} if $\findTD$ returns
    $\false$ then $\tw(G) > k$.
  \end{enumerate}
\end{claim}

\begin{proof}
  We start by proving (\ref{item:2:claim:findtd-nlogn-correct}).
  Suppose the algorithm returns $\false$.  This happens when at some
  point we are unable to find a balanced $S$-separator for an induced
  subgraph $G' = G[S \cup U]$.  By Lemma~\ref{lemma:halfhalf} the
  treewidth of $G'$ is more than $k$.  Hence $\tw(G) > k$ as well.
  
  To show (\ref{item:1:claim:findtd-nlogn-correct}) we proceed by
  induction.  In the induction we prove that the algorithm creates a
  tree decomposition, and we therefore argue that the necessary
  conditions are satisfied, namely
  \begin{itemize}
  \item the bags have size at most $3k+5$,
  \item every vertex and every edge is contained in some bag and
  \item for each $v \in V(G)$ the subtree of bags containing $v$ is
    connected.
  \end{itemize}
  The base case is at the leaf of the obtained tree decomposition,
  namely when $U \subseteq S \cup \sep$.  Then we return a tree
  decomposition containing two bags, $B$ and $B'$ where $B = \{S\}$
  and $B' = \{S \cup \sep\}$.  Clearly, every edge and every vertex of
  $G[S \cup U] = G[S \cup \sep]$ is contained in the tree
  decomposition.  Furthermore, since the tree has size two, the
  connectivity requirement holds and finally, since $|S| \leq 2k+3$
  (invariant) and $\sep \leq k+2$ it follows that $|S \cup \sep | \leq
  3k+5$.  Note that due to the definition of the base case, the
  algorithm will find no pins and hence it will not recurse further.
  
  %
  %
  The induction step is as follows.  Since $U \nsubseteq S \cup \sep$,
  the algorithm have found some pins $\pi_1, \pi_2, \dots, \pi_d$ and
  the corresponding components $C_1, C_2, \dots, C_d$ in $G[S \cup
  U]\setminus (S\cup \sep)$.  Let $N_i = N(C_i) \cap (S \cup \sep)$.
  By the induction hypothesis the algorithm gives us valid tree
  decompositions $\td_i$ of $G[N_i \cup C_i]$.  Note that the root bag
  of $\td_i$ consists of the vertices in $N_i$.  By the same argument
  as for the base case, the two bags $B = S$ and $B' = S \cup \sep$
  that we construct have appropriate sizes.
  
  Let $v$ be an arbitrary vertex of $S \cup U$.  If $v \in S \cup \sep$, then
  it is contained in $B'$.  Otherwise there exists a unique $i$ such
  that $v \in C_i$.  It then follows from the induction hypothesis
  that $v$ is contained in some bag of $\td_i$.
  
  It remains to show that the edge property and the connectivity
  property hold.  Let $uv$ be an arbitrary edge of $G[S \cup U]$.  If
  $u$ and $v$ both are in $S \cup \sep$, then the edge is contained in
  $B'$.  Otherwise, assume without loss of generality that $u$ is in
  some component $C_i$.  Then $u$ and $v$ are in $N_i \cup C_i$ and
  hence they are in some bag of $\td_i$ by the induction hypothesis.
  
  Finally, for the connectivity property, let $v$ be some vertex in $S
  \cup U$.  If $v \notin S \cup \sep$, then there is a unique $i$ such
  that $v \in C_i$, hence we can apply the induction hypothesis.  So
  assume that $v \in S \cup \sep = B'$.  Let $A$ be some bag of $\td$
  containing $v$.  We will complete the proof by proving that there is
  a path of bags containing $v$ from $A$ to $B'$.  If $A$ is $B$ or
  $B'$, then this follows directly from the construction.  Otherwise
  there exists a unique $i$ such that $A$ is a bag in $\td_i$.
  Observe that $v$ is in $N_i$ as it is in $S \cup \sep$.  By the
  induction hypothesis the bags containing $v$ in $\td_i$ are
  connected and hence there is a path of bags containing $v$ from $A$
  to the root bag $R_i$ of $\td_i$.  By construction $B'$ contains $v$
  and the bags $B'$ and $R_i$ are adjacent.  Hence there is a path of
  bags containing $v$ from $A$ to $B'$ and as $A$ was arbitrary
  chosen, this proves that the bags containing $v$ form a connected
  subtree of the decomposition.  This concludes the proof of
  Claim~\ref{claim:findtd-nlogn-correct}.
\end{proof}

\paragraph{Complexity}
\begin{claim}
  \label{claim:findtd-nlogn-complexity}
  The invocation of $\findTD$ in the algorithm $\compress{1}$ runs in
  $O(c^k\cdot n \log n)$ time for some $c \in \mathbb{N}$.
\end{claim}
\begin{proof}
  We will prove the complexity of the algorithm by first arguing that
  the constructed tree decomposition contains at most $2n$ bags.  Then
  we will partition the used running time between the bags, charging
  each bag with at most $O(c^k\cdot \log n)$ time.  It then follows
  that that $\findTD$ runs in $O(c^k\cdot n \log n)$ time.
  
  
  To bound the number of bags we simply observe that at each recursion
  step, we add the previous pin to $S$ and create two bags.  Since a
  vertex can only be added to $S$ one time during the entire process,
  at most $2n$ bags are created.
  
  \newcommand{\reccall}{\mathcal{C}}
  
  It remains to charge the bags.  For a call $\reccall$ to $\findTD$,
  let $B$ and $B'$ be as previously and let $R_i$ be the root bag of
  $\td_i$.  We will charge $B'$ and $R_1, \dots, R_d$ for the time
  spent on $\reccall$.  Notice that as $R_i$ will correspond to $B$ in
  the next recursion step, each bag will only be charged by one call
  to $\findTD$.  We charge $B'$ with everything in $\reccall$ not
  executed in the two loops iterating through the components, plus
  with the last call to $\qpin$ that returned $\bot$.  Since every
  update and query in the data structure is executed in $O(c^k\cdot
  \log n)$ time, and there is a constant number of queries charged to
  $B'$, it follows that $B'$ is charged with $O(c^k\cdot \log n)$
  time.  For each iteration in one of the loops we consider the
  corresponding $\pi_i$ and charge the bag $R_i$ with the time spent
  on this iteration.  As all the operations in the loops can be
  performed in $O(c^k\cdot \log n)$ time, each $R_i$ is charged with
  at most $O(c^k\cdot \log n)$ time.  Since our tree decomposition has
  at most $2n$ bags and each is charged with at most $O(c^k\cdot \log
  n)$ time, it follows that $\findTD$ runs in $O(c^k\cdot n \log n)$
  time and the proof is complete.
\end{proof}

\newcommand{\forg}{\texttt{forgotten}}

\section{$O(c^k n \log^{(\depth)} n)$ 5-approximation algorithm for
  treewidth}
\label{section:logi}

In this section we provide formal details of the proof of
Theorem~\ref{theorem:nlogin}
\begin{theorem}[Theorem~\ref{theorem:nlogin}, restated]
  For every $\depth\in \mathbb{N}$, there exists an algorithm which,
  given a graph $G$ and an integer $k$, in $O(c^k\cdot n
  \log^{(\depth)} n)$ time for some $c \in \mathbb{N}$ either computes
  a tree decomposition of $G$ of width at most $5k+4$ or correctly
  concludes that $\tw(G) > k$.
\end{theorem}

In the proof we give a sequence of algorithms $\alg{\depth}$ for
$\depth=2,3,\ldots$; $\alg{1}$ has been already presented in the
previous section.  Each $\alg{\depth}$ in fact solves a slightly more
general problem than stated in Theorem~\ref{theorem:nlogin}, in the
same manner as $\alg{1}$ solved a more general problem than the one
stated in Theorem~\ref{theorem:nlogn}.  Namely, every algorithm
$\alg{\depth}$ gets as input a connected graph $G$, an integer $k$ and
a subset of vertices $S_0$ such that $|S_0|\leq 4k+3$ and $G\setminus
S_0$ is connected, and either concludes that $\tw(G)>k$ or constructs
a tree decomposition of width at most $5k+4$ with $S_0$ as the root
bag.  The running time of $\alg{\depth}$ is $O(c^k\cdot n
\log^{(\depth)} n)$ for some $c \in \mathbb{N}$; hence, in order to
prove Theorem~\ref{theorem:nlogin} we can again apply $\alg{\depth}$
to every connected component of $G$ separately, using $S_0=\emptyset$.

The algorithms $\alg{\depth}$ are constructed inductively; by that we
mean that $\alg{\depth}$ will call $\alg{\depth-1}$, which again will
call $\alg{\depth-2}$, and all the way until $\alg{1}$, which was
given in the previous section.  Let us remark that a closer
examination of our algorithms in fact shows that the constants $c$ in
the bases of the exponents of consecutive algorithms can be bounded by
some universal constant.  However, of course the constant factor
hidden in the $O$-notation depends on $\depth$.

In the following we present a quick outline of what will be given in
this section.  For $\depth = 1$, we refer to the previous section, and
for $\depth > 1$, $\alg{\depth}$ and $\compress{\depth}$ are described
in this section, in addition to the subroutine $\findPTD$.


\begin{itemize}
\item $\alg{\depth}$ takes as input a graph $G$, an integer $k$ and a
  vertex set $S_0$ with similar assumptions as in the previous
  section, and returns a tree decomposition $\td$ of $G$ of width at
  most $5k+4$ with $S_0$ as the root bag.  The algorithm is almost
  exactly as $\alg{1}$ given as Algorithm~\ref{alg:nlogn-alg-1},
  except that it uses $\compress{\depth}$ for the compression step.
\item $\compress{\depth}$ is an advanced version of $\compress{1}$
  (see Algorithm~\ref{alg:nlogn-compress}), it allows $S_0$ to be of
  size up to $4k+3$ and gives a tree decomposition of width at most
  $5k+4$ in time $O(c^k\cdot n\log^{(\alpha)}n)$.  It starts by
  initializing the data structure, and then it calls $\findPTD$, which
  returns a tree decomposition $\td'$ of an induced subgraph $G'
  \subseteq G$.  The properties of $G'$ and $\td'$ are as follows.
  All the connected components $C_1, \dots C_p$ of $G\setminus V(G')$
  are of size less than $\log n$.  Furthermore, for every connected
  component $C_j$, the neighborhood $N(C_j)$ in $G$ is contained in a
  bag of $\td'$.  Intuitively, this ensures that we are able to
  construct a tree decomposition of $C_j$ and attach it to $\td'$
  without blowing up the width of $\td'$.  More precisely, for every
  connected component $C_j$, the algorithm constructs the induced
  subgraph $G_j = G[C_j \cup N(C_j)]$ and calls $\alg{\depth-1}$ on
  $G_j$, $k$, and $N(C_j)$. The size of $N(C_j)$ will be bounded by
  $4k+3$, making the recursion valid with respect to the invariants of
  $\alg{\alpha}$. If this call returned a tree decomposition $\td_j$
  with a root bag $N(C_j)$, we can conveniently attach $\td_j$ to
  $\td'$; otherwise we conclude that $\tw(G[C_j \cup N(C_j)])>k$ so
  $\tw(G)>k$ as well.
\item $\findPTD$ differs from $\findTD$ in two ways.  First, we use
  the fact that when enumerating the components separated by the
  separator using query $\qpin$, these components are identified in
  the descending order of cardinalities.  We continue the construction
  of partial tree decomposition in the identified components only as
  long as they are of size at least $\log n$, and we terminate the
  enumeration when we encounter the first smaller component.  It
  follows that all the remaining components are smaller then $\log n$;
  these remainders are exactly the components $C_1, \dots C_p$ that
  are left not decomposed by $\alg{\depth}$, and on which
  $\alg{\depth-1}$ is run.
  
  The other difference is that the data structure has a new
  \emph{flag}, $\whatsep$, which is set to either $u$ or $s$ and is
  alternated between calls.  If $\whatsep = s$, we use the same type
  of separator as $\findTD$ did, namely $\qSsep$, but if $\whatsep =
  u$, then we use the (new) query $\qUsep$.  Query $\qUsep$, instead
  of giving a balanced $S$-separator, provides a
  $\frac{8}{9}$-balanced $U$-separator, that is, a separator that
  splits the whole set $U$ of vertices to be decomposed in a balanced
  way.  Using the fact that on every second level of the decomposition
  procedure the whole set of available vertices shrinks by a constant
  fraction, we may for example observe that the resulting partial tree
  decomposition will be of logarithmic depth.  More importantly, it
  may be shown that the total number of constructed bags is at most
  $O(n / \log n)$ and hence we can spend $O(c^k\cdot \log n)$ time
  constructing each bag and still obtain running time linear in $n$.
\end{itemize}

In all the algorithms that follow we assume that the cardinality of
the edge set is at most $k$ times the cardinality of the vertex set,
because otherwise we may immediately conclude that treewidth of the
graph under consideration is larger than $k$ and terminate the
algorithm.

\subsection{The main procedure $\alg{\depth}$}

The procedure $\alg{\depth}$ works exactly as $\alg{1}$, with the
exception that it applies Lemma~\ref{lem:prelim:bodlaender} for
parameter $5k+4$ instead of $3k+4$, and calls recursively
$\alg{\depth}$ and $\compress{\depth}$ instead of $\alg{1}$ and
$\compress{1}$.  The running time analysis is exactly the same, hence
we omit it here.

\subsection{Compression algorithm}

The following lemma explains the behavior of the compression algorithm
$\compress{\depth}$.

\begin{lemma}
  \label{lemma:nlogin-compression}
  For every integer $\depth\geq 1$ there exists an algorithm, which on
  input $G,k,S_0,\td_\apx$, where (i) $S_0\subseteq V(G)$, $|S_0|\leq
  4k+3$, (ii) $G$ and $G\setminus S_0$ are connected, and (iii)
  $\td_\apx$ is a tree decomposition of $G$ of width at most $O(k)$,
  in $O(c^k\cdot n \log^{(\depth)} n)$ time for some $c \in
  \mathbb{N}$ either computes a tree decomposition $\td$ of $G$ with
  $\w(\td) \leq 5k+4$ and $S_0$ as the root bag, or correctly
  concludes that $\tw(G)>k$.
\end{lemma}

The outline of the algorithm $\compress{\depth}$ for $\depth > 1$ is
given as Algorithm~\ref{alg:nlogin-compress-i}.  Having initialized
the data structure using $\td_\apx$, the algorithm asks $\findPTD$ for
a partial tree decomposition $\td'$, and then the goal is to decompose
the remaining small components and attach the resulting tree
decompositions in appropriate places of $\td'$.

First we traverse $\td'$ in linear time and store information on where
each vertex appearing in $\td'$ is forgotten in $\td'$.  More
precisely, we compute a map $\forg: V(G) \to V(\td')\cup \{\bot\}$,
where for every vertex $v$ of $G$ we either store $\bot$ if it is not
contained in $\td'$, or we remember the top-most bag $B_i$ of $\td'$
such that $v\in B_i$ (the connectivity requirement of the tree
decomposition ensures that such $B_i$ exists and is unique).  The map
$\forg$ may be very easily computed via a DFS traversal of the tree
decomposition: when accessing a child node $i$ from a parent $i'$, we
put $\forg(v)=i$ for each $v\in B_i\setminus B_{i'}$.  Moreover, for
every $v\in B_r$, where $r$ is the root node, we put $\forg(v)=r$.
Clearly, all the vertices not assigned a value in $\forg$ in this
manner, are not contained in any bag of $\td'$, and we put value
$\bot$ for them.  Let $W$ be the set of vertices contained in $\td'$,
i.e., $W=\bigcup_{i\in V(\td')} B_i$.

Before we continue, let us show how the map $\forg$ will be used.
Suppose that we have some set $Y\subseteq W$, and we have a guarantee
that there exists a node $i$ of $\td'$ such that $B_i$ contains the
whole $Y$.  We claim the following: then one of the bags associated
with $\forg(v)$ for $v\in Y$ contains the whole $Y$.  Indeed, take the
path from $i$ to the root of the tree decomposition $\td'$, and
consider the last node $i'$ of this path whose bag contains the whole
$Y$.  It follows that $i'=\forg(v)$ for some $v\in Y$ and $Y\subseteq
B_{i'}$, so the claim follows.  Hence, we can locate the bag
containing $Y$ in $O(k^{O(1)}\cdot |Y|)$ time by testing each of $|Y|$
candidate nodes $\forg(v)$ for $v\in Y$.

The next step of the algorithm is locating the vertices which has not
been accounted for, i.e., those assigned $\bot$ by $\forg$.  The
reason each of these vertices has not been put into the tree
decomposition, is precisely because the size of its connected
component $C$ of $G\setminus W$, is smaller than $\log n$.  The
neighborhood of this component in $G$ is $N(C)$, and this neighborhood
is guaranteed to be of size at most $4k+3$ and contained in some bag
of $\td'$ (a formal proof of this fact will be given when presenting
the algorithm $\findPTD$, i.e., in Lemma~\ref{lemma:findPTD}).

Let $C_1,C_2,\ldots,C_p$ be all the connected components of $G\setminus W$,
i.e., the connected components \emph{outside} the obtained partial
tree decomposition $\td'$.  To complete the partial tree decomposition
into a tree decomposition, for every connected component $C_j$, we
construct a graph $G_j=G[C_j \cup N(C_j)]$ that we then aim to
decompose.  These graphs may be easily identified and constructed in
$O(k^{O(1)}\cdot n)$ time using depth-first search as follows.

We iterate through the vertices of $G$, and for each vertex $v$ such
that $\forg(v)=\bot$ and $v$ was not visited yet, we apply a
depth-first search on $v$ to identify its component $C$.  During this
depth-first search procedure, we terminate searching and return from a
recursive call whenever we encounter a vertex from $W$.  In this
manner we identify the whole component $C$, and all the visited
vertices of $W$ constitute exactly $N(C)$.  Moreover, the edges
traversed while searching are exactly those inside $C$ or between $C$
and $N(C)$.  To finish the construction of $G_j$, it remains to
identify edges between vertices of $N(C)$.  Recall that we have a
guarantee that $N(C)\subseteq W$ and $N(C)$ is contained in some bag
of $\td'$.  Using the map $\forg$ we can locate some such bag in
$O(k^{O(1)})$ time, and in $O(k^{O(1)})$ time check which vertices of
$N(C)$ are adjacent in it, thus finishing the construction of $G_j$.
Observe that during the presented procedure we traverse each edge of
the graph at most once, and for each of at most $n$ components $C$ we
spend $O(k^{O(1)})$ time on examination of $N(C)$.  It follows that
the total running time is $O(k^{O(1)}\cdot n)$.


Having constructed $G_j$, we run the algorithm $\alg{\depth-1}$ on
$G_j$ using $S_0=N(C_j)$.  Note that in this manner we have that both
$G_j$ and $G_j\setminus S_0$ are connected, which are requirements of
the algorithm $\alg{\depth-1}$.  If $\alg{\depth-1}$ concluded that
$\tw(G_j)>k$, then we can consequently answer that $\tw(G)>k$ since
$G_j$ is an induced subgraph of $G$.  On the other hand, if
$\alg{\depth-1}$ provided us with a tree decomposition $\td_j$ of
$G_j$ having $N(C_j)$ as the root bag, then we may simply attach this
root bag as a child of the bag of $\td'$ that contains the whole
$N(C_j)$.  Any such bag can be again located in $O(k^{O(1)})$ time
using the map $\forg$.

\begin{algorithm}
  \KwIn{Connected graph $G$, $k \in \mathbb{N}$, a set $S_0$ s.t.
    $|S_0|\leq 4k+3$ and $G\setminus S_0$ is connected, and a tree decomposition
    $\td_\apx$ with $\w(\td_\apx) \leq O(k)$}
  
  \KwOut{Tree decomposition of $G$ of width at most $5k + 4$ with
    $S_0$ as the root bag, or conclusion that $\tw(G) > k$.}
  
  \Indp \BlankLine
  Initialize data structure $\ds$ with $G,k,S_0,\td_\apx$\\
  
  $\td' \leftarrow \findPTD()$\\
  \If{$\td' = \false$}{ \KwRet{$\false$} }
  
  Create the map $\forg: V(G) \to  V(\td')$ using a DFS traversal of $\td'$\\
  
  Construct components $C_1,C_2,\ldots,C_p$ of $G\setminus W$, and graphs $G_j=G[C_j\cup N(C_j)]$ for $j=1,2,\ldots,p$\\
  
  \BlankLine \For{$j=1,2,\ldots,p$}{
    
    $\td_{j} \leftarrow \alg{\depth - 1}$ on $G_j, k, N(C_j)$\\
    
    \If{$\td_{j} = \false$}{
      
      \KwRet{$\false$}
      
    }
    
    Locate a node $i$ of $\td'$ s.t.  $N(C_j)\subseteq B_i$, by checking
    $\forg(v)$ for each $v\in N(C_j)$\\
    Attach the root of $\td_{j}$ as a child of $i$ }
  
  \KwRet{$\td'$}
  
  \caption{$\compress{\depth}$}
  \label{alg:nlogin-compress-i}
\end{algorithm}

\subsubsection{Correctness and complexity}
In this section we prove Lemma~\ref{lemma:nlogin-compression} and
Theorem~\ref{theorem:nlogin}, and we proceed by induction on $\depth$.
To this end we will assume the correctness of
Lemma~\ref{lemma:findPTD}, which will be proved later, and which
describes behavior of the subroutine $\findPTD()$.

For the base case, $\depth = 1$, we use $\compress{1}$ given as
Algorithm~\ref{alg:nlogn-compress}.  When its correctness was proved
we assumed $|S_0| \leq 2k+3$ and this is no longer the case.  However,
if $\alg{1}$ is applied with $|S_0| \leq 4k+3$ it will conclude that
$\tw(G) > k$ or give a tree decomposition of width at most $5k+4$.
The reason is as follows; Assume that $\findTD$ is applied with the
invariant $|S| \leq 4k+3$ instead of $2k+3$.  By the same argument as
in the original proof this invariant will hold, since
$\frac{1}{2}(4k+3)+k+2 \leq 4k+3$.  The only part of the correctness
(and running time analysis) affected by this change is the width of
the returned decomposition, and when the algorithm adds the separator
to $S$ it creates a bag of size at most $(4k+3) + (k+2) = 5k+5$ and
hence our argument for the base case is complete.  For the induction
step, suppose that the theorem and lemma hold for $\depth - 1$.  We
show that $\compress{\depth}$ is correct and runs in $O(c^k \cdot n
\log^{(\depth)} n)$ time.  This immediately implies correctness and
complexity of $\alg{\depth}$, in the same manner as in
Section~\ref{section:nlogn}.


To prove correctness of $\compress{\depth}$, suppose that $\td'$ is a
valid tree decomposition for some $G' \subseteq G$ that we have
obtained from $\findPTD$.  Observe that if $\td'=\false$, then $\tw(G)
> k$ by Lemma~\ref{lemma:findPTD}.  Otherwise, let $C_1, \dots, C_p$
be the connected components of $G\setminus W$, and let $G_j=G[C_j\cup
N(C_j)]$ for $j=1,2,\ldots,p$.  Let $\td_j$ be the tree decompositions
obtained from application of the algorithm $\alg{\depth-1}$ on graphs
$G_j$.  If $\td_j=\bot$ for any $j$, we infer that $\tw(G_j)>k$ and,
consequently, $\tw(G)>k$.  Assume then that for all the components we
have indeed obtained valid tree decompositions, with $N(C_j)$ as root
bags.  It can be easily seen that since $N(C_j)$ separates $C_j$ from
the rest of $G$, then attaching the root of $\td_j$ as a child of any
bag containing the whole $N(C_j)$ gives a valid tree decomposition;
the width of this tree decomposition is the maximum of widths of
$\td'$ and $\td_j$, which is at most $5k+3$.  Moreover, if we perform
this operation for all the components $C_j$, then all the vertices and
edges of the graph will be contained in some bag of the obtained tree
decomposition.

We now proceed to the time complexity of $\compress{\depth}$.  The
first thing done by the algorithm is the initialization of the data
structure and running $\findPTD$ to obtain $\td'$.  Application of
$\findPTD$ takes $O(c^k n)$ time by Lemma~\ref{lemma:findPTD}, and so
does initialization of the data structure (see
Section~\ref{section:datastructure}).  As discussed, creation of the
$\forg$ map and construction of the graphs $G_j$ takes
$O(k^{O(1)}\cdot n)$ time.


Now, the algorithm applies $\alg{\depth-1}$ to each graph $G_j$.  Let
$n_j$ be the number of vertices of $G_j$.  Note that
\[
\sum_{j=1}^p n_j= \sum_{j=1}^p |C_j|+\sum_{j=1}^p |N(C_j)|\leq n+p\cdot
(4k+3)\leq (5k+3)n.
\]
Moreover, as $n_j\leq \log n + (4k+3)$, it follows from concavity of $t
\to \log^{(\depth-1)} t$ that
\[
\log^{(\depth-1)} n_j \leq \log^{(\depth-1)}(\log n + (4k+3)) \leq
\log^{(\depth)} n + \log^{(\depth-1)} (4k+3).
\]

By the induction hypothesis, the time complexity of $\alg{\depth-1}$
on $G_j$ is $O(c^k \cdot n_j \log^{(\depth-1)} n_j)$, hence we spend
$O(c^k \cdot n_j \log^{(\depth-1)}n_j)$ time for $G_j$.  Attaching
each decomposition $\td_j$ to $\td'$ can be done in $O(k^{O(1)})$
time.

Let $C_\depth$ denote the complexity of $\compress{\depth}$ and
$A_\depth$ the complexity of $\alg{\depth}$.  By applying the
induction hypothesis we analyze the complexity of $\compress{\depth}$
(we use a constant $c_1>c$ to hide polynomial factors depending on
$k$):
\begin{align*}
  C_\depth(n,k) &= O(c^k\cdot n) + \sum\limits_{j=1}^{n}A_{\depth-1}(n_j,k)\\
  &= O(c^k\cdot n) + \sum\limits_{j=1}^{p}O(c^k\cdot n_j \log^{(\depth-1)}n_j)\\
  &\leq O(c^k\cdot n) + \sum\limits_{j=1}^{p}O(c^k\cdot n_j (\log^{(\depth)}n + \log^{(\depth-1)} (4k+3)))\\
  &= O(c_1^k\cdot n) + \sum\limits_{j=1}^{p}O(c^k\cdot n_j\log^{(\depth)}n )\\
  &\leq O(c_1^k\cdot n) + (5k+3)n \cdot O(c^k\cdot \log^{(\depth)}n) =
  O(c_1^k\cdot n\log^{(\depth)}n).
\end{align*}

We conclude that $\compress{\depth}$ is both correct and that it runs
in $O(c^k\cdot n\log^{(\depth)}n)$ time for some $c\in \mathbb{N}$.
The correctness and time complexity $O(c^kn\log^{(\depth)}n)$ of
$\alg{\depth}$ follow in the same manner as in the previous section.
And hence our induction step is complete and the correctness of
Lemma~\ref{lemma:nlogn-compression} and Theorem~\ref{theorem:nlogin}
follows.  The only assumption we made was that of the correctness of
Lemma~\ref{lemma:findPTD}, which will be given immediately.

\subsection{The algorithm $\findPTD$}

The following lemma describes behavior of the subroutine $\findPTD$.

\begin{lemma}
  \label{lemma:findPTD}
  There exists an algorithm that, given data structure $\ds$ in a
  state such that $|S| \leq 4k+3$ if $\whatsep = s$ or $|S| \leq 3k+2$
  if $\whatsep = u$, in time $O(c^k n)$ either concludes that $\tw(G[U
  \cup S]) > k$, or give a tree decomposition $\td'$ of $G' \subseteq
  G[U \cup S]$ such that
  \begin{itemize}
  \item the width of the decomposition is at most $5k+4$ and $S$ is
    its root bag;
  \item for every connected component $C$ of $G[U \cup S]\setminus V(G')$, the
    size of the component is less than $\log n$, its neighborhood is
    of size at most $4k+3$, and there is a bag in the decomposition
    $\td'$ containing this whole neighborhood.
  \end{itemize}
\end{lemma}

The pseudocode of $\findPTD$ is presented as
Algorithm~\ref{alg:nlogin-findptd}.  The algorithm proceeds very
similarly to the subroutine $\findTD$, given in
Section~\ref{section:nlogn}.  The main differences are the following.
\begin{itemize}
\item We alternate usage of $\qSsep$ and $\qUsep$ between the levels
  of the recursion to achieve that the resulting tree decomposition is
  also balanced.  A special flag in the data structure, $\whatsep$,
  that can be set to $s$ or $u$, denotes whether we are currently
  about to use $\qSsep$ or $\qUsep$, respectively.  When initializing
  the data structure we set $\whatsep=s$, so we start with finding a
  balanced $S$-separator.
\item When identifying the next components using query $\qpin$, we
  stop when a component of size less than $\log n$ is discovered.  The
  remaining components are left without being decomposed.
\end{itemize}
The new query $\qUsep$, provided that we have the data structure with
$S$ and $\pin$ distinguished, gives a $\frac{8}{9}$-balanced separator
of $U$ in $G[U]$ of size at most $k+1$.  That is, it returns a subset
$Y$ of vertices of $U$, with cardinality at most $k+1$, such that
every connected component of $G[U]\setminus Y$ has at most
$\frac{8}{9}|U|$ vertices.  If such a separator cannot be found (which
is signalized by $\bot$), we may safely conclude that $\tw(G[U])>k$
and, consequently $\tw(G)>k$.  The running time of query $\qUsep$ is
$O(c^k\cdot \log n)$.

We would like to remark that the usage of balanced $U$-separators make
it not necessary to add the pin to the obtained separator.  Recall
that this was a technical trick that was used in
Section~\ref{section:nlogn} to ensure that the total number of bags of
the decomposition was linear.

\begin{algorithm}[h!]
  \KwData{Data structure $\ds$} \KwOut{Partial tree decomposition of
    width at most $k$ of $G[S \cup U]$ with $S$ as root bag or conclusion
    that $\tw(G) > k$.}
\Indp
  \BlankLine
  $\oldS \leftarrow \ds.\getS()$\\
  $\oldPin \leftarrow \ds.\getPin()$\\
  $\mathtt{old_w} \leftarrow \ds.\whatsep$
  
  \eIf{$\ds.\whatsep = s$}{
    $\sep \leftarrow \ds.\qSsep()$\\
    $\ds.\whatsep \leftarrow u$ }{
    $\sep \leftarrow \ds.\qUsep()$\\
    $\ds.\whatsep \leftarrow s$ }

  \If{$\sep = \false$}{
    $\ds.\whatsep \leftarrow \mathtt{old_w}$\\
    \KwRet $\false$ \hspace{1cm}/*\textsf{\footnotesize{ safe to return: the 
    state not changed } }*/\\
  }
  $\ds.\insertX(\sep)$\\
  $\pins \leftarrow \emptyset$\\

  \While{$(u,l) \leftarrow \ds.\qpin() \neq \false \; \mathrm{and} \; l \geq \log
    n$}{
    $\pins.append(u)$\\
    $\ds.\insertF(u)$\\
  }
  
  $\ds.\clearX()$\\
  $\ds.\clearF()$\\
  $\ds.\insertS(\sep)$\\
  $\bags \leftarrow \emptyset$\\
  \For{$u \in \pins$}{
    $\ds.\setPin(u)$\\
    $\bags.append(\ds.\qnei())$\\
  }
  $\children \leftarrow \emptyset$\\

  \For{$u,b \in \pins,\bags$}{
    $\ds.\setPin(u)$\\
    $\ds.\clearS()$\\
    $\ds.\insertS(b)$\\
    $\children.append(\findPTD())$ }
  
  $\ds.\whatsep \leftarrow \mathtt{old_w}$
  
   $\ds.\clearS()$\\
  $\ds.\insertS(\oldS)$\\
  $\ds.\setPin(\oldPin)$\\
  \If{$\false \in \children$}{
    \KwRet $\false$  \hspace{1cm}/*\textsf{\footnotesize{ postponed because of 
    rollback of $S$ and $\pin$ } }*/\\ }
    \KwRet $\build(\oldS, \sep, \children)$
 
  \caption{$\findPTD$}
  \label{alg:nlogin-findptd}
\end{algorithm}

\subsubsection{Correctness}
The invariants of Algorithm~\ref{alg:nlogin-findptd} are as for
Algorithm~\ref{alg:nlogn-findtd}, except for the size of $S$, in which
case we distinguish whether $\whatsep$ is $s$ or $u$.  In the case of
$s$ the size of $S$ is at most $4k+3$ and for $u$ the size of $S$ is
at most $3k+2$.

If $\whatsep = u$ then, since $|S| \leq 3k+2$ and we add an $U$-separator
of size at most $k+1$ and make this our new $S$, the size of the new
$S$ will be at most $4k+3$ and we set $\whatsep = s$.  For every
component $C$ on which we recurse, the cardinality of its
neighborhood ($S$ at the moment of recursing) is therefore bounded by
$4k+3$.  So the invariant holds when $\whatsep = u$.

We now show that the invariant holds when $\whatsep = s$.  Now
$|\oldS| \leq 4k+3$.  We find $\frac{1}{2}$-balanced $S$-separator
$\sep$ of size at most $k+1$.  When recursing, the new $S$ is the
neighborhood of some component $C$ of $G[\oldS\cup U]\setminus
(\oldS\cup \sep)$ (note that we refer to $U$ before resetting the
pin).  This component is contained in some component $C'$ of $G[\oldS
\cup U]\setminus \sep$, and all the vertices of $\oldS$ adjacent to
$C$ must be contained in $C'$.  Since $\sep$ is a
$\frac{1}{2}$-balanced $\oldS$-separator, we know that $C'$ contains
at most $\frac{1}{2}|\oldS|$ vertices of $\oldS$.  Hence, when
recursing we have that $|S| \leq \frac{1}{2}|\oldS| + |\sep| =
\frac{1}{2}(4k+3)+k+1 = 3k + \frac{5}{2}$ and, since $|S|$ is an
integer, it follows that $|S| \leq 3k+2$.  Hence, the invariant also
hold when $\whatsep = s$.

Note that in both the checks we did not assume anything about the size
of the component under consideration.  Therefore, it also holds for
components on which we do not recurse, i.e., those of size at most
$\log n$, that the cardinalities of their neighborhoods will be
bounded by $4k+3$.

The fact that the constructed partial tree decomposition is a valid
tree decomposition of the subgraph induced by vertices contained in
it, follows immediately from the construction, similarly as in
Section~\ref{section:nlogn}.  A simple inductive argument also shows
that the width of this tree decomposition is at most $5k+4$: at each
step of the construction, we add two bags of sizes at most
$(4k+3)+(k+1) \leq 5k+5$ to the obtained decompositions of the
components, which by inductive hypothesis are of width at most $5k+4$.

Finally, we show that every connected component of $G[S \cup U]\setminus V(G')$
has size at most $\log n$ and that the neighborhood of each of these
connected component is contained in some bag on the partial tree
decomposition $\td'$.  First, by simply breaking out of the loop shown
in Algorithm~\ref{alg:nlogin-findptd} at the point we get a pair
$(\pin, l)$ such that $l < \log n$, we are guaranteed that the
connected component of $G[S \cup U]\setminus \sep$ containing $\pin$
has size less than $\log n$, and so does every other connected
component of $G[S \cup U]$ not containing a vertex from $F$ and which
has not been visited by $\ds.\qpin()$.  Furthermore, since immediately
before we break out of the loop due to small size we add $S \cup \sep$
to a bag, we have ensured that the neighborhood of any such small
component is contained in this bag.  The bound on the size of this
neighborhood has been already argued.


\subsubsection{Complexity}
Finally, we show that the running time of the algorithm is $O(c^k\cdot
n)$.  The data structure operations all take time $O(c^k \log n)$ and
we get the data structure $\ds$ as input.

The following combinatorial lemma will be helpful to bound the number
of bags in the tree decomposition produced by $\findPTD$.  We aim to
show that the tree decomposition $\td'$ contains at most $O(n / \log
n)$ bags, so we will use the lemma with $\mu(i) = w_i / \log n$, where
$i$ is a node in a tree decomposition $\td'$ and $w_i$ is the number
of vertices in $G[U]$ when $i$ is added to $\td'$.  Having proven the
lemma, we can show that the number of bags is bounded by $O(\mu(r)) =
O(n / \log n)$, where $r$ is the root node of $\td'$.

\begin{lemma}
  \label{lemma:bounded-tree}
  Let $T$ be a rooted tree with root $r$.  Assume that we are given a
  measure $\mu:V(T)\to \R$ with the following properties:
  \begin{enumerate}[(i)]
  \item[(i)] $\mu(v)\geq 1$ for every $v\in V(T)$,
  \item[(ii)] for every vertex $v$, let $v_1,v_2,\ldots,v_p$ be its
    children, we have that $\sum_{i=1}^p \mu(v_i)\leq \mu(v)$, and
  \item[(iii)] there exists a constant $0<C<1$ such that for for every
    two vertices $v,v'$ such that $v$ is a parent of $v'$, it holds
    that $\mu(v')\leq C\cdot \mu(v)$.
  \end{enumerate}
  Then $|V(T)|\leq \left(1+\frac{1}{1-C}\right)\mu(r)-1$.
\end{lemma}
\begin{proof}
  We prove the claim by induction with respect to the size of $V(T)$.
  If $|V(T)|=1$, the claim trivially follows from property (i).  We
  proceed to the induction step.
  
  Let $v_1,v_2,\ldots,v_p$ be the children of $r$ and let
  $T_1,T_2,\ldots,T_p$ be subtrees rooted in $v_1,v_2,\ldots,v_p$,
  respectively.  If we apply the induction hypothesis to trees
  $T_1,\ldots,T_p$, we infer that for each $i=1,2,\ldots,p$ we have
  that $|V(T_i)| \leq \left(1+\frac{1}{1-C}\right)\mu(v_i)-1$.  By
  summing the inequalities we infer that:
  \begin{eqnarray*}
    |V(T)|& \leq & 1-p+\left(1+\frac{1}{1-C}\right)\sum_{i=1}^p \mu(v_i).
  \end{eqnarray*}
  We now consider two cases.  Assume first that $p\geq 2$; then:
  \begin{eqnarray*}
    |V(T)|& \leq & 1-2+\left(1+\frac{1}{1-C}\right)\sum_{i=1}^p \mu(v_i)\leq
    \left(1+\frac{1}{1-C}\right)\mu(r)-1,
  \end{eqnarray*}
  and we are done.  Assume now that $p=1$; then
  \begin{eqnarray*}
    |V(T)| & \leq & \left(1+\frac{1}{1-C}\right)\mu(v_1) \leq
    C\left(1+\frac{1}{1-C}\right)\mu(r)\\
    & = & \left(1+\frac{1}{1-C}\right)\mu(r)-(2-C)\mu(r) \leq
    \left(1+\frac{1}{1-C}\right)\mu(r)-1, 
  \end{eqnarray*}
  and we are done as well.
\end{proof}

We now prove the following claim.

\newcommand{\sm}{\textrm{small}}
\newcommand{\lr}{\textrm{large}}
\newcommand{\inte}{\textrm{int}}
\newcommand{\leaf}{\textrm{leaf}}

\begin{claim}~\label{claim:sublinear-bags} The partial tree
  decomposition $\td'$ contains at most $42n / \log n$ nodes.
\end{claim}
\begin{proof}
  Let us partition the set of nodes $V(\td')$ into two subsets.  At
  each recursive call of $\findPTD$, we create two nodes: one
  associated with the bag $\oldS$, and one associated with the bag
  $\oldS\cup \sep$.  Let $I_\sm$ be the set of nodes associated with
  bags $\oldS$, and let $I_\lr$ the the set of remaining bags,
  associated with bags $\oldS\cup \sep$.  As bags are always
  constructed in pairs, it follows that
  $|I_\sm|=|I_\lr|=\frac{1}{2}|V(\td')|$.  Therefore, it remains to
  establish a bound on $|I_\sm|$.
  
  For a node $i\in V(\td')$, let $w_i$ be the number of vertices
  strictly below $i$ in the tree decomposition $\td'$, also counting
  the vertices outside the tree decomposition.  Note that by the
  construction it immediately follows that $w_i\geq \log n$ for each
  $i\in I_\sm$.
  
  We now partition $I_\sm$ into three parts: $I_\sm^s$,
  $I_\sm^{u,\inte}$, and $I_\sm^{u,\leaf}$.  $I_\sm^s$ consists of all
  the nodes created in recursive calls where $\whatsep=s$.
  $I_\sm^{u,\leaf}$ consists of all the nodes created in recursive
  calls where $\whatsep=u$, and moreover the algorithm did not make
  any more recursive calls to $\findTD$ (in other words, all the
  components turned out to be of size smaller than $\log n$).
  $I_\sm^{u,\leaf}$ consists of all the remaining nodes created in
  recursive calls where $\whatsep=u$, that is, such that the algorithm
  made at least one more call to $\findTD$.  We aim at bounding the
  size of each of the sets $I_\sm^s$, $I_\sm^{u,\inte}$, and
  $I_\sm^{u,\leaf}$ separately.
  
  We first claim that $|I_\sm^{u,\leaf}|\leq n / \log n$.  Indeed, we
  have that the sets of vertices strictly below nodes of
  $I_\sm^{u,\leaf}$ are pairwise disjoint.
  And since any bag in $I_\sm^{i,\leaf}$ is a subset of its parent and the
  recursive call to create the bag was made we know that there is at least
  $\log n$ vertices below. As their total union is of size at most $n$, the
  claim follows.
  
  We now claim that $|I_\sm^{u,\inte}|\leq |I_\sm^{s}|$.  Indeed, if with
  every node $i\in I_\sm^{u,\inte}$ we associate any of its grandchild
  belonging to $I_\sm^{s}$, whose existence is guaranteed by the
  definition of $I_\sm^{u,\inte}$, we obtain an injective map from
  $I_\sm^{u,\inte}$ into $I_\sm^{s}$.
  
  We are left with bounding $|I_\sm^{s}|$.  For this, we make use of
  Lemma~\ref{lemma:bounded-tree}.  Recall that vertices of $I_\sm^{s}$
  are exactly those that are in levels whose indices are congruent to
  $1$ modulo $4$, where the root has level $1$; in particular, $r\in
  I_\sm^{s}$.  We define a rooted tree $T$ as follows.  The vertex set
  of $T$ is $I_\sm^{s}$, and for every two nodes $i,i'\in I^s_\sm$
  such that $i'$ is an ancestor of $i$ exactly $4$ levels above
  (grand-grand-grand-parent), we create an edge between $i$ and $i'$.
  It is easy to observe that $T$ created in this manner is a rooted
  tree, with $r$ as the root.
  
  We can now construct a measure $\mu:V(T)\to \R$ by taking
  $\mu(i)=w_i/\log n$.  Let us check that $\mu$ satisfies the
  assumptions of Lemma~\ref{lemma:bounded-tree} for $C=\frac{8}{9}$.
  Property (i) follows from the fact that $w_i\geq \log n$ for every
  $i\in I_\sm$.  Property (ii) follows from the fact that the parts of
  the components on which the algorithm recurses below the bags are
  always pairwise disjoint.  Property (iii) follows from the fact that
  between every pair of parent, child in the tree $T$ we have used a
  $\frac{8}{9}$-balanced $U$-separator.  Application of
  Lemma~\ref{lemma:bounded-tree} immediately gives that $|I_\sm^s|\leq
  10 n/\log n$, and hence $|V(\td')|\leq 42 n/\log n$.
\end{proof}

To conclude the running time analysis of $\findPTD$, we provide a
similar charging scheme as in Section~\ref{section:nlogn}.  More
precisely, we charge every node of $\td'$ with $O(c^k\cdot \log n)$
running time; Claim~\ref{claim:sublinear-bags} ensures us that then
the total running time of the algorithm is then $O(c^k\cdot n)$.

Let $B=\oldS$ and $B'=\oldS\cup \sep$ be the two bags constructed at some
call of $\findPTD$.  All the operations in this call, apart from the
two loops over the components, take $O(c^k\cdot \log n)$ time and are
charged to $B'$.  Moreover, the last call of $\qpin$, when a component
of size smaller than $\log n$ is discovered, is also charged to $B'$.
As this call takes $O(1)$ time, $B'$ is charged with $O(c^k\cdot \log
n)$ time in total.

We now move to examining the time spent while iterating through the
loops.  Let $B_j$ be the root bag of the decomposition created for
graph $G_j$.  We charge $B_j$ with all the operations that were done
when processing $G_j$ within the loops.  Note that thus every such
$B_j$ is charged at most once, and with running time $O(c^k\cdot \log
n)$.  Summarizing, every bag of $\td'$ is charged with $O(c^k\cdot
\log n)$ running time, and we have at most $42 n/\log n$ bags, so the
total running time of $\findPTD$ is $O(c^k\cdot n)$.

\section{An $O(c^k n)$ 5-approximation algorithm for treewidth}
\label{section:linear}
In this section we give the main result of the paper.  The algorithm
either calls $\alg{\depth}$ for $\depth = 2$ or a version of
Bodlaender~\cite{Bodlaender96} applying a table lookup implementation
of the dynamic programming algorithm by Bodlaender and
Kloks~\cite{BodlaenderK96}, depending on how $n$ and $k$ relate.
These techniques combined will give us a $5$-approximation algorithm
for treewidth in time single exponential in $k$ and linear in $n$.

\begin{theorem}[Theorem~\ref{thm:mainThm}, restated]
  There exists an algorithm, that given an $n$-vertex graph $G$ and an
  integer $k$, in time $2^{O(k)} n$ either outputs that the treewidth
  of $G$ is larger than $k$, or constructs a tree decomposition of $G$
  of width at most $5k + 4$.
\end{theorem}


As mentioned above, our algorithm distinguishes between two cases.
The first case is when $n$ is ``sufficiently small'' compared to $k$.
By this, we mean that $n \leq 2^{2^{c_0 k^3}}$.  The other case is
when this is \emph{not} the case.  For the first case, we can apply
$\alg{2}$ and since $n$ is sufficiently small compared to $k$ we can
observe that $\log^{(2)}n= k^{O(1)}$, resulting in a $2^{O(k)}n$ time
algorithm.  For the case when $n$ is large compared to $k$, we
construct a tree automata in time double exponential in $k$.  In this
case, double exponential in $k$ is in fact also linear in $n$.  This
automaton is then applied on a nice expression tree constructed from
our tree decomposition and this results in an algorithm running in
time $2^{O(k)}n$.

\begin{lemma}[Bodlaender and Kloks~\cite{BodlaenderK96}]
  There is an algorithm, that given a graph $G$, an integer $k$, and a
  nice tree decomposition of $G$ of width at most $\ell$ with $O(n)$
  bags, either decides that the treewidth of $G$ is more than $k$, or
  finds a tree decomposition of $G$ of width at most $k$ in time
  $O(2^{O(k \ell^2)} n )$.
  \label{lemma:BodlaenderKloks}
\end{lemma}

Our implementation with table lookup of this result gives the
following:

\begin{lemma}
  There is an algorithm, that given a graph $G$, an integer $k$, and a
  nice tree decomposition of $G$ of width at most $\ell = O(k)$ with
  $O(n)$ bags, either decides that the treewidth of $G$ is more than 
  $k$, or finds a tree decomposition of $G$ of width at most $k$, in time
  $O( 2^{2^{c_0 k \ell^2}} + k^{c_1} n)$, for constants
  $c_0$ and $c_1$.
  \label{lemma:tablelookupBodlaenderKloks}
\end{lemma}

The proof of Lemma~\ref{lemma:tablelookupBodlaenderKloks} will be
given later.  We first discuss how
Lemma~\ref{lemma:tablelookupBodlaenderKloks} combined with the results
in other sections imply the main result of the paper.  First we handle
the case when $n \leq 2^{2^{c_0 k^3}}$.  We then call $\alg{2}$ on
each connected component of $G$ separately, with $S=\emptyset$.  This
algorithm runs in time $O(c'^k \cdot n \log \log n) = O(c'^k n k^3) =
O(c^k n)$ time for some constants $c$ and $c'$.

For the remaining of this section we will assume $n > 2^{2^{c_0 k^3}}$.
An inspection of Bodlaender's algorithm~\cite{Bodlaender96} shows that it
contains the following parts:
\begin{itemize}
\item A recursive call to the algorithm is made on a graph with $c_3
  n$ vertices, for $c_3 = 1 - \frac{1}{8k^6+ O(k^4)}$.
\item The algorithm of Lemma~\ref{lemma:BodlaenderKloks} is called
  with $\ell = 2k+1$.
\item Some additional work that uses time linear in $n$ and
  polynomial in $k$.
\end{itemize}

The main difference of our algorithm in the case of ``large $n$'' is
that we replace the call to the algorithm of
Lemma~\ref{lemma:BodlaenderKloks} with a call to the algorithm of
Lemma~\ref{lemma:tablelookupBodlaenderKloks}, again with $\ell =
2k+1$.  At some point in the recursion, instances will have size less
than $2^{2^{c_0 k^3}}$.  At that point, we call $\alg{2}$ on this
instance; with the main difference that at one level higher in the
recursion, the algorithm of
Lemma~\ref{lemma:tablelookupBodlaenderKloks} is called with $\ell =
10k+9$.

The analysis of the running time is now simple.  A call to the
algorithm of Lemma~\ref{lemma:tablelookupBodlaenderKloks} uses time
which is bounded by $O( 2^{2^{c_0 k \ell^2}} + k^{c_1} n) = O(n +
k^{c_1} n)$, i.e., time linear in $n$ and polynomial in $k$.  The
total work of the algorithm is thus bounded by a function $T(n)$ that
fulfills $T(n) \leq T(c_3 n) + O(k^{c_4}n) = O(k^{6+c_4}n)$, for
constant $c_4$, i.e., time polynomial in $k$ and linear in $n$.  This
proves our main result.  What remains for this section is a proof of
Lemma~\ref{lemma:tablelookupBodlaenderKloks}.

\subsection{Nice expression trees}
The dynamic programming algorithm in~\cite{BodlaenderK96} is described
with help of so called {\em nice tree decompositions}.  As we need to
represented a nice tree decomposition as a labeled tree with the label
alphabet of size a function of $k$, we use a slightly changed notion
of {\em labeled nice tree decomposition}.  The formalism is quite
similar to existing formalisms, e.g., the operations on $k$-terminal
graphs by Borie~\cite{Borie88}.

A labeled terminal graph is a 4-tuple $G=(V,E,X,f)$, with $(V,E)$ a
graph, $X \rightarrow V$ a set of {\em terminals}, and $f: X
\rightarrow {\mathcal{N}}$ an injective mapping of the terminals to
non-negative integers, which we call {\em labels}.  A $k$-labeled
terminal graph is a labeled terminal graph with the maximum label at
most $k$, i.e., $\max_{x\in X} f(x)\leq k$.  Let $O_k$ be the set of
the following operations on $k$-terminal graphs.

\vskip 0.3cm
\noindent{\bf{Leaf$_\ell$():}} gives a $k$-terminal graph with one vertex $v$,
  no edges, with $v$ a terminal with label $\ell$
  \vskip 0.3cm
  \noindent{\bf{Introduce$_{\ell, S}$($G$):}} $G = (V,E,X,f)$ is a
  $k$-terminal graph, $\ell$ a non-negative integer, and $S \subseteq
  \{1, \ldots, k\}$ a set of labels.  If there is a terminal vertex in
  $G$ with label $\ell$, then the operation returns $G$, otherwise it
  returns the graph, obtained by adding a new vertex $v$, making $v$ a
  terminal with label $\ell$, and adding edges $\{v,w\}$ for each
  terminal $w\in X$ with $f(w)\in S$.  I.e., we make the new vertex
  adjacent to each existing terminal whose label is in $S$.  \vskip
  0.3cm
  \noindent{\bf{Forget{$_\ell$}($G$):}} Again $G=(V,E,X,f)$ is a
  $k$-terminal graph.  If there is no vertex $v\in X$ with
  $f(v)=\ell$, then the operation returns $G$, otherwise, we 'turn $v$
  into a non-terminal, i.e., we return the $k$-terminal graph
  $(V,E,X-\{v\},f')$ for the vertex $v$ with $f(v)=\ell$, and $f'$ is
  the restriction of $f$ to $X-\{v\}$.  \vskip 0.3cm
  \noindent{\bf{Join($G$,$H$):}} $G =(V,E,X,f)$ and $H=(W,F,Y,g)$ are
  $k$-terminal graphs.  If the range of $f$ and $g$ are not equal,
  then the operation returns $G$.  Otherwise, the result is obtained
  by taking the disjoint union of the two graphs, and then identifying
  terminals with the same label.
  \vskip 0.3cm Note that for given $k$, $O_k$ is a collection of $k +
  k\cdot 2^k + k + k^2$ operations.  When the treewidth is $k$, we
  work with $k+1$-terminal graphs.  The set of operations mimics
  closely the well known notion of nice tree decompositions (see
  e.g.,~\cite{Kloks93,Bodlaender98}).

\begin{proposition}
  Suppose a tree decomposition of $G$ is given of width at most $k$
  with $m$ bags.  Then, in time, linear in $n$ and polynomial in $k$,
  we can construct an expression giving a graph isomorphic to $G$ in
  terms of operations from $O_{k+1}$ with the length of the expression
  $O(mk)$.
\end{proposition}

\begin{proof}
  First build with standard methods a nice tree decomposition of $G$
  of width $k$; this has $O(km)$ bags, and $O(m)$ join nodes.  Now,
  construct the graph $H=(V,F)$, with for all $v,w\in V$, $\{v,w\}\in
  F$, if and only if there is a bag $i$ with $v,w\in X_i$.  It is well
  known that $H$ is a chordal super graph of $G$ with maximum clique
  size $k+1$ (see e.g.,~\cite{Bodlaender98}).  Use a greedy linear
  time algorithm to find an optimal vertex coloring $c$ of $H$
  (see~\cite[Section 4.7]{Golumbic80}.)
  
  Now, we can transform the nice tree decomposition to the expression
  as follows: each leaf bag that contains a vertex $v$ is replaced by
  the operation Leaf$_{c(v)}$, i.e., we label the vertex by its color
  in $H$.  We can now replace bottom up each bag in the nice tree
  decomposition by the corresponding operation; as we labeled vertices
  with the color in $H$, we have that all vertices in a bag have
  different colors, which ensures that a Join indeed performs
  identifications of vertices correctly.  Bag sizes are bounded by
  $k+1$, so all operations belong to $O_{k+1}$.
\end{proof}

View the expression as a labeled rooted tree, i.e., each node is
labeled with an operation from $O_{k+1}$; leaves are labeled with a
Leaf operation, and binary nodes have a Join-label.  To each node of
the tree $i$, we can associate a graph $G_i$, and the graph $G_r$
associated to the root node $r$ is isomorphic to $G$.  Call such a
labeled rooted tree a {\em nice expression tree} of width $k$.

\subsection{Dynamic programming and finite state tree automata}
The discussion in this paragraph holds for all problems invariant
under isomorphism,.  Note that the treewidth of a graph is also
invariant under isomorphisms.  We use ideas from the early days of
treewidth, see e.g.,~\cite{FellowsL89,AbrahamsonF93}.

A dynamic programming algorithm on nice tree decompositions can be
viewed also as a dynamic programming algorithm on a nice expression
tree of width $k$.  Suppose that we have a dynamic programming
algorithm that computes in bottom-up order for each node of the
expression tree a table with at most $r=O(1)$ bits per table, and to
compute a table, only the label of the node (type of operation) and
the tables of the children of the node are used.  We remark that the
DP algorithm for treewidth from Bodlaender and
Kloks~\cite{BodlaenderK96} is indeed of this form, if we see $k$ as a
fixed constant.  Such an algorithm can be seen as a finite state tree
automaton: the states of the automaton correspond to the at most $2^r
= O(1)$ different tables; the alphabet are the $O(1)$ different labels
of tree nodes.

To decide if the treewidth of $G$ is at most $k$, we first explicitly
build this finite state tree automaton, and then execute it on the
expression tree.  For actually building the corresponding tree
decomposition of $G$ of width at most $k$, if existing, some more work
has to be done, which is described later.

\subsection{Table lookup implementation of dynamic programming}
The algorithm of Bodlaender and Kloks~\cite{BodlaenderK96} builds for
each node in the nice tree decomposition of {\em table of
  characteristics}: each characteristic represents the `essential
information' of a tree decomposition of width at most $k$ of the graph
associated with the bag.

Inspection of the algorithm~\cite{BodlaenderK96} easily shows that the
number of different characteristics is bounded by $2^{O(k \cdot
  \ell^2)}$ when we are given an expression tree of width $\ell$ and
want to test if the treewidth is at most $k$.  (See~\cite[Definition
5.9]{BodlaenderK96}.)

We now use that we represent the vertices in the bag, i.e., the
terminals, by a label from $\{1, \ldots, \ell+1\}$ if $\ell$ is the
width of the nice expression tree.  Thus, we have a set $C_{k,\ell}$
(only depending of $k$ and $\ell$) that contains all possible
characteristics that belong to a table; each table is just a subset of
$C_{k,\ell}$, i.e., an element of ${\cal P}(C_{k,\ell})$.  I.e., we
can view the decision variant of the dynamic programming algorithm of
Bodlaender and Kloks as a finite state tree automaton with alphabet
$O_{\ell+1}$ and state set ${\cal P}(C_{k,\ell})$.

The first step of the algorithm is now to explicitly construct this
finite state tree automaton.  We can do this as follows.  Enumerate
all characteristics in ${\cal P}(C_{k,\ell})$, and number them $c_1,
\ldots, c_s$, $s = 2^{O(k \cdot \ell^2)}$.  Enumerate all elements of
${\cal P}(C_{k,\ell})$, and number them $t_1, \ldots, t_{s'}$, $s' =
2^{2^{O(k \cdot \ell^2)}}$; store with $t_i$ the elements of its set.

Then, we compute a transition function $F: O_{\ell+1} \times \{1, \ldots,
s'\} \times \{1, \ldots, s'\} \rightarrow \{1, \ldots, s'\}$.  In
terms of the finite state automaton view, $F$ gives the state of a
node given its symbol and the states of its children.  (If a node has
less than two children, the third, and possibly the second argument
are ignored.) In terms of the DP algorithm, if we have a tree node $i$
with operation $o \in O_{\ell+1}$, and the children of $i$ have tables
corresponding to $t_\alpha$ and $t_\beta$, then $F(o,\alpha,\beta)$
gives the number of the table obtained for $i$ by the algorithm.  To
compute one value of $F$, we just execute one part of the algorithm of
Bodlaender and Kloks~\cite{BodlaenderK96}.  Suppose we want to compute
$F(o,\alpha,\beta)$.  If $o$ is a shift operation, then a simple
renaming suffices.  Otherwise, build the tables $T_\alpha$ and
$T_\beta$ corresponding to $t_\alpha$ and $t_\beta$, and execute the
step of the algorithms from~\cite{BodlaenderK96} for a node with
operation $o$ whose children have tables $T_\alpha$ and $T_\beta$.
(If the node is not binary, we ignore the second and possibly both
tables.)  Then, lookup what is the index of the resulting table; this
is the value of $F(o,\alpha,\beta)$.

We now estimate the time to compute $F$.  We need to compute $O(2^\ell
\cdot \ell \cdot s'^2) = O( 2^{2^{O(k \cdot \ell^2)}})$ values; each
executes one step of the DP algorithm and does a lookup in the table,
which is easily seen to be bounded again by $O( 2^{2^{O(k \cdot
    \ell^2)}})$, so the total time to compute $F$ is still bounded by
$O( 2^{2^{O(k \cdot \ell^2)}})$.  To decide if the treewidth of $G$ is
at most $k$, given a nice tree decomposition of width at most $\ell$,
we thus carry out the following steps:
\begin{itemize}
\item Compute $F$.
\item Transform the nice tree decomposition to a nice expression tree
  of width $\ell$.
\item Compute bottom-up (e.g., in post-order) for each node $i$ in the
  expression tree a value $v_i$, with a node $i$ labeled by operation
  $o\in O_{\ell+1}$ and children with values $v_{j_1}$, $v_{j_2}$, we
  have $v_i = F(o,v_{j_1},v_{j_2})$.  If $v$ has less than two
  children, we take some arbitrary argument for the values of missing
  children.  In this way, $v_i$ corresponds to the table that is
  computed by the DP algorithm of~\cite{BodlaenderK96}.
\item If the value $v_r$ for the root of the expression tree
  corresponds to the empty set, then the treewidth of $G$ is more than
  $k$, otherwise the treewidth of $G$ is at most $k$.
  (See~\cite[Section 4.6 and 5.6]{BodlaenderK96}.)
\end{itemize}

If our decision algorithm decides that the treewidth of $G$ is more
than $k$, we reject, and we are done.  Otherwise, we need to do
additional work to construct a tree decomposition of $G$ of width at
most $k$, which is described next.

\subsection{Constructing tree decompositions}
After the decision algorithm has determined that the treewidth of $G$
is at most $k$, we need to find a tree decomposition of $G$ of width
at most $k$.  Again, the discussion is necessarily not self contained
and we refer to details given in~\cite[Chapter 6]{BodlaenderK96}.

Basically, each table entry (characteristic) in the table of a join
node is the result of a combination of a table entry in the table of
the left child and a table entry of the table of the right child.
Similarly, for nodes with one child, each table entry is the result of
an operation to a table entry in the table of the child node.  Leaf
nodes represent a graph with one vertex, and we have just one tree
decomposition of this graph, and thus one table entry in the table of
a leaf node.

If we select a characteristic of the root bag, this recursively
defines one characteristic from each table.  In the first phase of the
construction, we make such a selection.  In order to do this, we first
pre-compute another function $g$, that helps to make this selection.
$g$ has four arguments: an operation from $O_{\ell+1}$, the index of a
characteristic (a number between $1$ and $s$), and the indexes of two
states (numbers between $1$ and $s'=2^s$).  As value, $g$ yields
$\false$ or a pair of two indexes of characteristics.  The intuition
is as follows: suppose we have a node $i$ in the nice expression tree
labeled with $o$, an index $c_i$ of a characteristic of a (not yet
known) tree decomposition of $G_i$, and indexes of the tables of the
children of $i$, say $t_{j_1}$ and $t_{j_2}$.  Now,
$g(o,c_i,t_{j_1},t_{j_2})$ should give a pair $(c_{j_1},c_{j_2})$ such
that $c_i$ is the result of the combination of $c_{j_1}$ and $c_{j_2}$
(in case that $o$ is the join operation) or of the operation as
indicated above to $c_{j_1}$ (in case $o$ is an operation with one
argument; $c_{j_2}$ can have any value and is ignored).  If no such
pair exists, the output is $\false$.

To compute $g$, we can perform the following steps for each 4-tuple
$o,c_i,t_{j_1},t_{j_2}$.  Let $S_1 \subseteq C_{k,\ell}$ be the set
corresponding to $t_{j_1}$, and $S_2 \subseteq C_{k,\ell}$ be the set
corresponding to $t_{j_2}$.  For each $c\in S_1$ and $c'\in S_2$, see
if a characteristic $c$ and a characteristic $c'$ can be combined (or,
in case of a unary operation, if the relevant operation can be applied
to $c$) to obtain $c_1$.  If we found a pair, we return it; if no
combination gives $c_1$, we return $\false$.  Again, in case of unary
operations $o$, we ignore $c'$.  We do not need $g$ in case $o$ is a
leaf operation, and can give any return values in such cases.  One can
easily see that the computation of $g$ uses again $2^{2^{O(k \cdot
    \ell^2)}}$ time.

The first step of our construction phase is to build $g$, as
described.  After this, we select a characteristic from $C_{k,\ell}$
for each node in the nice expression tree, as follows.  As we arrived
in this phase, the state of the root bag corresponds to a nonempty set
of characteristics, and we take an arbitrary characteristic from this
set (e.g., the first one from the list).  Now, we select top-down in
the expression tree (e.g., in pre-order) a characteristic.  Leaf nodes
always receive the characteristic of the trivial tree decomposition of
a graph with one vertex.  In all other cases, if node $i$ has
operation $o$ and has selected characteristic $c$, the left child of
$i$ has state $t_{j_1}$ and the right child of $i$ has state $t_{j_2}$
(or, take any number, e.g., 1, if $i$ has only one child, i.e., $o$ is
a unary operation), look up the precomputed value of
$g(o,c,t_{j+1},t_{j_2})$.  As $c$ is a characteristic in the table
that is the result of $F(o,t_{j+1},t_{j_2})$, $g \neq \false$, so
suppose $g$ is the pair $(c', c'')$.  We associate $c'$ as
characteristic with the left child of $i$, and (if $i$ has two
children) $c''$ as characteristic with the right child of $i$.

At this point, we have associated a characteristic with each node in
the nice expression tree.  These characteristics are precisely the
same as the characteristics that are computed in the constructive
phase of the algorithm from~\cite[Section 6]{BodlaenderK96}, with the
sole difference that we work with labeled terminals instead of the
`names' of the vertices (i.e., in~\cite{BodlaenderK96}, terminals /
bag elements are identified as elements from $V$).

From this point on, we can follow without significant changes the
algorithm from~\cite[Section 6]{BodlaenderK96}: bottom-up in the
expression tree, we build for each node $i$, a tree decomposition of
$G_i$ whose characteristic is the characteristic we just selected for
$i$, together with a number of pointers from the characteristic to the
tree decomposition.  Again, the technical details can be found
in~\cite{BodlaenderK96}, our only change is that we work with
terminals labeled with integers in $\{1, \ldots, \ell+1\}$ instead of
bag vertices.

At the end of this process, we obtain a tree decomposition of the
graph associated with the root bag $G_r = G$ whose characteristic
belongs to the set corresponding to the state of $r$.  As we only work
with characteristics of tree decompositions of width at most $k$, we
obtained a tree decomposition of $G$ of width at most $k$.

All work we do, except for the pre-computation of the tables of $F$
and $g$ is linear in $n$ and polynomial in $k$; the time for the
pre-computation does not depend on $n$, and is bounded by $2^{2^{O(k
    \ell^2)}}$.  This ends the description of the proof of
Lemma~\ref{lemma:tablelookupBodlaenderKloks}.

\section{A data structure for queries in $O(c^k \log n)$ time}
\label{section:datastructure}

\subsection{Overview of the data structure}
Assume we are given a tree decomposition $(\{B_i \mid i \in I\},
T=(I,F))$ of $G$ of width $O(k)$. First we turn our tree decomposition
into a tree decomposition of depth $O(\log n)$, keeping the width to
$t = O(k)$, by the work of Bodlaender and
Hagerup~\cite{BodlaenderH98}.  Furthermore, by standard arguments we
turn this decomposition into a {\emph{nice}} tree decomposition in
$O(t^{O(1)}\cdot n)$ time, that is, a decomposition of the same width
and satisfying following properties:
\begin{itemize}
\item All the leaf bags, as well as the root bag, are empty.
\item Every node of the tree decomposition is of one of four different
  types:
  \begin{itemize}
  \item {\bf{Leaf node}}: a node $i$ with $B_i=\emptyset$ and no children.
  \item {\bf{Introduce node}}: a node $i$ with exactly one child $j$
    such that $B_i=B_j\cup\{v\}$ for some vertex $v\notin B_j$; we say
    that $v$ is {\emph{introduced}} in $i$.
  \item {\bf{Forget node}}: a node $i$ with exactly one child $j$ such
    that $B_i=B_j\setminus \{v\}$ for some vertex $v\in B_i$; we say
    that $v$ is {\emph{forgotten}} in $i$.
  \item {\bf{Join node}}: a node $i$ with two children $j_1,j_2$ such
    that $B_i=B_{j_1}=B_{j_2}$.
  \end{itemize}
\end{itemize}
The standard technique of turning a tree decomposition into a nice one
includes (i) adding paths to the leaves of the decomposition on which
we consecutively introduce the vertices of corresponding bags; (ii)
adding a path to the root on which we consecutively forget the
vertices up to the new root, which is empty; (iii) introducing paths
between every non-root node and its parent, on which we first forget
all the vertices that need to be forgotten, and then introduce all the
vertices that need to be introduced; (iv) substituting every node with
$d>2$ children with a balanced binary tree of $O(\log d)$ depth.  It
is easy to check that after performing these operations, the tree
decomposition has depth at most $O(t\log n)$ and contains at most
$O(t\cdot n)$ bags. Moreover, using folklore preprocessing routines,
in $O(t^{O(1)}\cdot n)$ time we may prepare the decomposition for
algorithmic uses, e.g., for each bag compute and store the list of
edges contained in this bag.  We omit here the details of this
transformation and refer to Kloks~\cite{Kloks93}.

In the data structure, we store a number of tables: three special
tables that encode general information on the current state of the
graph, and one table per each query.  The information stored in the
tables reflect some choice of subsets of $V$, which we will call the
{\emph{current state of the graph}}.  More precisely, at each moment
the following subsets will be distinguished: $S,X,F$ and a single
vertex $\pin$, called the pin.  The meaning of these sets is described
in Section~\ref{section:nlogn}. On the data structure we can perform
the following updates: adding/removing vertices to $S,X,F$ and
marking/unmarking a vertex as a pin.  In the following table we gather
the tables used by the algorithm, together with an overview of the
running times of updates.  The meaning of the table entries uses
terminology that is described in the following sections.

The following lemma follows from each of the entries in the table
below, and will be proved in this section:
\begin{lemma}
  \label{lemma:ds:init-time}
  The data structure can be initialized in $O(c^k n)$ time.
\end{lemma}

\vskip 0.3cm

\noindent\begin{tabular}{| c | c | c | c |}
\hline
{\bf{Table}} & {\bf{Meaning}} & {\bf{Update}} & {\bf{Initialization}} \\
\hline
$P[i]$ & Boolean value $\pin\in W_i$ & $O(t\cdot \log n)$ & $O(t\cdot n)$ \\
\hline
$C[i][(S_i,U_i)]$ & Connectivity information on $U^{\ext}_i$ & $O(3^t\cdot t^{O(1)}\cdot \log n)$ & $O(3^t\cdot t^{O(1)}\cdot n)$ \\
\hline
$CardU[i][(S_i,U_i)]$ & Integer value $|U^{\ext}_i\cap W_i|$ & $O(3^t\cdot t^{O(1)}\cdot \log n)$ & $O(3^t\cdot t^{O(1)}\cdot n)$ \\
\hline
$T_1[i][(S_i,U_i)]$ & Table for query \qnei & $O(3^t\cdot k^{O(1)}\cdot \log n)$ & $O(3^t\cdot k^{O(1)}\cdot  n)$ \\
\hline
$T_2[i][(S_i,U_i)][\psi]$ & Table for query \qSsep & $O(9^t\cdot k^{O(1)}\cdot \log n)$ & $O(9^t\cdot k^{O(1)}\cdot  n)$ \\
\hline
$T_3[i][(S_i,U_i,X_i,F_i)]$ & Table for query \qpin & $O(6^t\cdot t^{O(1)}\cdot \log n)$ & $O(6^t\cdot t^{O(1)}\cdot n)$ \\
\hline
$T_4[i][(S_i,U_i)][\psi]$ & Table for query \qUsep & $O(5^t\cdot k^{O(1)}\cdot \log n)$ & $O(5^t\cdot k^{O(1)}\cdot n)$ \\
\hline
\end{tabular}

\vskip 0.3cm

We now proceed to description of the description of the table $P$, and
then to the two tables $C$ and $CardU$ that handle the important
component $U$.  The tables $T_1,T_2,T_3$ are described together with the
description of realization of the corresponding queries.  Whenever
describing the table, we argue how the table is updated during updates
of the data structure, and initialized in the beginning.

\subsection{The table $P$}

In the table $P$, for every node $i$ of the tree decomposition we store a
boolean value $P[i]$ equal to $(\pin\in W_i)$.  We now show how to
maintain the table $P$ when the data structure is updated. The table $P$
needs to be updated whenever the pin $\pin$ is marked or unmarked.
Observe, that the only nodes $i$ for which the information whether
$\pin\in W_i$ changed, are the ones on the path from $r_\pin$ to the
root of the tree decomposition.  Hence, we can simply follow this path
and update the values.  As the tree decomposition has depth $O(t\log
n)$, this update can be performed in $O(t\cdot \log n)$ time.  As when
the data structure is initialized, no pin is assigned, $P$ is
initially filled with $\bot$.

\subsection{Maintaining the important component $U$}

Before we proceed to the description of the queries, let us describe
what is the reason of introducing the pin $\pin$.  During the
computation, the algorithm recursively considers smaller parts of the
graph, separated from the rest via a small separator: at each step we
have distinguished set $S$ and we consider only one connected
component $U$ of $G\setminus S$.  Unfortunately, we cannot afford
recomputing the tree decomposition of $U$ at each recurrence call, or
even listing the vertices of $U$.  Therefore we employ a different
strategy for identification of $U$.  We will distinguish one vertex of
$U$ as a representative pin $\pin$, and $U$ can then be defined as the
set of vertices reachable from $\pin$ in $G\setminus S$.  Instead of
recomputing $U$ at each recursive call we will simply change the
pin.

In the tables, for each node $i$ of the tree decomposition we store
entries for each possible intersection of $U$ with $B_i$, and in this
manner we are prepared for every possible interaction of $U$ with
$G_i$.  In this manner, changing the pin can be done more efficiently.
Information needs to be recomputed on two paths to the root in the
tree decomposition, corresponding the previous and the next pin, while
for subtrees unaffected by the change we do not need to recompute
anything as the tables stored there already contain information about
the new $U$ as well --- as they contain information for {\emph{every
    possible}} new $U$.  As the tree decomposition is of logarithmic
depth, the update time is logarithmic instead of linear.

We proceed to the formal description.  We store the information about
$U$ in two special tables: $C$ and $CardU$.  As we intuitively
explained, tables $C$ and $CardU$ store information on the
connectivity behavior in the subtree, for every possible interaction
of $U$ with the bag.  Formally, for every node of the tree
decomposition $i$ we store an entry for every member of the family of
{\emph{signatures}} of the bag $B_i$.  A signature of the bag $B_i$ is
a pair $(S_i,U_i)$, such that $S_i,U_i$ are disjoint subsets of $B_i$.
Clearly, the number of signatures is at most $3^{|B_i|}$.

Let $i$ be a node of the tree decomposition.  For a signature
$\phi=(S_i,U_i)$ of $B_i$, let $S^{\ext}_i=S_i\cup (S\cap W_i)$ and
$U^{\ext}_i$ consists of all the vertices reachable in $G_i\setminus
S^{\ext}_i$ from $U_i$ or $\pin$, providing that it belongs to $W_i$.
Sets $S^{\ext}_i$ and $U^{\ext}_i$ are called {\emph{extensions}} of
the signature $\phi$; note that given $S_i$ and $U_i$, the extensions
are defined uniquely.  We remark here that the definition of
extensions depend not only on $\phi$ but also on the node $i$; hence,
we will talk about extensions of signatures only when the associated
node is clear from the context.

We say that signature $\phi$ of $B_i$ with extensions $S^{\ext}_i$ and
$U^{\ext}_i$ is {\emph{valid}} if it holds that
\begin{itemize}
\item[(i)] $U^\ext_i\cap B_i=U_i$,
\item[(ii)] if $U_i\neq \emptyset$ and $\pin\in W_i$ (equivalently, $P[i]$ is 
  true), then the
  component of $G[U^{\ext}_i]$ that contains $\pin$ contains also at
  least one vertex of $U_i$.
\end{itemize}
Intuitively, invalidity means that $\phi$ cannot contain consistent
information about intersection of $U$ and $G_i$.  The second condition
says that we cannot fully forget the component of $\pin$, unless the
whole $U^{\ext}_i$ is already forgotten.

Formally, the following invariant explains what is stored in tables
$C$ and $CardU$:
\begin{itemize}
\item if $\phi$ is invalid then $C[i][\phi]=CardU[i][\phi]=\bot$;
\item otherwise, $C[i][\phi]$ contains an equivalence relation $R$
  consisting of all pairs of vertices $(a,b)\in U_i$ that are
  connected in $G_i[U^{\ext}_i]$, while $CardU[i][\phi]$ contains
  $|U^{\ext}_i\cap W_i|$.
\end{itemize}

Note that in this definition we actually ignore the information about
alignment of vertices of $B_i$ to set $S,F,X$ in the current state of
the graph: the stored information depends only on the alignment of forgotten
vertices and the signature of the bag that overrides the actual
information in the current state.  In this manner we are prepared for
possible changes in the data structure, as after an update some other
signature will reflect the current state of the graph.  Moreover, it
is clear from this definition that during the computation, the
alignment of every vertex $v$ in the current state of the graph is
being checked only in the single node $r_v$ when this vertex is being
forgotten; we use this property heavily to implement the updates
efficiently enough.

We now explain how for every node $i$, values in tables $C[i]$ and
$CardU[i]$ can be computed using entries for the children of $i$.
These formulas will be crucial both for implementing updates and
initialization.  We consider different cases, depending on the type of
node $i$. 

\vskip 0.3cm
\noindent{\bf{Case 1: Leaf node.}} If $i$ is a leaf node then
$C[i][(\emptyset,\emptyset)]=\emptyset$ and $CardU[i][(\emptyset,\emptyset)]=0$.  \vskip 0.3cm

\noindent{\bf{Case 2: Introduce node.}} Let $i$ be a node that
introduces vertex $v$, and $j$ be its only child.  Consider some
signature $\phi=(S_i,U_i)$ of $B_i$; we would like to compute
$R_i = C[i][\phi]$.  Let $\phi'$ be a natural projection of $\phi$ onto
$B_j$, that is, $\phi'=(S_i\cap B_j, U_i\cap B_j)$.  Let
$R_j = C[j][\phi']$.  We consider some sub-cases, depending on the
alignment of $v$ in $\phi$.

\vskip 0.1cm {\bf{Case 2.1: $v\in S_i$.}} If we introduce a vertex from
$S_i$, then it follows that extensions of $U_i=U_j$ are equal.
Therefore, we can put $C[i][\phi]=C[j][\phi']$ and
$CardU[i][\phi]=CardU[j][\phi']$.

\vskip 0.1cm {\bf{Case 2.2: $v\in U_i$.}} In the beginning we check
whether conditions of validity are not violated.  First, if $v$ is the
only vertex of $U_i$ and $P[i]=\true$, then we simply put
$C[i][\phi]=\bot$: condition (ii) of validity is violated.  Second, we
check whether $v$ is adjacent only to vertices of $S_j$ and $U_j$; if
this is not the case, we put $C[i][\phi]=\bot$ as condition (i) of
validity is violated.

If the validity checks are satisfied, we can infer that the extension
$U^{\ext}_i$ of $U_i$ is extension $U^{\ext}_j$ of $U_j$ with $v$
added; this follows from the fact that $B_j$ separates $v$ from $W_j$,
so the only vertices of $U^{\ext}_i$ adjacent to $v$ are already
belonging to $U_j$.  Now we would like to compute the equivalence
relation $R_i$ out of $R_j$.  Observe that $R_i$ should be basically
$R_j$ augmented by connections introduced by the new vertex $v$
between its neighbors in $B_j$.  Formally, $R_i$ may be obtained from
$R_j$ by merging equivalence classes of all the neighbors of $v$ from
$U_j$, and adding $v$ to the obtained equivalence class; if $v$ does
not have any neighbors in $U_j$, we put it as a new singleton
equivalence class.  Clearly, $CardU[i][\phi]=CardU[j][\phi']$.

\vskip 0.1cm

{\bf{Case 2.3: $v\in B_i\setminus (S_i\cup U_i)$.}} We first check whether the
validity constraints are not violated.  As $v$ is separated from $W_j$
by $B_j$, the only possible violation introduced by $v$ is that $v$ is
adjacent to a vertex from $U_j$.  In this situation we put
$C[i][\phi]=CardU[i][\phi]=\bot$, and otherwise we can put
$C[i][\phi]=C[j][\phi']$ and $CardU[i][\phi]=CardU[j][\phi']$, because
extensions of $\phi$ and $\phi'$ are equal.

\vskip 0.3cm

\noindent{\bf{Case 3: Forget node.}} Let $i$ be a node that forgets
vertex $w$, and $j$ be its only child.  Consider some signature
$\phi=(S_i,U_i)$ of $B_i$ and define extensions $S^{\ext}_i$,
$U^{\ext}_i$ for this signature.  Observe that there is at most one
valid signature $\phi'=(S_j,U_j)$ of $B_j$ for which
$S^{\ext}_j=S^{\ext}_i$ and $U^{\ext}_j=U^{\ext}_i$, and this
signature is simply $\phi$ with $w$ added possibly to $S_i$ or $U_i$,
depending whether it belongs to $S^{\ext}_i$ or $U^{\ext}_i$: the
three candidates are $\phi_S=(S_i\cup \{w\},U_i)$,
$\phi_U=(S_i,U_i\cup \{w\})$ and $\phi_0=(S_i,U_i)$.
Moreover, if $\phi$ is valid then so is $\phi'$.  Formally, in the
following manner we can define signature $\phi'$, or conclude that $\phi$
is invalid:
\begin{itemize}
\item if $w\in S$, then $\phi'=\phi_S$;
\item otherwise, if $w=\pin$ then $\phi'=\phi_U$;
\item otherwise, we look into entries $C[j][\phi_U]$ and $C[j][\phi_0]$.  If
  \begin{itemize}
  \item[(i)] $C[j][\phi_U]=C[j][\phi_0]=\bot$ then $\phi$ is invalid, and we 
    put
    $C[i][\phi]=CardU[i][\phi]=\bot$;
  \item[(ii)] if $C[j][\phi_U]=\bot$ or $C[j][\phi_0]=\bot$, we take 
    $\phi'=\phi_0$ or
    $\phi'=\phi_U$, respectively;
  \item[(iii)] if $C[j][\phi_U]\neq\bot$ and $C[j][\phi_0]\neq\bot$, it follows 
    that $w$
    must be a member of a component of $G_i\setminus S^{\ext}_i$ that
    is fully contained in $W_i$ and does not contain $\pin$.  Hence we
    take $\phi'=\phi_0$.
  \end{itemize}
\end{itemize}
The last point is in fact a check whether $w\in U^{\ext}_i$: whether $w$
is connected to a vertex from $U_i$ in $G_i$, can be looked up in
table $C[j]$ by adding or not adding $w$ to $U_i$, and checking the
stored connectivity information.  If $w\in S^{\ext}_i$ or $w\in
U^{\ext}_i$, we should be using the information for the signature with
$S_i$ or $U_i$ updated with $w$, and otherwise we do not need to add
$w$ anywhere.

As we argued before, if $\phi$ is valid then so does $\phi'$, hence if
$C[j][\phi']=\bot$ then we can take $C[i][\phi]=CardU[i][\phi]=\bot$.
On the other hand, if $\phi'$ is valid, then the only possibility for
$\phi$ to be invalid is when condition (ii) cease to be satisfied.
This could happen only if $\phi'=\phi_U$ and $w$ is in a singleton
equivalence class of $C[j][\phi']$ (note that then the connected
component corresponding to this class needs to necessarily contain
$\pin$, as otherwise we would have $\phi'=\phi_0$).
Therefore, if this is the case, we put
$C[i][\phi]=CardU[i][\phi]=\bot$, and otherwise we conclude that
$\phi$ is valid and move to defining $C[i][\phi]$ and $CardU[i][\phi]$.

Let now $R_j=C[j][\phi']$.  As extensions of $\phi'$ and $\phi$ are equal, it
follows directly from the maintained invariant that $R_i$ is equal to
$R_j$ with $w$ removed from its equivalence class.  Moreover,
$CardU[i][\phi]$ is equal to $CardU[j][\phi']$, possibly incremented
by $1$ if we concluded that $\phi'=\phi_U$.

\vskip 0.3cm

\noindent{\bf{Case 4: Join node.}} Let $i$ be a join node and
$j_1,j_2$ be its two children.  Consider some signature
$\phi=(S_i,U_i)$ of $B_i$.  Let $\phi_1=(S_i,U_i)$ be a signature of
$B_{j_1}$ and $\phi_2=(S_i,U_i)$ be a signature of $B_{j_2}$.  From
the maintained invariant it follows that $C[i][\phi]$ is a minimum
transitive closure of $C[j_1][\phi_1]\cup C[j_2][\phi_2]$, or $\bot$
if any of these entries contains $\bot$.  Similarly,
$CardU[i][\phi]=CardU[j_1][\phi_1]+CardU[j_2][\phi_2]$.

\vskip 0.3cm

We now explain how to update tables $C$ and $CardU$ in $O(3^t\cdot
t^{O(1)}\cdot \log n)$ time.  We perform a similar strategy as with
table $P$: whenever some vertex $v$ is included or
removed from $S$, or marked or unmarked as a pin, we follow the path
from $r_v$ to the root and fully recompute the whole tables $C,CardU$
in the traversed nodes using the formulas presented above.  At each
step we recompute the table for some node using the tables of its
children; these tables are up to date since they did not need an
update at all, or were updated in the previous step.  Observe that
since the alignment of $v$ in the current state of the graph is
accessed only in computation for $r_v$, the path from $r_v$ to the
root of the decomposition consists of all the nodes for which the
tables should be recomputed.  Note also that when marking or unmarking
the pin $\pin$, we must first update $P$ and then $C$ and $CardU$.
The update takes $O(3^t\cdot t^{O(1)}\cdot \log n)$ time:
re-computation of each table takes $O(3^t\cdot t^{O(1)})$ time, and we
perform $O(t\log n)$ re-computations as the tree decomposition has
depth $O(t\log n)$.

Similarly, tables $C$ and $CardU$ can be initialized in $O(3^t\cdot
t^{O(1)}\cdot n)$ time by processing the tree in a bottom-up manner:
for each node of the tree decomposition, in $O(3^t\cdot t^{O(1)})$
time we compute its table based on the tables of the children, which
were computed before.

\subsection{Queries}

In our data structure we store one table per each query.  When
describing every query, we first introduce the formal invariant on
storage of table's entry, and how this stored information can be
computed based on the entries for children.  We then shortly discuss
performing updates and initialization of the tables, as they are in
all the cases based on the same principle as with tables $C$ and
$CardU$.  Queries themselves can be performed by reading a single
entry of the data structure, with the exception of query \qUsep, whose
implementation is more complex.


\subsubsection{Query \qnei}

\newcommand{\tldr}{\XBox}

We begin the description of the queries with the simplest one, namely
\qnei{}.  This query lists all the vertices of $S$ that are adjacent
to $U$.  In the algorithm we have an implicit bound on the size of
this neighborhood, which we can use to cut the computation when the
accumulated list grows too long.  We use $\ell$ to denote this bound;
in our case we have that $\ell=O(k)$.

\defquery{\qnei}{A list of vertices of $N(U)\cap S$, or marker '$\tldr$'
  if their number is larger than $\ell$.}{$O(\ell)$}

Let $i$ be a node of the tree decomposition, let $\phi=(S_i,U_i)$ be a
signature of $B_i$, and let $U^{\ext}_i, S^{\ext}_i$ be extensions of
this signature.  In entry $T_1[i][\phi]$ we store the following:
\begin{itemize}
\item if $\phi$ is invalid then $T_1[i][\phi]=\bot$;
\item otherwise $T_1[i][\phi]$ stores the list of elements of
  $N(U^{\ext}_i)\cap S^{\ext}_i$ if there is at most $\ell$ of them,
  and $\tldr$ if there is more of them.
\end{itemize}
Note that the information whether $\phi$ is invalid can be looked up in
table $C$.  The return value of the query is stored in $T[r][(\emptyset,\emptyset)]$.

We now present how to compute entries of table $T_1$ for every node
$i$ depending on the entries of children of $i$.  We consider
different cases, depending of the type of node $i$.  For every case,
we consider only signatures that are valid, as for the invalid ones we
just put value $\bot$.  

\vskip 0.3cm
\noindent{\bf{Case 1: Leaf node.}} If $i$ is a leaf node then
$T_1[i][(\emptyset,\emptyset)]=\emptyset$.  \vskip 0.3cm

\noindent{\bf{Case 2: Introduce node.}} Let $i$ be a node that
introduces vertex $v$, and $j$ be its only child.  Consider some
signature $\phi=(S_i,U_i)$ of $B_i$; we would like to compute
$T_1[i][\phi]=L_i$.  Let $\phi'$ be a natural intersection of $\phi$
with $B_j$, that is, $\phi'=(S_i\cap B_j, U_i\cap B_j)$.  Let
$T_1[j][\phi']=L_j$.  We consider some sub-cases, depending on the
alignment of $v$ in $\phi$.

\vskip 0.1cm {\bf{Case 2.1: $v\in S_i$.}} If we introduce a vertex from
$S_i$, we have that $U$-extensions of $\phi$ and $\phi'$ are equal.
It follows that $L_i$ should be simply list $L_j$ with $v$ appended if
it is adjacent to any vertex of $U_j=U_i$.  Note here that $v$ cannot
be adjacent to any vertex of $U^{\ext}_i\setminus U_i$, as $B_j$
separates $v$ from $W_j$.  Hence, we copy the list $L_j$ and append
$v$ if it is adjacent to any vertex of $U_j$ and $L_j\neq \tldr$.
However, if the length of the new list exceeds the $\ell$ bound, we
replace it by $\tldr$.  Note that copying the list takes $O(\ell)$
time, as its length is bounded by $\ell$.

\vskip 0.1cm {\bf{Case 2.2: $v\in U_i$.}} If we introduce a vertex from
$U_i$, then possibly some vertices of $S_i$ gain a neighbor in
$U^{\ext}_i$.  Note here that vertices of $S^{\ext}_i\setminus S_i$
are not adjacent to the introduced vertex $v$, as $B_j$ separates $v$
from $W_j$.  Hence, we copy list $L_j$ and append to it all the
vertices of $S_i$ that are adjacent to $v$, but were not yet on $L_j$.
If we exceed the $\ell$ bound on the length of the list, we put
$\tldr$ instead.  Note that both copying the list and checking whether
a vertex of $S_i$ is on it can be done in $O(\ell)$ time, as its
length is bounded by $\ell$.

\vskip 0.1cm

{\bf{Case 2.3: $v\in B_i\setminus (S_i\cup U_i)$.}} In this case extensions of $\phi$
and $\phi'$ are equal, so it follows from the invariant that we may
simply put $T[i][\phi]=T[j][\phi']$.

\vskip 0.3cm

\noindent{\bf{Case 3: Forget node.}} Let $i$ be a node that forgets
vertex $w$, and $j$ be its only child.  Consider some signature
$\phi=(S_i,U_i)$ of $B_i$.  Define $\phi'$ in the same manner as in
the Forget step in the computation of $C$.  As extensions of $\phi$
and $\phi'$ are equal, it follows that $T_1[i][\phi]=T_1[j][\phi']$.

\vskip 0.3cm

\noindent{\bf{Case 4: Join node.}} Let $i$ be a join node and
$j_1,j_2$ be its two children.  Consider some signature
$\phi=(S_i,U_i)$ of $B_i$.  Let $\phi_1=(S_i,U_i)$ be a signature of
$B_{j_1}$ and $\phi_2=(S_i,U_i)$ be a signature of $B_{j_2}$.  It
follows that $T_1[i][\phi]$ should be the merge of lists
$T_1[j_1][\phi_1]$ and $T_1[j_2][\phi_2]$, where we remove the
duplicates.  Of course, if any of these entries contains $\tldr$, we
simply put $\tldr$.  Otherwise, the merge can be done in $O(\ell)$
time due to the bound on lengths of $T_1[j_1][\phi_1]$ and
$T_1[j_2][\phi_2]$, and if the length of the result exceeds the bound
$\ell$, we replace it by $\tldr$.

\vskip 0.3cm

Similarly as before, for every addition/removal of vertex $v$ to/from
$S$, or marking/unmarking $v$ as a pin, we can update table $T_1$ in
$O(3^t\cdot k^{O(1)}\cdot \log n)$ time by following the path from
$r_v$ to the root and recomputing the tables in the traversed nodes.
Also, $T_1$ can be initialized in $O(3^t\cdot k^{O(1)}\cdot n)$ time
by processing the tree decomposition in a bottom-up manner and
applying the formula for every node.  Note that updating/initializing
table $T_1$ must be performed after updating/initializing tables $P$
and $C$.

\subsubsection{Query \qSsep}\label{sec:qssep}


We now move to the next query, namely finding a balanced
$S$-separator.  By Lemma~\ref{lemma:halfhalf}, as $G[U\cup S]$ has 
treewidth at most $k$, such a $\frac{1}{2}$-balanced $S$-separator
of size at most $k+1$ always exists.  We therefore implement the following query.

\defquery{\qSsep}{A list of elements of a $\frac{1}{2}$-balanced $S$-separator of
  $G[U\cup S]$ of size at most $k+1$, or $\bot$ if no such
  exists.}{$O(t^{O(1)})$}

Before we proceed to the implementation of the query, we show how to
translate the problem of finding a $S$-balanced separator into a
partitioning problem.

\begin{lemma}[Lemma~\ref{lem:balanced-3coloring}, restated]
  Let $G$ be a graph and $S\subseteq V(G)$.  Then a set $X$ is a balanced
  $S$-separator if and only if there exists a partition
  $(M_1,M_2,M_3)$ of $V(G)\setminus X$, such that there is no edge
  between $M_i$ and $M_j$ for $i\neq j$, and $|M_i\cap S|\leq |S|/2$
  for $i=1,2,3$.
\end{lemma}

The following combinatorial observation is crucial in the proof of
Lemma~\ref{lem:balanced-3coloring}.

\begin{lemma}\label{lem:comb-balanced-3coloring}
  Let $a_1,a_2,\dots,a_p$ be non-negative integers such that $\sum_{i=1}^p
  a_i = q$ and $a_i \leq q/2$ for $i=1,2,\ldots,p$.  Then there exists
  a partition of these integers into three sets, such that sum of
  integers in each set is at most $q / 2$.
\end{lemma}
\begin{proof}
  Without loss of generality assume that $p>3$, as otherwise the claim
  is trivial.
  We perform a greedy procedure as follows.  At each time step of the
  procedure we have a number of sets, maintaining an invariant that each
  set is of size at most $q/2$.  During the procedure we gradually
  merge the sets, i.e., we take two sets and replace them with their
  union.  We begin with each integer in its own set.  If we arrive at
  three sets, we end the procedure, thus achieving a feasible
  partition of the given integers.  We therefore need to present how
  the merging step is performed.
  
  At each step we choose the two sets with smallest sums of elements
  and merge them (i.e., replace them by their union).  As the number
  of sets is at least $4$, the sum of elements of the two chosen ones
  constitute at most half of the total sum, so after merging them we
  obtain a set with sum at most $q/2$.  Hence, unless the number of
  sets is at most $3$, we can always apply this merging step.
\end{proof}

\begin{proof}[Proof of Lemma~\ref{lem:balanced-3coloring}]
  One of the implications is trivial: if there is a partition
  $(M_1,M_2,M_3)$ of $G\setminus X$ with the given properties, then
  every connected component of $G\setminus X$ must be fully contained
  either in $M_1$, $M_2$, or $M_3$, hence it contains at most $|S|/2$
  vertices of $S$.  We proceed to the second implication.
  
  Assume that $X$ is a balanced $S$-separator of $G$ and let
  $C_1,C_2,\ldots,C_p$ be connected components of $G\setminus X$.  For
  $i=1,2,\ldots, p$, let $a_p=|S\cap C_i|$.  By
  Lemma~\ref{lem:comb-balanced-3coloring}, there exists a partition
  of integers $a_i$ into three sets, such that the sum of elements of
  each set is at most $|S|/2$.  If we partition vertex sets of
  components $C_1,C_2,\ldots,C_p$ in the same manner, we obtain a
  partition $(M_1,M_2,M_3)$ of $V(G)\setminus X$ with postulated properties.
\end{proof}

Lemma \ref{lem:balanced-3coloring} shows that, when looking for a
balanced $S$-separator, instead of trying to bound the number of
elements of $S$ in each connected component of $G[U\cup S]\setminus X$
separately, which could be problematic because of connectivity
condition, we can just look for a partition of $G[U\cup S]$ into four
sets with prescribed properties that can be checked locally.  This
suggest the following definition of table $T_2$.

In table $T_2$ we store entries for every node $i$ of the tree
decomposition, for every signature $\phi=(S_i,U_i)$ of $B_i$, and for
every $8$-tuple $\psi=(M_1,M_2,M_3,X,m_1,m_2,m_3,x)$ where
\begin{itemize}
\item $(M_1,M_2,M_3,X)$ is a partition of $S_i\cup U_i$,
\item $m_1,m_2,m_3$ are integers between $0$ and $|S|/2$,
\item and $x$ is an integer between $0$ and $k+1$.
\end{itemize}
This $8$-tuple $\psi$ will be called the {\emph{interface}}, and
intuitively it encodes the interaction of a potential solution with
the bag.  Observe that the set $U$ is not given in our graph directly but
rather via connectivity information stored in table $C$, so we need to
be prepared also for all the possible signatures of the bag; this is
the reason why we introduce the interface on top of the signature.
Note however, that the number of possible pairs $(\phi,\psi)$ is at
most $9^{|B_i|}\cdot k^{O(1)}$, so for every bag $B_i$ we store
$9^{|B_i|}\cdot k^{O(1)}$ entries.

We proceed to the formal definition of what is stored in table $T_2$.
For a fixed signature $\phi=(S_i,U_i)$ of $B_i$, let
$(S^{\ext}_i,U^{\ext}_i)$ be its extension, we say that partitioning
$(M^{\ext}_1, M^{\ext}_2, M^{\ext}_3, X^{\ext})$ of $S^{\ext}_i\cup
U^{\ext}_i$ is an {\emph{extension consistent}} with interface
$\psi=(M_1,M_2,M_3,X,m_1,m_2,m_3,x)$, if:
\begin{itemize}
\item $X^{\ext}\cap B_i=X$ and $M^{\ext}_j\cap B_i=M_j$ for $j=1,2,3$;
\item there is no edge between vertices of $M^{\ext}_j$ and
  $M^{\ext}_{j'}$ for $j\neq j'$;
\item $|X^{\ext}\cap W_i|=x$ and $|M_j^{\ext}\cap W_i|=m_j$ for $j=1,2,3$.
\end{itemize}

In entry $T_2[i][\phi][\psi]$ we store:
\begin{itemize}
\item $\bot$ if $\phi$ is invalid or no consistent extension of $\psi$ exists;
\item otherwise, a list of length $x$ of vertices of $X^{\ext}\cap W_i$
  in some consistent extension of $\psi$.
\end{itemize}

The query \qSsep can be realized in $O(t^{O(1)})$ time by checking entries
in the table $T$, namely 
$T[r][(\emptyset,\emptyset)][(\emptyset,\emptyset,\emptyset,\emptyset,m_1,m_2,m_3,x)]$
for all possible values $0\leq m_j\leq |S|/2$ and $0\leq x\leq k+1$, and
outputting the list contained in any of them that is not equal to
$\bot$, or $\bot$ if all of them are equal to $\bot$.

We now present how to compute entries of table $T_2$ for every node
$i$ depending on the entries of children of $i$.  We consider
different cases, depending of the type of node $i$.  For every case,
we consider only signatures that are valid, as for the invalid ones we
just put value $\bot$.  

\vskip 0.3cm
\noindent{\bf{Case 1: Leaf node.}} If $i$ is a leaf node then
$T_2[i][(\emptyset,\emptyset)][(\emptyset,\emptyset,\emptyset,\emptyset,0,0,0,0)]=\emptyset$,
and all the other interfaces are assigned $\bot$.  \vskip 0.3cm

\noindent{\bf{Case 2: Introduce node.}} Let $i$ be a node that
introduces vertex $v$, and $j$ be its only child.  Consider some
signature $\phi=(S_i,U_i)$ of $B_i$ and an interface
$\psi=(M_1,M_2,M_3,X,m_1,m_2,m_3,x)$; we would like to compute
$T_2[i][\phi][\psi]=L_i$.  Let $\phi',\psi'$ be natural intersections
of $\phi,\psi$ with $B_j$, respectively, that is, $\phi'=(S_i\cap B_j,
U_i\cap B_j)$ and $\psi'=(M_1\cap B_j,M_2\cap B_j,M_3\cap B_j,X\cap
B_j,m_1,m_2,m_3,x)$.  Let $T_2[j][\phi'][\psi']=L_j$.  We consider
some sub-cases, depending on the alignment of $v$ in $\phi$ and $\psi$.
The cases with $v$ belonging to $M_1$, $M_2$ and $M_3$ are symmetric,
so we consider only the case for $M_1$.

\vskip 0.1cm {\bf{Case 2.1: $v\in X$.}} Note that every extension
consistent with interface $\psi$ is an extension consistent with
$\psi'$ after trimming to $G_j$.  On the other hand, every extension
consistent with $\psi'$ can be extended to an extension consistent
with $\psi$ by adding $v$ to the extension of $X$.  Hence, it follows
that we can simply take $L_i=L_j$.

\vskip 0.1cm {\bf{Case 2.2: $v\in M_1$.}} Similarly as in the previous
case, every extension consistent with interface $\psi$ is an extension
consistent with $\psi'$ after trimming to $G_j$.  On the other hand,
if we are given an extension consistent with $\psi'$, we can add $v$
to $M_1$ and make an extension consistent with $\psi$ if and only if
$v$ is not adjacent to any vertex of $M_2$ or $M_3$; this follows from
the fact that $B_j$ separates $v$ from $W_j$, so the only vertices
from $M^{\ext}_2$, $M^{\ext}_3$ that $v$ could be possibly adjacent
to, lie in $B_j$.  However, if $v$ is adjacent to a vertex of $M_2$ or
$M_3$, we can obviously put $L_i=\bot$, as there is no extension
consistent with $\psi$: property that there is no edge between
$M^{\ext}_1$ and $M^{\ext}_3\cup M^{\ext}_3$ is broken already in the
bag.  Otherwise, by the reasoning above we can put $L_i=L_j$.

\vskip 0.1cm

{\bf{Case 2.3: $v\in B_i\setminus (S_i\cup U_i)$.}} Again, in this case we have
one-to-one correspondence of extensions consistent with $\psi$ with
$\psi'$ after trimming to $B_j$, so we may simply put $L_i=L_j$.

\vskip 0.3cm

\noindent{\bf{Case 3: Forget node.}} Let $i$ be a node that forgets
vertex $w$, and $j$ be its only child.  Consider some signature
$\phi=(S_i,U_i)$ of $B_i$, and some interface
$\psi=(M_1,M_2,M_3,X,m_1,m_2,m_3,x)$; we would like to compute
$T_2[i][\phi][\psi]=L_i$.  Let $\phi'=(S_j,U_j)$ be the only extension of 
signature $\phi$ to $B_j$ that has the same extension as
$\phi$; $\phi'$ can be deduced by looking up which signatures are
found valid in table $C$ in the same manner as in the forget step for
computation of table $C$.  We consider three cases depending on
alignment of $w$ in $\phi'$:

\vskip 0.1cm {\bf{Case 3.1: $w\notin S_j\cup U_j$.}} If $w$ is not in $S_j\cup
U_j$, then it follows that we may put $L_i=T_2[j][\phi'][\psi']$:
extensions of $\psi$ consistent with $\psi$ correspond one-to-one to
extensions consistent with $\psi'$.

\vskip 0.1cm {\bf{Case 3.2: $w\in S_j$.}} Assume that there exist some
extension $(M^{\ext}_1, M^{\ext}_2, M^{\ext}_3, X^{\ext})$ consistent
with $\psi$.  In this extension, vertex $w$ is either in $M^{\ext}_1$,
$M^{\ext}_2$, $M^{\ext}_3$, or in $X^{\ext}$.  Let us define the
corresponding interfaces:
\begin{itemize}
\item $\psi_1=(M_1\cup \{w\},M_2,M_3,X,m_1-1,m_2,m_3,x)$;
\item $\psi_2=(M_1,M_2\cup \{w\},M_3,X,m_1,m_2-1,m_3,x)$;
\item $\psi_3=(M_1,M_2,M_3\cup \{w\},X,m_1,m_2,m_3-1,x)$;
\item $\psi_X=(M_1,M_2,M_3,X\cup \{w\},m_1,m_2,m_3,x-1)$.
\end{itemize}
If any of integers $m_1-1,m_2-1,m_3-1,x-1$ turns out to be negative,
we do not consider this interface.  It follows that for at least one
$\psi'\in \{\psi_1,\psi_2,\psi_3,\psi_X\}$ there must be an extension
consistent with $\psi'$: it is just the extension $(M^{\ext}_1,
M^{\ext}_2, M^{\ext}_3, X^{\ext})$.  On the other hand, any extension
consistent with any of interfaces $\psi_1,\psi_2,\psi_3,\psi_X$ is
also consistent with $\psi$.  Hence, we may simply put
$L_i=T_2[i][\phi'][\psi']$, and append $w$ on the list in case
$\psi'=\psi_X$.

\vskip 0.1cm {\bf{Case 3.3: $w\in U_j$.}} We proceed in the same manner
as in Case 3.2, with the exception that we do not decrement $m_j$ by
$1$ in interfaces $\psi_j$ for $j=1,2,3$.

\vskip 0.3cm

\noindent{\bf{Case 4: Join node.}} Let $i$ be a join node and
$j_1,j_2$ be its two children.  Consider some signature
$\phi=(S_i,U_i)$ of $B_i$, and an interface
$\psi=(M_1,M_2,M_3,X,m_1,m_2,m_3,x)$; we would like to compute
$T_2[i][\phi][\psi]=L_i$.  Let $\phi_1=(S_i,U_i)$ be a signature of
$B_{j_1}$ and $\phi_2=(S_i,U_i)$ be a signature of $B_{j_2}$.  Assume
that there is some extension $(M^{\ext}_1, M^{\ext}_2, M^{\ext}_3,
X^{\ext})$ consistent with $\psi$.  Define $m^p_q=|W_{j_p}\cap M_q|$
and $x^p=|W_{j_p}\cap X|$ for $p=1,2$ and $q=1,2,3$; note that
$m^1_q+m^2_q=m_q$ for $q=1,2,3$ and $x^1+x^2=x$.  It follows that in
$G_{j_1}$, $G_{j_2}$ there are some extensions consistent with
$(M_1,M_2,M_3,X,m^1_1,m^1_2,m^1_3,x^1)$ and
$(M_1,M_2,M_3,X,m^2_1,m^2_2,m^2_3,x^2)$, respectively --- these are
simply extension $(M^{\ext}_1, M^{\ext}_2, M^{\ext}_3, X^{\ext})$
intersected with $V_i,V_j$, respectively.  On the other hand, if we
have some extensions in $G_{j_1}$, $G_{j_2}$ consistent with
$(M_1,M_2,M_3,X,m^1_1,m^1_2,m^1_3,x^1)$ and
$(M_1,M_2,M_3,X,m^2_1,m^2_2,m^2_3,x^2)$ for numbers $m^q_p,x^p$ such
that $m^1_q+m^2_q=m_q$ for $q=1,2,3$ and $x^1+x^2=x$, then the
point-wise union of these extensions is an extension consistent with
$(M_1,M_2,M_3,X,m_1,m_2,m_3,x)$.  It follows that in order to compute
$L_i$, we need to check if for any such choice of $m^q_p,x^p$ we have
non-$\bot$ entries in
$T_2[j_1][\phi_1][(M_1,M_2,M_3,X,m^1_1,m^1_2,m^1_3,x^1)]$ and
$T_2[j_2][\phi_2][(M_1,M_2,M_3,X,m^2_1,m^2_2,m^2_3,x^2)]$.  This is
the case, we put the union of the lists contained in these entries as
$L_i$, and otherwise we put $\bot$.  Note that computing the union of
these lists takes $O(k)$ time as their lengths are bounded by $k$, and
there is $O(k^4)$ possible choices of $m^q_p,x^p$ to check.

\vskip 0.3cm

Similarly as before, for every addition/removal of vertex $v$ to/from
$S$ or marking/unmarking $v$ as a pin, we can update table $T_2$ in
$O(9^t\cdot k^{O(1)}\cdot \log n)$ time by following the path from
$r_v$ to the root and recomputing the tables in the traversed nodes.
Also, $T_2$ can be initialized in $O(9^t\cdot k^{O(1)}\cdot n)$ time
by processing the tree decomposition in a bottom-up manner and
applying the formula for every node.  Note that updating/initializing
table $T_2$ must be performed after updating/initializing tables $P$
and $C$.

\subsubsection{Query \qpin}

We now proceed to the next query.  Recall that at each point, the
algorithm maintains the set $F$ of vertices marking components of
$G[U\cup S]\setminus (X\cup S)$ that have been already processed.  A
component is marked as processed when one of its vertices is added to
$F$.  Hence, we need a query that finds the next component to process
by returning any of its vertices.  As in the linear-time approximation
algorithm we need to process the components in decreasing order of
sizes, the query in fact provides a vertex of the largest component.

\defquery{\qpin}{A pair $(u,\ell)$, where (i) $u$ is a vertex of a
  component of $G[U\cup S]\setminus (X\cup S)$ that does not contain a
  vertex from $F$ and is of maximum size among such components, and
  (ii) $\ell$ is the size of this component; or, $\bot$ if no such
  component exists.}{$O(1)$}

To implement the query we create a table similar to table $C$, but
with entry indexing enriched by subsets of the bag corresponding to
possible intersections with $X$ and $F$.  Formally, we store entries
for every node $i$, and for every signature $\phi=(S_i,U_i,X_i,F_i)$,
which is a quadruple of subsets of $B_i$ such that (i) $S_i\cap
U_i=\emptyset$, (ii) $X_i\subseteq S_i\cup U_i$, (iii) $F_i\subseteq
U_i\setminus X_i$.  The number of such signatures is equal to
$6^{|B_i|}$.

For a signature $\phi=(S_i,U_i,X_i,F_i)$, we say that
$(S^{\ext}_i,U^{\ext}_i,X^{\ext}_i,F^{\ext}_i)$ is the extension of
$\phi$ if (i) $(S^{\ext}_i,U^{\ext}_i)$ is the extension of
$(S_i,U_i)$ as in the table $C$, (ii) $X^{\ext}_i=X_i\cup (W_i\cap X)$
and $F^{\ext}_i=F_i\cup (W_i\cap F)$.  We may now state what is stored
in entry $T_3[i][(S_i,U_i,X_i,F_i)]$:
\begin{itemize}
\item if $(S_i,U_i)$ is invalid then we store $\bot$;
\item otherwise we store:
  \begin{itemize}
  \item an equivalence relation $R$ between vertices of $U_i\setminus X_i$,
    such that $(v_1,v_2)\in R$ if and only if $v_1,v_2$ are connected
    in $G[U^{\ext}_i\setminus X^{\ext}_i]$;
  \item for every equivalence class $K$ of $R$, an integer $m_K$ equal
    to the number of vertices of the connected component of
    $G[U^{\ext}_i\setminus X^{\ext}_i]$ containing $K$, which are
    contained in $W_i$, or to $\bot$ if this connected component
    contains a vertex of $F^{\ext}_i$;
  \item a pair $(u,m)$, where $m$ is equal to the size of the largest
    component of $G[U^{\ext}_i\setminus X^{\ext}_i]$ not containing
    any vertex of $F^{\ext}_i$ or $U_i$, while $u$ is any vertex of
    this component; if no such component exists, then $(u,m)=(\bot,\bot)$.
  \end{itemize}
\end{itemize}

Clearly, query \qpin may be implemented by outputting the pair $(u,m)$
stored in the entry
$T_3[r][(\emptyset,\emptyset,\emptyset,\emptyset)]$, or $\bot$ if this
pair is equal to $(\bot,\bot)$.

We now present how to compute entries of table $T_3$ for every node
$i$ depending on the entries of children of $i$.  We consider
different cases, depending of the type of node $i$.  For every case,
we consider only signatures $(S_i,U_i,X_i,F_i)$ for which $(S_i,U_i)$
is valid, as for the invalid ones we just put value $\bot$. 

\vskip 0.3cm
\noindent{\bf{Case 1: Leaf node.}} If $i$ is a leaf node then
$T_3[i][(\emptyset,\emptyset,\emptyset,\emptyset)]=(\emptyset,\emptyset,(\bot,\bot))$.  \vskip 0.3cm

\noindent{\bf{Case 2: Introduce node.}} Let $i$ be a node that
introduces vertex $v$, and $j$ be its only child.  Consider some
signature $\phi=(S_i,U_i,X_i,F_i)$ of $B_i$; we would like to compute
$T_3[i][\phi]=(R_i,(m^i_K)_{K\in R_i},(u_i,m_i))$.  Let $\phi'$ be a
natural projection of $\phi$ onto $B_j$, that is, $\phi'=(S_i\cap B_j,
U_i\cap B_j, X_i\cap B_j, F_i\cap B_j)$.  Let
$T_3[j][\phi']=(R_j,(m^j_K)_{K\in R_j},(u_j,m_j))$; note that this
entry we know, but entry $T_3[i][\phi]$ we would like to compute.  We
consider some sub-cases, depending on the alignment of $v$ in $\phi$.

\vskip 0.1cm {\bf{Case 2.1: $v\in U_i\setminus (X_i\cup F_i)$.}} If we introduce a
vertex from $U_i\setminus (X_i\cup F_i)$, then the extension of $\phi$
is just the extension of $\phi'$ plus vertex $v$ added to
$U^{\ext}_i$.  If we consider the equivalence classes of $R_i$, then
these are equivalence classes of $R_j$ but possibly some of them have
been merged because of connections introduced by vertex $v$.  As $B_j$
separates $v$ from $W_j$, $v$ could only create connections between
two vertices from $B_j\cap (U_j\setminus X_j)$.  Hence, we can obtain
$R_i$ from $R_j$ by merging all the equivalence classes of vertices of
$U_j\setminus X_j$ adjacent to $v$; the corresponding entry in
sequence $(m_K)_{K\in R_i}$ is equal to the sum of entries from the
sequence $(m^j_K)_{K\in R_j}$ corresponding to the merged classes.  If
any of these entries is equal to $\bot$, we put simply $\bot$.  If $v$
was not adjacent to any vertex of $U_j\setminus X_j$, we put $v$ in a
new equivalence class $K$ with $m_K=0$.  Clearly, we can also put
$(u_i,m_i)=(u_j,m_j)$.

\vskip 0.1cm {\bf{Case 2.2: $v\in (U_i\setminus X_i)\cap F_i$.}} We perform in the
same manner as in Case 2.2, with the exception that the new entry in
sequence $(m_K)_{K\in R_i}$ will be always equal to $\bot$, as the
corresponding component contains a vertex from $F^{\ext}_i$.

\vskip 0.1cm {\bf{Case 2.3: $v\in S_i\cup X_i$.}} In this case we can
simply put $T_3[i][\phi]=T_3[j][\phi']$ as the extensions of $\phi$
and $\phi'$ are the same with the exception of $v$ being included into
$X^{\ext}_i$ and/or into $S^{\ext}_i$, which does not influence
information to be stored in the entry.

\vskip 0.1cm

{\bf{Case 2.4: $v\in B_i\setminus (S_i\cup U_i)$.}} In this case we can simply put
$T_3[i][\phi]=T_3[j][\phi']$ as the extensions of $\phi$ and $\phi'$ are equal.

\vskip 0.3cm

\noindent{\bf{Case 3: Forget node.}} Let $i$ be a node that forgets
vertex $w$, and $j$ be its only child.  Consider some signature
$\phi=(S_i,U_i,X_i,F_i)$ of $B_i$; we would like to compute
$T_3[i][\phi]=(R_i,(m^i_K)_{K\in R_i},(u_i,m_i))$.  Let
$(S^{\ext}_i,U^{\ext}_i,X^{\ext}_i,F^{\ext}_i)$ be extension of
$\phi$.  Observe that there is exactly one signature
$\phi'=(S_j,U_j,X_j,F_j)$ of $B_j$ with the same extension as $\phi$,
and this signature is simply $\phi$ with $w$ added possibly to $S_i$,
$U_i$, $X_i$ or $F_i$, depending whether it belongs to $S^{\ext}_i$,
$U^{\ext}_i$, $X^{\ext}_i$, or $F^{\ext}_i$.  Coloring $\phi'$ may be
defined similarly as in case of forget node for table $C$; we just
need in addition to include $w$ in $X^{\ext}_i$ or $F^{\ext}_i$ if it
belongs to $X$ or $F$, respectively.

Let $T_3[j][\phi]=(R_j,(m^j_K)_{K\in R_j},(u_j,m_j))$.  As the extensions
of $\phi$ and $\phi'$ are equal, it follows that we may take $R_i$
equal to $R_j$ with $w$ possibly excluded from its equivalence class.
Similarly, for every equivalence class $K\in R_i$ we put $m^i_K$ equal
to $m^j_{K'}$, where $K'$ is the corresponding equivalence class of
$R_j$, except the class that contained $w$ which should get the
previous number incremented by $1$, providing it was not equal to
$\bot$.  We also put $(u_i,m_i)=(u_j,m_j)$ except the situation, when
we forget the last vertex of a component of $G[U^{\ext}_j\setminus
X^{\ext}_j]$: this is the case when $w$ is in $U_j\setminus X_j$ and
constitutes a singleton equivalence class of $R_j$.  Let then
$m^j_{\{w\}}$ be the corresponding entry in sequence $(m^j_K)_{K\in
  R_j}$.  If $m^j_{\{w\}}=\bot$, we simply put $(u_i,m_i)=(u_j,m_j)$.
Else, if $(u_j,m_j)=(\bot,\bot)$ or $m^j_{\{w\}}>m_j$, we put
$(u_i,m_i)=(w,m^j_{\{w\}})$, and otherwise we put
$(u_i,m_i)=(u_j,m_j)$.

\vskip 0.3cm

\noindent{\bf{Case 4: Join node.}} Let $i$ be a join node and
$j_1,j_2$ be its two children.  Consider some signature
$\phi=(S_i,U_i,X_i,F_i)$ of $B_i$; we would like to compute
$T_3[i][\phi]=(R_i,(m^i_K)_{K\in R_i},(u_i,m_i))$.  Let
$\phi_1=(S_i,U_i,X_i,F_i)$ be a signature of $B_{j_1}$ and
$\phi_2=(S_i,U_i,X_i,F_i)$ be a signature of $B_{j_2}$.  Let
$T_3[j_1][\phi_1]=(R_{j_1},(m^{j_1}_K)_{K\in
  R_{j_1}},(u_{j_1},m_{j_1}))$ and
$T_3[j_2][\phi_2]=(R_{j_2},(m^{j_2}_K)_{K\in
  R_{j_2}},(u_{j_2},m_{j_2}))$.  Note that equivalence relations
$R_{j_1}$ and $R_{j_2}$ are defined on the same set $U_i\setminus X_i$.  It
follows from the definition of $T_3$ that we can put:
\begin{itemize}
\item $R_i$ to be the minimum transitive closure of $R_{j_1}\cup
  R_{j_2}$;
\item for every equivalence class $K$ of $R_i$, $m^i_K$ equal to the
  sum of (i) numbers $m^{j_1}_{K_1}$ for $K_1\subseteq K$, $K_1$ being
  an equivalence class of $R_{j_1}$, and (ii) numbers $m^{j_2}_{K_2}$
  for $K_2\subseteq K$, $K_2$ being an equivalence class of $R_{j_2}$;
  if any of these numbers is equal to $\bot$, we put $m^i_K=\bot$;
\item $(u_i,m_i)$ to be equal to $(u_{j_1},m_{j_1})$ or
  $(u_{j_2},m_{j_2})$, depending whether $m_{j_1}$ or $m_{j_2}$ is
  larger; if any of these numbers is equal to $\bot$, we take the
  second one, and if both are equal to $\bot$, we put $(u_i,m_i)=(\bot,\bot)$.
\end{itemize}

\vskip 0.3cm

Similarly as before, for every addition/removal of vertex $v$ to/from
$S$, to/from $X$, to/from $F$, or marking/unmarking $v$ as a pin, we
can update table $T_3$ in $O(6^t\cdot t^{O(1)}\cdot \log n)$ time by
following the path from $r_v$ to the root and recomputing the tables
in the traversed nodes.  Also, $T_3$ can be initialized in $O(6^t\cdot
t^{O(1)}\cdot n)$ time by processing the tree decomposition in a
bottom-up manner and applying the formula for every node.  Note that
updating/initializing table $T_3$ must be performed after
updating/initializing tables $P$ and $C$.

\subsubsection{Query \qUsep}\label{sec:queries}

In this section we implement the last query, needed for the
linear-time algorithm; the query is significantly more involved than
the previous one.  The query specification is as follows:

\defquery{\qUsep}{A list of elements of a $\frac{8}{9}$-balanced
  separator of $G[U]$ of size at most $k+1$, or $\bot$ if no such
  exists.}{$O(c^t\cdot k^{O(1)}\cdot \log n)$}

  Note that Lemma~\ref{lemma:halfhalf} guarantees that in fact $G[U]$ contains 
  a
$\frac{1}{2}$-balanced separator of size at most $k+1$.  Unfortunately,
we are not able to find a separator with such a good guarantee on the
sizes of the sides; the difficulties are explained in 
Section~\ref{sec:outline}.
Instead, we again make use of the precomputed approximate tree
decomposition to find a balanced separator with slightly worse
guarantees on the sizes of the sides.

In the following we will also use the notion of a {\emph{balanced separation}}. 
For a graph $G$, we say that a partition $(L,X,R)$ of $V(G)$ is an 
$\alpha${\emph{-balanced separation of $G$}}, if there is no edge between $L$ 
and $R$, and $|L|,|R|\leq \alpha |V(G)|$. The {\emph{order}} of a separation is 
the size of $X$. Clearly, if $(L,X,R)$ is an $\alpha$-balanced separation of 
$G$, then $X$ is an $\alpha$-balanced separator of $G$. By folklore [see the 
proof of Lemma~\ref{lem:halfhalfalg}] we know that every graph of treewidth at 
most $k$ has a $\frac{2}{3}$-balanced separation of order at most $k+1$.

\paragraph{Expressing the search for a balanced separator as a
  maximization problem.}

Before we start explaining the query implementation, we begin with a
few definitions that enable us to express finding a balanced separator
as a simple maximization problem.

\begin{definition}
  Let $G$ be a graph, and $T_L,T_R$ be disjoint sets of terminals in
  $G$.  We say that a partition $(L,X,R)$ of $V(G)$ is a
  {\emph{terminal separation of $G$ of order $\ell$}}, if the
  following conditions are satisfied:
  \begin{itemize}
  \item[(i)] $T_L\subseteq L$ and $T_R\subseteq R$;
  \item[(ii)] there is no edge between $L$ and $R$;
  \item[(iii)] $|X|\leq \ell$.
  \end{itemize}
  We moreover say that $(L,X,R)$ is {\emph{left-pushed}}
  ({\emph{right-pushed}}) if $|L|$ ($|R|$) is maximum among possible terminal 
  separations of order $\ell$.
\end{definition}

Pushed terminal separations are similar to important separators of
Marx~\cite{Marx06}, and their number for fixed $T_L,T_R$ can be
exponential in $\ell$.  Pushed terminal separations are useful for us
because of the following lemma, that enables us to express finding a
small balanced separator as a maximization problem, providing that
some separator of a reasonable size is given.

\begin{lemma}\label{lem:c4vc}
  Let $G$ be a graph of treewidth at most $k$ and let $(A_1,B,A_2)$ be
  some separation of $G$, such that
  $|A_1|,|A_2|\leq \frac{3}{4}|V(G)|$.  Then there exists a partition
  $(T_L,X_B,T_R)$ of $B$ and integers $k_1,k_2$ with
  $k_1+k_2+|X_B|\leq k+1$, such that if $G_1,G_2$ are $G[A_1\cup
  (B\setminus X_B)]$ and $G[A_2\cup (B\setminus X_B)]$ with terminals
  $T_L,T_R$, then
  \begin{itemize}
  \item[(i)] there exist a terminal separations of $G_1,G_2$ of orders
    $k_1,k_2$, respectively;
  \item[(ii)] for any left-pushed terminal separation $(L_1,X_1,R_1)$
    of order $k_1$ in $G_1$ and any right-pushed separation
    $(L_2,X_2,R_2)$ of order $k_2$ in $G_2$, the triple $(L_1\cup
    T_L\cup L_2,X_1\cup X_B\cup X_2,R_1\cup T_R\cup R_2)$ is a
    terminal separation of $G$ of order at most $k+1$ with $|L_1\cup T_L\cup
    L_2|,|R_1\cup T_R\cup R_2|\leq \frac{7}{8}|V(G)|+\frac{|X|+(k+1)}{2}$.
  \end{itemize}
\end{lemma}
\begin{proof}
  As the treewidth of $G$ is at most $k$, there is a separation
  $(L,X,R)$ of $G$ such that $|L|,|R|\leq \frac{2}{3} |V(G)|$ and $|X|\leq k+1$ 
  by folklore [see the proof of lemma~\ref{lem:halfhalfalg}].  Let us set 
  $(T_L,X_B,T_R)=(L\cap B, X\cap B, R\cap B)$,
  $k_1=|X\cap A_1|$ and $k_2=|X\cap A_2|$.  Observe that $X\cap A_1$
  and $X\cap A_2$ are terminal separations in $G_1$ and $G_2$ of
  orders $k_1$ and $k_2$, respectively, hence we are done with (i).
  We proceed to the proof of (ii).
  
  Let us consider sets $L\cap A_1$, $L\cap A_2$, $R\cap A_1$ and $R\cap A_2$.
  Since $(A_1,B,A_2)$ and $(L,X,R)$ are $\frac{1}{4}$- and
  $\frac{1}{3}$- balanced separations, respectively, we know that:
  \begin{itemize}
  \item $|L\cap A_1|+|L\cap A_2|+|B|\geq \frac{1}{3}|V(G)|-(k+1)$;
  \item $|R\cap A_1|+|R\cap A_2|+|B|\geq \frac{1}{3}|V(G)|-(k+1)$;
  \item $|L\cap A_1|+|R\cap A_1|+(k+1)\geq \frac{1}{4}|V(G)|-|B|$;
  \item $|L\cap A_2|+|R\cap A_2|+(k+1)\geq \frac{1}{4}|V(G)|-|B|$.
  \end{itemize}
  We claim that either $|L\cap A_1|,|R\cap A_2|\geq
  \frac{1}{8}|V(G)|-\frac{|B|+(k+1)}{2}$, or $|L\cap A_2|,|R\cap A_1|\geq
  \frac{1}{8}|V(G)|-\frac{|B|+(k+1)}{2}$.  Assume first that $|L\cap
  A_1|<\frac{1}{8}|V(G)|-\frac{|B|+(k+1)}{2}$.  Observe that then $|L\cap
  A_2|\geq
  \frac{1}{3}|V(G)|-|B|-(k+1)-(\frac{1}{8}|V(G)|-\frac{|B|+(k+1)}{2})\geq
  \frac{1}{8}|V(G)|-\frac{|B|+(k+1)}{2}$.  Similarly, $|R\cap A_1|\geq
  \frac{1}{4}|V(G)|-|B|-(k+1)-(\frac{1}{8}|V(G)|-\frac{|B|+(k+1)}{2})\geq
  \frac{1}{8}|V(G)|-\frac{|B|+(k+1)}{2}$.  The case when $|R\cap
  A_2|<\frac{1}{8}|V(G)|-\frac{|B|+(k+1)}{2}$ is symmetric.  Without loss
  of generality, by possibly flipping separation $(L,X,R)$, assume
  that $|L\cap A_1|,|R\cap A_2|\geq \frac{1}{8}|V(G)|-\frac{|B|+(k+1)}{2}$.
  
  Let $(L_1,X_1,R_1)$ be any left-pushed terminal separation of order
  $k_1$ in $G_1$ and $(L_2,X_2,R_2)$ be any right-pushed terminal
  separation of order $k_2$ in $G_2$.  By the definition of being
  left- and right-pushed, we have that $|L_1\cap A_1|\geq |L\cap
  A_1|\geq \frac{1}{8}|V(G)|-\frac{|B|+(k+1)}{2}$ and $|R_2\cap A_2|\geq
  |R\cap A_2|\geq \frac{1}{8}|V(G)|-\frac{|B|+(k+1)}{2}$.  Therefore, we
  have that $|L_1\cup T_L\cup L_2|\leq
  \frac{7}{8}|V(G)|+\frac{|B|+(k+1)}{2}$ and $|L_1\cup T_L\cup L_2|\leq
  \frac{7}{8}|V(G)|+\frac{|B|+(k+1)}{2}$.
\end{proof}

The idea of the rest of the implementation is as follows.  First,
given an approximate tree decomposition of with $O(k)$ in the data
structure, in logarithmic time we will find a bag $B_{i_0}$ that
splits the component $U$ in a balanced way.  This bag will be used as
the separator $B$ in the invocation of Lemma~\ref{lem:c4vc}; the right
part of the separation will consist of vertices contained in the
subtree below $B_{i_0}$, while the whole rest of the tree will
constitute the left part.  Lemma~\ref{lem:c4vc} ensures us that we may
find some balanced separator of $U$ by running two maximization
dynamic programs: one in the subtree below $B_{i_0}$ to identify a
right-pushed separation, and one on the whole rest of the tree to find
a left-pushed separation.  As in all the other queries, we will store
tables of these dynamic programs in the data structure, maintaining them
with $O(c^t \log n)$ update times.

\paragraph{Case of a small $U$}
At the very beginning of the implementation of the query we read
$|U|$, which is stored in the entry $CardU[r][(\emptyset,\emptyset)]$.
If it turns out that $|U|< 36(k+t+2)=O(k)$, we perform the following
explicit construction.  We apply a DFS search from $\pin$ to identify
the whole $U$; note that this search takes $O(k^2)$ time, as $U$ and
$S$ are bounded linearly in $k$.  Then we build subgraph $G[U]$, which
again takes $O(k^2)$ time.  As this subgraph has $O(k)$ vertices and
treewidth at most $k$, we may find its $\frac{1}{2}$-balanced
separator of order at most $k+1$ in $c^k$ time using a brute-force
search through all the possible subsets of size at most $k+1$.  This
separator may be returned as the result of the query.  Hence, from now
on we assume that $|U|\geq 36(k+t+2)$.

\paragraph{Tracing $U$}
We first aim to identify bag $B_{i_0}$ in logarithmic time.  The
following lemma encapsulates the goal of this subsection.  Note that
we are not only interested in the bag itself, but also in the
intersection of the bag with of $S$ and $U$ (defined as the connected
component of $G\setminus S$ containing $\pin$).  While intersection
with $S$ can be trivially computed given the bag, we will need to
trace the intersection with $U$ inside the computation.

\begin{lemma}\label{lem:tracing}
  There exists an algorithm that, given access to the data structure,
  in $O(t^{O(1)}\cdot \log n)$ time finds a node $i_0$ of the tree
  decomposition such that $|U|/4 \leq |W_{i_0}\cap U|\leq |U|/2$
  together with two subsets $U_i,S_i$ of $B_{i_0}$ such that $U_0=U\cap
  B_{i_0}$ and $S_0=S\cap B_{i_0}$.
\end{lemma}
\begin{proof}
  The algorithm keeps track of a node $i$ of the tree decomposition
  together with a pair of subsets $(U_i,S_i)=(B_i\cap U,B_i\cap S)$
  being the intersections of the bag associated to the current node
  with $U$ and $S$, respectively.  The algorithm starts with the root
  node $r$ and two empty subsets, and iteratively traverses down the
  tree keeping an invariant that $CardU[i][(U_i,S_i)]\geq |U|/2$.
  Whenever we consider a join node $i$ with two sons $j_1,j_2$, we
  choose to go down to the node where $CardU[i_t][(U_{j_t},U_{j_t})]$
  is larger among $t=1,2$.  In this manner, at each step
  $CardU[i][(U_i,S_i)]$ can be decreased by at most $1$ in case of a
  forget node, or can be at most halved in case of a join node.  As
  $|U|\geq 36(k+t+2)$, it follows that the first node $i_0$ when the
  invariant $CardU[i][(U_i,S_i)]\geq |U|/2$ ceases to hold, satisfies
  $|U|/4 \leq CardU[i_0][(U_{i_0},S_{i_0})]\leq |U|/2$, and therefore
  can be safely returned by the algorithm.
  
  It remains to argue how sets $(U_i,S_i)$ can be updated at each step
  of the traverse down the tree.  Updating $S_i$ is trivial as we
  store an explicit table remembering for each vertex whether it
  belongs to $S$.  Therefore, now we focus on updating $U$.
  
  The cases of introduce and join nodes are trivial.  If $i$ is an
  introduce node with son $j$, then clearly $U_j=U_i\cap B_j$.
  Similarly, if $i$ is a join node with sons $j_1,j_2$, then
  $U_{j_1}=U_{j_2}=U_i$.  We are left with the forget node.
  
  Let $i$ be a forget node with son $j$, and let $B_j=B_i\cup \{w\}$.  We
  have that $U_j=U_i\cup \{w\}$ or $U_j=U_i$, depending whether $w\in
  U$ or not.  This information can be read from the table $C[j]$ as
  follows:
  \begin{itemize}
  \item if $C[j][(U_i\cup \{w\},S_j)]=\bot$, then $w\notin U$ and $U_j=U_i$;
  \item if $C[j][(U_i,S_j)]=\bot$, then $w\in U$ and $U_j=U_i\cup \{w\}$;
  \item otherwise, both $C[j][(U_i,S_j)]$ and $C[j][(U_i\cup \{w\},S_j)]$
    are not equal to $\bot$; this follows from the fact that at least
    one of them, corresponding to the correct choice whether $w\in U$
    or $w\notin U$, must be not equal to $\bot$.  Observe that in this
    case $w$ is in a singleton equivalence class of $C[j][(U_i\cup
    \{w\},S_j)]$, and the connected component of $w$ in the extension
    of $U_i\cup \{w\}$ cannot contain the pin $\pin$.  It follows that
    $w\notin U$ and we take $U_j=U_i$.
  \end{itemize}
  
  Computation at each step of the tree traversal takes $O(t^{O(1)})$
  time.  As the tree has logarithmic depth, the whole algorithm runs
  in $O(t^{O(1)}\cdot \log n)$ time.
\end{proof}

\paragraph{Dynamic programming for pushed separators}
In this subsection we show how to construct dynamic programming tables
for finding pushed separators.  The implementation resembles that of
table $T_2$, used for balanced $S$-separators.

In table $T_4$ we store entries for every node $i$ of the tree
decomposition, for every signature $\phi=(S_i,U_i)$ of $B_i$, and for
every $4$-tuple $\psi=(L,X,R,x)$, called again the {\emph{interface}},
where
\begin{itemize}
\item $(L,X,R)$ is a partition of $U_i$,
\item $x$ is an integer between $0$ and $k+1$.
\end{itemize}
Again, the intuition is that the interface encodes the interaction of
a potential solution with the bag.  Note that for every bag $B_i$ we
store at most $5^{|B_i|}\cdot (k+2)$ entries.

We proceed to the formal definition of what is stored in table $T_4$.
Let us fix a signature $\phi=(S_i,U_i)$ of $B_i$, and let
$(S^{\ext}_i,U^{\ext}_i)$ be its extension.  For an interface
$\psi=(L,X,R,x)$, we say that a terminal separation $(L^{\ext},
X^{\ext}, R^{\ext})$ in $G[U^{\ext}_i]$ with terminals $L,R$ is an
{\emph{extension consistent}} with interface $\psi=(L,X,R,x)$ if
\begin{itemize}
\item $L^{\ext}\cap B_i=L$, $X^{\ext}\cap B_i=X$ and $R^{\ext}\cap B_i=R$;
\item $|X^{\ext}\cap W_i|=x$.
\end{itemize}
Then entry $T_4[i][\phi][\psi]$ contains the pair $(r,X_0)$ where $r$ is the
maximum possible $|R^{\ext}\cap W_i|$ among extensions consistent with
$\psi$, and $X_0$ is the corresponding set $X^{\ext}\cap W_i$ for
which this maximum was attained, or $\bot$ if the signature $\phi$ is
invalid or no consistent extension exists.

We now present how to compute entries of table $T_4$ for every node
$i$ depending on the entries of children of $i$.  We consider
different cases, depending of the type of node $i$.  For every case,
we consider only signatures that are valid, as for the invalid ones we
just put value $\bot$.  

\vskip 0.3cm
\noindent{\bf{Case 1: Leaf node.}} If $i$ is a leaf node then
$T_4[i][(\emptyset,\emptyset)][(\emptyset,\emptyset,\emptyset,0)]=(0,\emptyset)$,
and all the other entries are assigned $\bot$.  \vskip 0.3cm

\noindent{\bf{Case 2: Introduce node.}} Let $i$ be a node that
introduces vertex $v$, and $j$ be its only child.  Consider some
signature $\phi=(S_i,U_i)$ of $B_i$ and an interface $\psi=(L,X,R,x)$;
we would like to compute $T_4[i][\phi][\psi]=(r_i,X^i_0)$.  Let
$\phi',\psi'$ be natural intersections of $\phi,\psi$ with $B_j$,
respectively, that is, $\phi'=(S_i\cap B_j, U_i\cap B_j)$ and
$\psi'=(L\cap B_j,X\cap B_j,R\cap B_j,x)$.  Let
$T_4[j][\phi'][\psi']=(r_j,X^j_0)$.  We consider some sub-cases,
depending on the alignment of $v$ in $\phi$ and $\psi$.  The cases
with $v$ belonging to $L$ and $R$ are symmetric, so we consider only
the case for $L$.

\vskip 0.1cm {\bf{Case 2.1: $v\in X$.}} Note that every extension
consistent with interface $\psi$ is an extension consistent with
$\psi'$ after trimming to $G_j$.  On the other hand, every extension
consistent with $\psi'$ can be extended to an extension consistent
with $\psi$ by adding $v$ to the extension of $X$.  Hence, it follows
that we can simply take $(r_i,X^i_0)=(r_j,X^j_0)$.

\vskip 0.1cm {\bf{Case 2.2: $v\in L$.}} Similarly as in the previous
case, every extension consistent with interface $\psi$ is an extension
consistent with $\psi'$ after trimming to $G_j$.  On the other hand,
if we are given an extension consistent with $\psi'$, then we can add
$v$ to $L$ and make an extension consistent with $\psi$ if and only if
$v$ is not adjacent to any vertex of $R$; this follows from the fact
that $B_j$ separates $v$ from $W_j$, so the only vertices from
$R^{\ext}$ that $v$ could be possibly adjacent to, lie in $B_j$.
However, if $v$ is adjacent to a vertex of $R$, then we can obviously
put $(r_i,X^i_0)=\bot$ as there is no extension consistent with
$\psi$: property that there is no edge between $L$ and $R$ is broken
already in the bag.  Otherwise, by the reasoning above we can put
$(r_i,X^i_0)=(r_j,X^j_0)$.

\vskip 0.1cm

{\bf{Case 2.3: $v\in B_i\setminus U_i$.}} Again, in this case we have one-to-one
correspondence of extensions consistent with $\psi$ and extensions
consistent with $\psi'$, so we may simply put $(r_i,X^i_0)=(r_j,X^j_0)$.

\vskip 0.3cm

\noindent{\bf{Case 3: Forget node.}} Let $i$ be a node that forgets
vertex $w$, and $j$ be its only child.  Consider some signature
$\phi=(S_i,U_i)$ of $B_i$, and some interface $\psi=(L,X,R,x)$; we
would like to compute $T_4[i][\phi][\psi]=(r_i,X^i_0)$.  Let
$\phi'=(S_j,U_j)$ be the only the extension of signature $\phi$ to
$B_j$ that has the same extension as $\phi$; $\phi'$ can be deduced by
looking up which signatures are found valid in table $C$ in the same
manner as in the forget step for computation of table $C$.  We
consider two cases depending on alignment of $w$ in $\phi'$:

\vskip 0.1cm {\bf{Case 3.1: $w\notin U_j$.}} If $w$ is not in $U_j$, then
it follows that we may put $(r_i,X^i_0)=T_4[j][\phi'][\psi']$:
extensions consistent with $\psi$ correspond one-to-one to extensions
consistent with $\psi'$.

\vskip 0.1cm {\bf{Case 3.2: $w\in U_j$.}} Assume that there exists some
extension $(L^{\ext}, X^{\ext}, R^{\ext})$ consistent with $\psi$, and
assume further that this extension is the one that maximizes
$|R^{\ext}\cap W_i|$.  In this extension, vertex $w$ is either in
$L^{\ext}$, $X^{\ext}$, or in $R^{\ext}$.  Let us define the
corresponding interfaces:
\begin{itemize}
\item $\psi_L=(L\cup \{w\},X,R,x)$;
\item $\psi_X=(L,X\cup \{w\},R,x-1)$;
\item $\psi_R=(L,X,R\cup \{w\},x)$.
\end{itemize}
If $x-1$ turns out to be negative, we do not consider $\psi_X$.  For
$t\in \{L,X,R\}$, let $(r_j,X^{j,t}_0)=T_4[j][\phi'][\psi_t]$.  It
follows that for at least one $\psi'\in \{\psi_L,\psi_X,\psi_R\}$
there must be an extension consistent with $\psi'$: it is just the
extension $(L^{\ext}, X^{\ext}, R^{\ext})$.  On the other hand, any
extension consistent with any of interfaces $\psi_L,\psi_X,\psi_R$ is
also consistent with $\psi$.  Hence, we may simply put
$r_i=\max(r_L,r_X,r_R+1)$, and define $X_0^i$ as the corresponding
$X_0^{j,t}$, with possibly $w$ appended if $t=X$.  Of course, in this
maximum we do not consider the interfaces $\psi_t$ for which
$T_4[j][\phi'][\psi_t]=\bot$, and if $T_4[j][\phi'][\psi_t]=\bot$ for
all $t\in \{L,X,R\}$, we put $(r_i,X^i_0)=\bot$.

\vskip 0.3cm

\noindent{\bf{Case 4: Join node.}} Let $i$ be a join node and
$j_1,j_2$ be its two children.  Consider some signature
$\phi=(S_i,U_i)$ of $B_i$, and an interface $\psi=(L,X,R,x)$; we would
like to compute $T_4[i][\phi][\psi]=(r_i,X_0^i)$.  Let
$\phi_1=(S_i,U_i)$ be a signature of $B_{j_1}$ and $\phi_2=(S_i,U_i)$
be a signature of $B_{j_2}$.  Assume that there is some extension
$(L^{\ext}, X^{\ext}, R^{\ext})$ consistent with $\psi$, and assume
further that this extension is the one that maximizes $|R^{\ext}\cap
W_i|$.  Define $r^p=|W_{j_p}\cap R|$ and $x^p=|W_{j_p}\cap X|$ for
$p=1,2$; note that $r^1+r^2=r_i$ and $x^1+x^2=x$.  It follows that in
$G_{j_1}$, $G_{j_2}$ there are some extensions consistent with
$(L,X,R,x^1)$ and $(L,X,R,x^2)$, respectively --- these are simply
extension $(L^{\ext}, X^{\ext}, R^{\ext})$ intersected with $V_i,V_j$,
respectively.  On the other hand, if we have some extensions in
$G_{j_1}$, $G_{j_2}$ consistent with $(L,X,R,x^1)$ and $(L,X,R,x^2)$
for numbers $x^p$ such that $x^1+x^2=x$, then the point-wise union of
these extensions is an extension consistent with $(L,X,R,x)$.  It
follows that in order to compute $(r_i,X^i_0)$, we need to iterate
through choices of $x^p$ such that we have non-$\bot$ entries in
$T_2[j_1][\phi_1][(L,X,R,x^1)]=(r^{x^1}_{j_1},X^{j_1,x^1}_0)$ and
$T_2[j_2][\phi_2][(L,X,R,x^2)]=(r^{x^1}_{j_1},X^{j_1,x^1}_0)$, choose
$x^1,x^2$ for which $r^{x^1}_{j_1}+r^{x^2}_{j_2}$ is maximum, and
define $(r_i,X_0^i)=(r^{x^1}_{j_1}+r^{x^2}_{j_2},X^{j_1,x^1}_0\cup
X^{j_2,x^2}_0)$.  Of course, if for no choice of $x^1,x^2$ it is
possible, we put $(r_i,X_0^i)=\bot$.  Note that computing the union of
the sets $X^{j_p,x^p}_0$ for $p=1,2$ takes $O(k)$ time as their sizes
are bounded by $k$, and there is $O(t)$ possible choices of $x^p$ to
check.

\vskip 0.3cm

Similarly as before, for every addition/removal of vertex $v$ to/from
$S$ or marking/unmarking $v$ as a pin, we can update table $T_4$ in
$O(5^t\cdot k^{O(1)}\cdot \log n)$ time by following the path from
$r_v$ to the root and recomputing the tables in the traversed nodes.
Also, $T_4$ can be initialized in $O(5^t\cdot k^{O(1)}\cdot n)$ time
by processing the tree decomposition in a bottom-up manner and
applying the formula for every node.  Note that updating/initializing
table $T_4$ must be performed after updating/initializing tables $P$
and $C$.

\paragraph{Implementing query \qUsep}
We now show how to combine Lemmata~\ref{lem:c4vc}
and~\ref{lem:tracing} with the construction of table $T_4$ to
implement query \qUsep.

The algorithm performs as follows.  First, using
Lemma~\ref{lem:tracing} we identify a node $i_0$ of the tree
decomposition, together with disjoint subsets
$(U_{i_0},S_{i_0})=(U\cap B_{i_0},S\cap B_{i_0})$ of $B_{i_0}$, such
that $ |U|/4\leq |W_{i_0}\cap U|\leq |U|/2$.  Let $A_2=W_{i_0}$ and
$A_1=V(G)\setminus V_{i_0}$.  Consider separation $(A_1\cap
U,B_{i_0}\cap U,A_2\cap U)$ of $G[U]$ and apply Lemma~\ref{lem:c4vc}
to it.  Let $(T^0_L,X^0_B,T^0_R)$ be the partition of $B_{i_0}$ and
$k^0_1,k^0_2$ be the integers with $k^0_1+k^0_2+|X^0_B|\leq k+1$, whose
existence is guaranteed by Lemma~\ref{lem:c4vc}.

The algorithm now iterates through all possible partitions
$(T_L,X_B,T_R)$ of $B_{i_0}$ and integers $k_1,k_2$ with
$k_1+k_2+|X_B|\leq k+1$.  We can clearly discard the partitions where
there is an edge between $T_L$ and $T_R$.  For a partition
$(T_L,X_B,T_R)$, let $G_1,G_2$ be defined as in Lemma~\ref{lem:c4vc}
for the graph $G[U]$.  For a considered tuple $(T_L,X_B,T_R,k_1,k_2)$,
we try to find:
\begin{itemize}
\item[(i)] a separator of a right-pushed separation of order $k_2$ in
  $G_2$, and the corresponding cardinality of the right side;
\item[(ii)] a separator of a left-pushed separation of order $k_1$ in
  $G_1$, and the corresponding cardinality of the left side.
\end{itemize}
Goal (i) can be achieved simply by reading entries
$T_4[i_0][(U_{i_0},S_{i_0})][(T_L,X_B,T_R,k')]$ for $k'\leq k_2$, and
taking the right-pushed separation with the largest right side.  We
are going to present how goal (ii) is achieved in the following
paragraphs, but firstly let us show that achieving both of the goals
is sufficient to answer the query.

Observe that if for some $(T_L,X_B,T_R)$ and $(k_1,k_2)$ we obtained
both of the separators, denote them $X_1,X_2$, together with
cardinalities of the corresponding sides, then using these
cardinalities and precomputed $|U|$ we may check whether $X_1\cup
X_2\cup X_B$ gives us a $\frac{8}{9}$-separation of $G[U]$.  On the
other hand, Lemma~\ref{lem:c4vc} asserts that when
$(T^0_L,X^0_B,T^0_R)$ and $(k^0_1,k^0_2)$ are considered, we will find
some pushed separations, and moreover any such two separations will
yield a $\frac{8}{9}$-separation of $G[U]$.  Note that this is indeed
the case as the sides of the obtained separation have cardinalities at
most $\frac{7}{8}|U|+\frac{(k+1)+(t+1)}{2}=\frac{8}{9}|U|+\frac{k+t+2}{2}-\frac{|U|}{72}\leq
\frac{8}{9}|U|$, since $|U|\geq 36(k+t+2)$.

We are left with implementing goal (ii).  Let $G_1'$ be $G_1$ with
terminal sets swapped; clearly, left-pushed separations in $G_1$
correspond to right-pushed separations in $G_1'$.  We implement
finding a right-pushed separations in $G_1'$ as follows.

Let $P=(i_0,i_1,\ldots,i_h=r)$ be the path from $i_0$ to the root $r$
of the tree decomposition.  The algorithm traverses the path $P$,
computing tables $D[i_t]$ for consecutive indexes $t=1,2,\ldots,t$.
The table $D[i_t]$ is indexed by signatures $\phi$ and interfaces $\psi$
in the same manner as $T_4$.  Formally, for a fixed signature
$\phi=(S_{i_t},U_{i_t})$ of $B_{i_t}$ with extension
$(S^{\ext}_{i_t},U^{\ext}_{i_t})$, we say that this signature is
{\emph{valid with respect to $(S_{i_0},U_{i_0})$}} if it is valid and
moreover $(S_{i_0},U_{i_0})=(S^{\ext}_{i_t}\cap
B_{i_0},U^{\ext}_{i_t}\cap B_{i_0})$.  For an interface $\psi$ we say
that separation $(L^{\ext}, X^{\ext}, R^{\ext})$ in
$G[U^{\ext}_i\setminus W_{i_0}]$ with terminals $L,R$ is
{\emph{consistent with $\psi$ with respect to $(T_L,X_B,T_R)$}}, if it
is consistent in the same sense as in table $T_4$, and moreover
$(T_L,X_B,T_R)=(L^{\ext}\cap B_{i_0}, X^{\ext}\cap B_{i_0},
R^{\ext}\cap B_{i_0})$.  Then entry $T[i_t][\phi][\psi]$ contains the
pair $(r,X_0)$ where $r$ is the maximum possible $|R^{\ext}\cap W_i|$
among extensions consistent with $\psi$ with respect to
$(T_L,X_B,T_R)$, and $X_0$ is the corresponding set $X^{\ext}\cap W_i$
for which this maximum was attained, or $\bot$ if the signature $\phi$
is invalid with respect to $(S_{i_0},U_{i_0})$ or no such consistent
extension exists.

The tables $D[i_t]$ can be computed by traversing the path $P$ using the
same recurrential formulas as for table $T_4$.  When computing the
next $D[i_t]$, we use table $D[i_{t-1}]$ computed in the previous step
and possible table $T_4$ from the second child of $i_t$.  Moreover, as
$D[i_0]$ we insert the dummy table $Dummy[\phi][\psi]$ defined as follows:
\begin{itemize}
\item $Dummy[(U_{i_0},S_{i_0})][(T_R,X_B,T_L,0)]=0$;
\item all the other entries are evaluated to $\bot$.
\end{itemize}
It is easy to observe that table $Dummy$ exactly satisfies the
definition of $D[i_0]$.  It is also straightforward to check that the
recurrential formulas used for computing $T_4$ can be used in the same
manner to compute tables $D[i_t]$ for $t=1,2,\ldots,h$.  The
definition of $D$ and the method of constructing it show, that the
values
$D[r][(\emptyset,\emptyset)][(\emptyset,\emptyset,\emptyset,x)]$ for
$x=0,1,\ldots,k$, correspond to exactly right-pushed separations with
separators of size exactly $x$ in the graph $G_1'$: insertion of the
dummy table removes $A_2$ from the graph and forces the separation to
respect the terminals in $B_{i_0}$.

Let us conclude with a summary of the running time of the query.
Algorithm of Lemma~\ref{lem:tracing} uses $O(t^{O(1)}\cdot \log n)$
time.  Then we iterate through at most $O(3^t\cdot k^2)$ tuples
$(T_L,X_B,T_R)$ and $(k_1,k_2)$, and for each of them we spend $O(k)$
time on achieving goal (i) and $O(5^t\cdot k^{O(1)}\cdot \log n)$ time
on achieving goal (ii).  Hence, in total the running time is
$O(15^t\cdot k^{O(1)}\cdot \log n)$.

\bibliographystyle{plain}

\end{document}